\documentclass[10pt]{article}
\usepackage{amsfonts}
\usepackage{amsmath, amsfonts, amsthm, amssymb, amscd, mathrsfs}
\usepackage[mathcal]{euscript}
\usepackage{color}
\usepackage {a4wide}
\usepackage{needspace,multirow,accents}
\theoremstyle{plain}
\newtheorem{theorem}{Theorem}[section]
\newtheorem{proposition}[theorem]{Proposition}

\newtheorem{lemma}[theorem]{Lemma}
\newtheorem{corollary}[theorem]{Corollary}
\newtheorem{definition}[theorem]{Definition}

\numberwithin{equation}{section}
\let\oldmarginpar\marginpar
\renewcommand\marginpar[1]{\- \oldmarginpar[\raggedleft\footnotesize #1]%
{\raggedright\footnotesize #1}}

\usepackage{a4wide}
\newcommand \loc {\text{loc}}
\definecolor{myred}{rgb}{0.858, 0.0, 0.0}

\newcommand \bse {\begin{subequations}}
\newcommand \bsel {\begin{subequations} \label}
\newcommand \ese {\end{subequations}}
\newcommand \gbf {\mathbf g}

\usepackage{lineno, blindtext}

\newcommand \bei {\begin{itemize}}
\newcommand \eei {\end{itemize}}
\newcommand \fb {\overline f}

\newcommand \be         {\begin{equation}}
\newcommand \bel {\be\label}

\newcommand \eps \epsilon

\newcommand \coeff \kappa

\newcommand \Ocal {\mathcal O}

\newcommand \hb {\overline h}

\newcommand \gb {\overline g}

\newcommand \wb {\overline w}

\newcommand \del \partial

\newcommand \Ric {\mathbf{Ric}}

\newcommand \Acal   {\mathcal A}

\newcommand \RR         {\mathbb R}
\newcommand \ee         {\end{equation}}
\newcommand \la \langle
\newcommand \ra \rangle

\newcommand \Hcal   {\mathcal H}
\newcommand \Kcal   {\mathcal K}

\usepackage[most]{tcolorbox} 
\newtcolorbox{done}[1][]{%
     enhanced, breakable, size=minimal, parbox=false, after={\mbox{} \par}, 
     before upper={\needspace{7\baselineskip} \indent}, colback=white, 
     overlay = {\draw[line width =2pt] (frame.north east) -|
                       ([xshift=3mm]frame.east)|-(frame.south east);},
     overlay first={\draw[line width =2pt] (frame.north east) -|
                           ([xshift=3mm]frame.south east);},
     overlay middle={\draw[line width =2pt] ([xshift=3mm]frame.north east) -- 
                              ([xshift=3mm]frame.south east);},
     overlay last={\draw[line width =2pt] ([xshift=3mm]frame.north east)|-
                          (frame.south east);},
     #1
}
 
\usepackage{perpage}
\MakePerPage{footnote} 
\begin{document}

\title{Global evolution in spherical symmetry for
\\
 self-gravitating massive fields 
} 

\author{Philippe G. LeFloch\footnote{
 Laboratoire Jacques-Louis Lions \& Centre National de la Recherche Scientifique,
Sorbonne Universit\'e, 4 Place Jussieu, 75252 Paris, France.
Email: {\sl contact@philippelefloch.org}
\newline $^\dag$ Centre for Mathematical Analysis, Geometry and Dynamical Systems, Instituto Superior T\'ecnico,  Universidade de Lisboa, Av. Rovisco Pais, 1049–001 Lisbon, Portugal. Email: {\sl filipecmena@tecnico.ulisboa.pt}
\newline 
$^\flat$ Centro de Matem\'atica, Universidade do Minho, Campus de Gualtar, 4710--057 Braga, Portugal.
\newline $^\ddagger$ Institut Montpelli\'erain A. Grothendieck, Universit\'e de Montpelier, Place E. Bataillon, 34090 Montpellier, France.
} \hskip.01cm, 
Filipe C. Mena$^{\dag\flat}$, and The-Cang Nguyen$^\ddagger$}


\date{December 2022}

\maketitle

\begin{abstract} 
We are interested in the global dynamics of a massive scalar field evolving under its own gravitational field and, in this paper, we study spherically symmetric solutions to Einstein's field equations coupled with a Klein-Gordon equation with quadratic potential. For the initial value problem we establish a global existence theory when initial data are prescribed on a future light cone with vertex at the center of symmetry. A suitably generalized solution in Bondi coordinates is sought
which has low regularity and possibly large but finite Bondi mass. A similar result was established first by Christodoulou for massless fields.
In order to deal with massive fields, we must overcome several challenges and significantly modify Christodoulou's original method. First of all, we formulate the Einstein-Klein-Gordon system in spherical symmetry as a non-local and nonlinear hyperbolic equation and, by carefully investigating the global dynamical behavior of the massive field, we establish various estimates concerning the Einstein operator, the Hawking mass, and the Bondi mass, including positivity and monotonicity properties. Importantly, in addition to a regularization at the center of symmetry we find it necessary to also introduce a regularization at null infinity. We also establish new energy and decay estimates for, both, regularized and generalized solutions.
\end{abstract}


\setcounter{secnumdepth}{3}
\setcounter{tocdepth}{1} 
\tableofcontents
 
 \
 
 \newpage  

\section{Introduction}

\subsection{Global evolution of self-gravitating massive matter}

\paragraph{Einstein-Klein-Gordon system.}

We initiate the mathematical study of the global evolution problem for self-gravitating massive scalar fields in spherical symmetry and, specifically, we generalize here an existence result established by Christodoulou for massless fields within a pioneering work beginning with the papers \cite{Chr,Chr2}. Our objective is to extend Christodoulou's method in \cite{Chr2}, which does not immediately apply to massive fields and must be significantly revisited. The main unknown of Einstein's field equations is a spacetime $(M,\gbf)$ 
consisting of a four-dimensional manifold $M$ endowed with a Lorentzian metric $\gbf$. In addition, the matter content is assumed here to be described by a massive scalar field $\phi: M \mapsto \RR$ with a potential $U=U(\phi)$, which evolves under its own gravitational field, according to the Klein-Gordon equation
\bse
\label{equa-KGequa} 
\bel{eq:KG1} 
\Box_{\gbf} \phi - U'(\phi) = 0. 
\ee 
Here, $\Box_{\gbf}$ denotes the wave operator associated with the unknown metric $\gbf$ and, for definiteness,  we assume a quadratic potential (with $c> 0$) 
\be
U(\phi)=\frac{c^2}{2} \phi^2. 
\ee
\ese 
Quadratic potentials are ``critical'' from the standpoint of global dynamics and decay, and commonly used in physical and numerical studies~\cite{Arbey,Bahamonde, Boson, Okawa, Urena}. 

Namely, we seek solutions $(\gbf, \phi)$ to the Einstein equations
\bel{Eq1-01}
\Ric - {R \over 2} \, \gbf = 8\pi \, T[\phi,\gbf], 
\ee
where $\Ric$ denotes the Ricci tensor and $R$ the scalar curvature of the metric. 
The stress-energy tensor $T[\phi,\gbf]$ reads  
\bel{stress} 
T[\phi,\gbf]=d\phi\otimes d\phi -{1 \over 2} \bigl(  \gbf( \nabla\phi, \nabla\phi) +c^2\phi^2 \bigr) \, \gbf. 
\ee
The Klein-Gordon equation coupled to the Einstein equations can be cast in the form of a nonlinear and coupled system of wave and Klein-Gordon equations. In comparison to earlier mathematical work, one of the major challenges that we need to overcome is that \eqref{eq:KG1} is {\sl not invariant under scaling.} Recall that the global evolution for massive fields in the near-Minkowski regime was recently investigated by Ionescu and Pausader \cite{IonescuPausader} and by LeFloch and Ma \cite{LeFloch-Ma-1,LeFloch-Ma-2,LeFloch-Ma-3}. Their result of nonlinear stability do not require spherical symmetry and cover small perturbations of the Minkowski spacetime and is formulated in suitably weighted Sobolev norms. 


\paragraph{Global evolution in spherical symmetry.}

On the other hand, in order to investigate the global evolution problem for ``arbitrary large'' initial data, we find it necessary to focus on the class of spherically symmetric solutions and to introduce a notion of generalized solutions with low regularity. 

The spacetimes $(M,\gbf)$ under consideration here are assume to enjoy spherical symmetry, in the sense that the rotation group $SO(3)$ acts as an isometry group on the spacetime and that the set of fixed-points of this group is a timelike world-line while the orbits of all other spacetime points are spacelike topological two-spheres $\mathbb{S}^2$. In order to uncover the global geometry of such spacetimes, we formulate the initial value problem associated with the Einstein-Klein-Gordon system when data are imposed on a future light cone with vertex at the origin. We follow the strategy proposed in~\cite{Chr,Chr2} and we thus introduce {\sl generalized Bondi coordinates,} which are based on a null coordinate and a retarded time. 
An initial data set is imposed on a complete future null geodesic cone extending to infinity, and is assumed to be asymptotically flat and have finite Bondi mass.  
On the one hand, for {\sl massless fields} enjoying suitable smallness conditions, Christodoulou~\cite{Chr} proved the global existence of {\sl classical solutions} (enjoying sufficient regularity so that all of the terms in the Einstein equations make sense as continuous functions, at least) and established that such solutions disperse in the infinite future. 
On the other hand, Christodoulou~\cite{Chr2} considered ``large'' initial data and established the global existence of {\sl generalized solutions,} defined in Bondi coordinates in a weak sense and covering a complete domain of outer communications.   


\paragraph{Einstein equations in Bondi coordinates.}

We introduce global functions $(u,r)$ on $M$, in which $u \geq 0$ represents a (future) null coordinate and $r \geq 0$ the (square-root of the) area of symmetry of each orbit $\mathbb{S}^2$. In Bondi coordinates the metric reads 
\bel{eq:metricrad}
\gbf= -e^{2\nu} \, du^2 - 2e^{\nu+ \lambda} \, dudr + r^2 \,\gbf_{S^2},
\ee
where the functions $\nu= \nu(u,r)$ and $\lambda= \lambda(u,r)$ depend only on the variables $(u,r)$, and $\gbf_{S^2}$ denotes the canonical metric on the two-sphere. 
From \eqref{Eq1-01} --\eqref{eq:metricrad} it is straightforward to derive the system satisfied by the metric coefficients $\nu,\lambda$ and the scalar field $\phi$, namely 
\bel{Bondi system diff}
\aligned
\del_r(\nu+ \lambda) & = 4\pi r  (\del_r \phi )^2,
\qquad
&&
\del_r\big(re^{\nu- \lambda} \big) 
 =  \big(1 - 4\pi r^{2}c^2 \phi^2  \big) e^{\nu+ \lambda},
\\
D(\del_r(r\phi))
& =  \frac{1}{2} \big(
1 - e^{-2\lambda} - 4r^2 \, c^2 \pi \phi^2
\big) e^{\nu+ \lambda} \del_r\phi- \frac{c^2}{2} r e^{\nu+ \lambda} \phi,
\endaligned
\ee
in which $D$ is the null derivative operator 
\bel{equa-defD} 
D :=  \del_u- \frac{1}{2}e^{\nu- \lambda} \del_r.
\ee
 

We can rewrite the Einstein-Klein-Gordon system in a first-order form, as follows. Let us define the mean value of a function $f=f(u,r)$ by 
\be
\label{mean-value-operator}
\overline f(u,r) := \frac{1}{r} \int_0^r f(u,s) \, ds.
\ee
In terms of the new unknown  $h = \del_r ( r \, \phi)$ or, equivalently,  $\hb := \phi$, we can rewrite \eqref{Bondi system diff} as the {\sl reduced Einstein equation}
\bsel{self-gravitating eq}
\be
\aligned
& D h = \frac{1}{2r}(g- \gb)(h - \hb) - \frac{c^2}{2}r \, \hb \, e^{\nu+ \lambda},%
\endaligned
\ee
which is a first-order, integro-differential equation with unknown $h$ and the coefficients given as suitable integrals of $h$: 
\bel{fund-eqs}
\aligned
\nu+ \lambda& =  -4\pi \, \int_r^{+ \infty} \frac{(h - \hb)^2}{s} \, ds,
\qquad 
g := \big(1-4c^2\pi r^2 \hb^2\big) e^{\nu+ \lambda}. 
\endaligned
\ee 
\ese
Hence, by definition $\gb = e^{\nu- \lambda}$ and we also have the alternative expressions
$h =  \del_r(r \, \hb)$ and $\del_r \hb = \frac{h - \hb}{r}$. 


\paragraph{Notions of mass.}

The {\sl Hawking mass,} available in spherical symmetry, is the function 
\bel{Hawk-mass}
m:= \frac{r}{2} \bigg(1- \frac{\gb}{e^{\nu+ \lambda}} \bigg), 
\ee
and the Bondi mass $M=M(u)$ is the limit (which will be shown to exist) 
\be
M (u) :=\lim_{r \to + \infty}  m(u,r).
\ee
These functions enjoy important monotonicity properties which will be established in the course of this paper and 
play an important role in the analysis of spherically symmetric solutions to the Einstein equations, as first observed in~\cite{Chr2}.
A main difficulty in the mathematical analysis comes from the fact that the unknown metric is not directly
controlled in the interior of the sphere of radius $2M(u)$. For later use, let us also point out the identity 
\bel{eoverbarg}
\frac{e^{\nu+ \lambda}}{\gb} = \left (1- \frac{2m}{r} \right )^{-1}.
\ee
  

\subsection{Global evolution in a class of generalized solutions}

\label{sec-IVP}

\paragraph{Notion of generalized solutions.} 

The global evolution problem under consideration may leads to {gravitational collapse} and this motivates one to search for solutions with singularities.
We thus consider the initial value problem associated with the reduced equation~\eqref{self-gravitating eq} when an initial data $h_0$ is prescribed on the null cone labelled $u=0$, that is, 
\bel{equa-initial}
h\big|_{u=0} = h_0. 
\ee
The spacetime domain is denoted  by 
\be
\Ocal := \big\{ (u,r) \, / \, 0\leq u< + \infty,~0<r< + \infty \big\},
\ee 
while, by definition, its closure $\overline \Ocal$ also includes the center line $r=0$. 
For each $(u_1, r_1) \in \Ocal$, let us introduce the characteristic curve $u \mapsto \chi(u; u_1,r_1)$ 
passing through $(u_1,r_1)$ which is the solution of the ordinary differential equation
\be
\aligned
& \del_u \chi 
= -  \frac{1}{2} \gb,
  \qquad
\chi(u_1; u_1, r_1) =r_1,  \qquad 
 u \in [0, u_1]. 
\endaligned
\ee 
It is also convenient to introduce the domain limited by this characteristic
\bel{equa-Ocalu1r1} 
\Ocal(u_1,r_1) := \big\{ (u,r) \, / \, 0 \leq u<u_1,\, \, 0< r <\chi(u; u_1,r_1)\big\}.
\ee
Recall that from the single function $h$ we have defined $\hb$ as well as $\nu, \lambda, g, \gb$, and $m$.  

\begin{definition}
\label{def-general-sol} 
A {\bf global generalized solution} is a continuously differentiable function $h\in C^1(\Ocal)$ satisfying the evolution equation \eqref{self-gravitating eq} within the spacetime domain $\Ocal$ and the initial condition \eqref{equa-initial}, together with the following properties. 
\begin{enumerate}

\item[1.] {\bf Integrability and regularity conditions.}  The matter field enjoys the integrability conditions 
\bse
\be
h\in L^2(\Ocal), 
\qquad 
\displaystyle{\int_0^{+ \infty} \Big(h^2 + \hb^2 + c^2 r^2\hb^2 \Big) \, dr\in L_\loc^\infty([0, + \infty))}, 
\ee 
while the metric coefficients satisfy  
\be
e^{\nu+ \lambda} + \gb\in C^0(\Ocal), 
\qquad
(e^{\nu+ \lambda}/\gb)(u, \cdot) \in L_\loc^1([0,+ \infty)) \text{ for all } u \geq 0.
\ee
In addition, the following limit exists
\be
\xi:= \lim_{\delta\to 0} \Big(\int^r_\delta \gb \, \frac{h - \hb}{r} \, dr - c^2\int^r_\delta r \, \hb \, e^{\nu+ \lambda} \, dr\Big)
\ee
(which represents the variations of the matter along null directions)
and enjoys the integrability condition
\be
\frac{e^{\nu+ \lambda}}{\gb^2} \,  {\xi^2 \over r} \in L_\loc^1(\overline\Ocal).
\ee
\ese

\item[2.] {\bf Two evolution equations.} The matter field $\hb$ is weakly differentiable in $\Ocal$ and evolves according to 
\bse
\be
D\hb= \frac{\xi}{2r}, 
\ee
while the Hawking mass $m$ (defined in \eqref{Hawk-mass}) is weakly differentiable in $\Ocal$ and evolves according to 
\be
\label{mass-evol-eq}
D m = -  e^{-(\nu+ \lambda)} \pi \xi^2 - \pi c^2  r^2\hb^2\gb.
\ee
\ese

\item[3.] {\bf Integral identity.} One also requires that, for all $(u_1, r_1) \in \Ocal$, 
\bel{integral identity 0}
\aligned
& \int_0^{r_1} \Big(\frac{e^{\nu+ \lambda}}{\gb} \Big)(u_1,r) \, dr
+ 2\pi \, \iint_{\Ocal(u_1,r_1)} \frac{e^{\nu+ \lambda}}{\gb^2} \frac{\xi^2}{r} \, dudr
 + \frac{1}{2} \int_0^{u_1}e^{(\nu+ \lambda)(u,0)} \, du
  = \int_0^{r_{0}} \Big(\frac{e^{\nu+ \lambda}}{\gb} \Big)(0,r) \, dr.
\endaligned
\ee
\end{enumerate}
\end{definition}

 The equation \eqref{integral identity 0} is derived from the integration of the expression for $D(\int_0^r (1-2m/r)^{-1} \, dr)$ along characteristic curves (cf.~\eqref{eoverbarg}). Since the point $(u_1, r_1) \in \Ocal$ is arbitrary, it provides us with a control of the geometry in the full domain of outer communication, namely the causal past of future null infinity. As a result of the integration along characteristics, the identity \eqref{integral identity 0} relates an integral over the initial future light cone $u=0$ (on the right-hand side) to the sum of three positive integral terms: 
 (1) the first integral (over the future light cone $u=u_1$) concerns the interior of the sphere of radius $2M(u_1)$ (which can not be controlled by the Bondi mass, as pointed out in \cite{Chr2}), 
 (2) the second integral is a spacetime integral and will be crucial to prove that generalized solutions satisfy the mass equation \eqref{mass-evol-eq}, 
 and (3) the third  integral is a boundary contribution over the central line $r=0$ (which will be analyzed within our proof).
 

\paragraph{Main statement.}

Our main result can be interpreted as a global nonlinear stability property of the exterior domain of communications of Schwarzschild spacetime, under under arbitrarily large perturbations by a massive matter field. 

\begin{theorem}[Global evolution of self-gravitating Klein-Gordon fields in spherical symmetry] 
\label{thrgensol}
For every initial data $h_0\in C^1([0, + \infty))$ with finite Bondi mass posed on an outgoing light cone with vertex at the center of symmetry, 
there exists a global generalized solution $h$ to the Cauchy problem \eqref{self-gravitating eq} and \eqref{equa-initial}
 describing the spherically symmetric evolution of a self-gravitating Klein-Gordon field with quadratic potential.
\end{theorem}

We emphasize that it is beyond the scope of the present paper to investigate qualitative properties of the generalized solutions beyond the ones stated in Definition~\ref{def-general-sol}, and this issue is postponed for future work. 


\paragraph{Related results.} 

Soon after Christodoulou's work, Goldwirth and Piran \cite{Goldwirth-Piran} also considered massless scalar fields
for a class of Gaussian initial data, and numerically investigated the formation of black holes. 
Choptuik \cite{Choptuik} numerically discovered a (self-similar) critical phenomena in the collapse of massless scalar fields, while
Christodoulou \cite{Chr-naked} proved rigorously that {\sl naked singularities} (with singularities ``visible'' from null infinity) 
 can form as result of the collapse of a scalar field, but however are not generic. 
 
These theoretical and numerical advances on massless fields next inspired more numerical works, including a numerical investigations of
massive fields. Brady et al. \cite{Brady, Brady-etal} studied massive fields governed by a quadratic potential for a class of initial data sets with exponential decay. A class of quartic potentials was also studied numerically by Honda~\cite{Honda}, who studied the stability of the collapse of oscillating solutions. 
Along similar lines, Zhang and Lu \cite{Zhang-Lu} considered the spherically-symmetric collapse of nonlinear massless scalar field with self-interactions
and first treated hyperbolic potentials, namely, exponentially decaying functions with powers of odd exponents.  More recently, Michel and Moss \cite{Michel-Moss} numerically studied trigonometric potentials, and emphasized the role of 
the Bondi mass in defining the final black hole mass: its variation is a measure of the energy flux emanating from the collapsing star.

As far as the mathematical analysis is concerned, only few works took over Christodoulou's framework and extended his results. Chae \cite{Chae} proved the global existence of classical solutions with small data for the Einstein-massive field system with quartic or quintic potentials; he treated
 initial data with polynomial decay and relied on a fixed-point method adapted to the potentials under consideration. 
 Extensions of this result to the Maxwell, Yang-Mills, and Higgs fields were given in~\cite{Chae2}.
More recently, Costa et al.~\cite{Costa-etal,Costa-Mena} considered the Einstein-massless field system with a positive cosmological 
constant and established the global existence of classical solutions for a class of initial data sets.


\paragraph{Outline of this paper.}

In Section~\ref{sec-section2} we present a regularization of the Einstein-Klein-Gordon system, which includes both a regularization at the center and a regularization at null infinity, and we discuss some feature of our method of proof. In Section~\ref{sec-section3} we establish the existence of global classical solutions for the $(\eps,b)$-regularized problem, and prove suitable Lipschitz continuity and sign properties. 
Geometric estimates on the Hawking mass and the energy of the Klein-Gordon fields are then derived in Section~\ref{sec-section4}. The existence of global classical solutions to the $\eps$-regularized problem is the, completed in Section~\ref{sec-section5}
In Section~\ref{sec-section6} we are finally in a position to derive several global estimates enjoyed by generalized solutions. 
Finally, the mass equation plays a key role in the completion of the proof in Section~\ref{sec-section7}. 


\section{A formulation based on two regularizations}
\label{sec-section2}
 
\subsection{Regularization at the center} 

We proceed by constructing first smooth solutions to a regularized version of the reduced Einstein equation. We observe that the fixed-point method used by Christodoulou does not work for massive fields (a fact that was pointed out in \cite{Chae}), 
and we thus propose a different strategy which is based on the introduction of a continuous and compact operator. 
Following Christodoulou \cite{Chr2} we modify the mean value operator \eqref{mean-value-operator} and, given a parameter $\eps \in (0,1)$, we set  
\be
\fb_\eps(u,r) := \frac{1}{r+ \eps} \int_0^rf_\eps(u,s) \, ds
\ee 
for any integrable function $f_\eps: [0, +\infty) \mapsto \RR$. 
 
We are going to study first the {\bf $\eps$-regularized reduced Einstein equation}  
\bse\label{epsilon self-gravitating eq}
\bel{epsilon self-gravitating eq-0}
\aligned
 D_\eps h_\eps= \frac{1}{2(r+ \eps)}(g_\eps- \gb_\eps)(h_\eps- \hb_\eps) - \frac{c^2}{2}(r+ \eps) e^{\nu_\eps + \lambda_\eps} \hb_\eps,
\qquad\quad
 h_\eps\big|_{u=0} = h_0,
\endaligned
\ee
in which $h_0\in C^1([0,+ \infty)$ is prescribed on the initial cone $u=0$ and 
\be
\aligned
\label{g-equation}
D_\eps &:=  \del_u- \frac{1}{2} \gb_\eps \del_r,
\qquad
&& 
g_\eps :=  \big(1-4c^2\pi (r+ \eps)^2 \hb_\eps^2\big) e^{\nu_\eps + \lambda_\eps},
\\
\gb_\eps & =  e^{\nu_\eps- \lambda_\eps},
&&
\nu_\eps + \lambda_\eps 
=  -4\pi \, \int_r^{+ \infty} \frac{(h_\eps- \hb_\eps)^2}{s+ \eps} \, ds.
\endaligned
\ee
\ese 
We also introduce the $\eps$-Hawking mass 
\bel{bondi-mass-eps}
m_\eps := \frac{r+ \eps}{2} \bigg(1- \frac{\gb_\eps}{e^{\nu_\eps + \lambda_\eps}} \bigg) 
\ee
and the $\eps$-Bondi mass 
\be  
M_\eps (u) :=\lim_{r \to + \infty}  m_\eps(u,r).
\ee

For the regularized problem \eqref{epsilon self-gravitating eq} we will establish estimates that are valid up to the central line $r=0$. 
Our estimates for the regularized problem involve the following ``norms'' determined by the initial data $h_0$: 
\bel{def A}
\aligned
\Acal_\eps
(u,r)
&:= \sup_{r' \in [r, +\infty)} \big| \int_0^r h_0 \, ds \big|
+ \sqrt{\frac{u \,  M_\eps(0)
}{2\pi}},
\\ 
\Hcal_\eps
(u)
&:= 
%
{33 \over 4} 
\int_0^u \max\Big( 4\pi (\Acal_\eps
(u_1,0))^4,
\, 
e^{-2\sqrt{2} \pi^{1/4}} \Big) \, du_1,
\\
\Kcal_\eps
(u)
&:=\frac{17}{\pi} \int_0^u \max\big( {c^2e^{-2},}
M_\eps(u_1)
\big) \, du_1.
\endaligned
\ee
More precisely, the latter function depends upon the Bondi mass $M_\eps(u_1)$ which can be bounded above by the initial Bondi mass $M_\eps(0)$.  
We will establish  the following result, in which a subscript $0$ is used for functions determined from the initial data $h_0$, while the subscript $\eps$ indicates that the regularized mean is applied. 

\begin{theorem}[Existence theory for the $\eps$-regularized reduced equation]
\label{lemma-global epsilon existence}
Assume that the initial $\eps$-Bondi mass $M_{\eps,0}$ is finite and $\gb_{\eps,0}> 0$ in the interval $(0, +\infty)$. For any $\eps \in (0,1)$, the $\eps$-regularized  problem \eqref{epsilon self-gravitating eq} admits a global classical solution $h_\eps$ which satisfies, in the spacetime domain $\Ocal$, 
\bel{equa-1234}  
(1) 
\quad \gb_\eps> 0, 
  \qquad \qquad 
(2) 
\quad 
(r+ \eps) \, |\hb_\eps | \leq \Acal_\eps,
\ee
while for all times $u \geq 0$, 
\be
\label{energy-int}
\aligned
& (3) 
\quad 
&
\int_0^{+ \infty}h^2_\eps (u, \cdot) \, dr
&&&\leq 
\int_{0}^{+ \infty}h_0^2 \, dr + \frac{u}{8\pi}
+c^2 \Hcal_\eps(u), 
\\
& (4) 
\quad 
&
\int_0^{+ \infty}(h_\eps- \hb_\eps)^2(u, \cdot) \, dr 
&&&\leq \int_0^{+ \infty}(h_0- \hb_0)^2 \, dr + \frac{u}{4\pi} +c^2 \Hcal_\eps(u)+ \Kcal_\eps(u), 
\\ 
& (5) 
\quad 
&
c^2\int_0^{+ \infty}(r+ \eps)^2\hb_\eps^2(u, \cdot) \, dr 
&&& \leq {M_\eps(0) \over 2 \pi} 
\endaligned
\ee
\end{theorem}

In agreement with our earlier notation, we also set 
\bel{definition of xi}
\xi_\eps:=2(r+ \eps)D_\eps \hb_\eps= \int_0^r\gb_\eps \frac{(h_\eps- \hb_\eps)}{s+ \eps} \, ds
- c^2\int_0^r(s+ \eps) \, \hb_\eps \, e^{\nu_\eps + \lambda_\eps} \, ds. 
\ee
For each $(u_1, r_1) \in \Ocal$, we have denoted by $u \mapsto \chi_\eps(u; u_1, r_1)$ the characteristic curve satisfying 
\be
\aligned
& \del_u \chi_\eps = -  \frac{1}{2} \gb_\eps,
\qquad 
\chi_\eps(u_1;u_1,r_1) =r_1,  \qquad 
 u \in [0, u_1]. 
\endaligned
\ee
In addition to  the spacetime domain \eqref{equa-Ocalu1r1}, we also work with 
\be
\Ocal_{\eps,\delta}(u_1,r_1)
:= \big\{ (u,r) \, / \, 0<u<u_1,\,\delta<r<\chi_\eps(u; u_1, r_1) \big\}, 
\ee 
in which the parameter $\delta\in (0,1)$ will be taken to tend to zero. %

\begin{corollary}[Evolution of the $\eps$-Hawking mass] 
\label{coroll-thr-21}
Under the conditions in Theorem~\ref{lemma-global epsilon existence}, the solution $h_\eps$ also satisfies the following properties.
\bei
\item[(6)] 
The $\eps$-Hawking mass satisfies the following evolution equation in $\Ocal$: 
\bel{evol-eps-mass}
D_\eps m_\eps= -  \frac{4\pi}{e^{\nu_\eps + \lambda_\eps}}(r+ \eps)^2(D_\eps \hb_\eps)^2- c^2\pi(r+ \eps)^2\hb_\eps^2\gb_\eps.
\ee

\item[(7)] For each $\delta > 0$ and any point $(u_1, r_1)$ with $r_1> \delta$ the following integral identity   
holds (with $r_{0,\eps} := \chi_\eps(0; u_1, r_1)$)
\bel{integral identity}
\hskip-.2cm
\int_{\delta}^{r_1} \Big(\frac{e^{\nu_\eps + \lambda_\eps}}{\gb_\eps} \Big)(u_1,r) \, dr + 2\pi \, \iint_{\Ocal_{\eps,\delta}(u_1,r_1)} \Big(\frac{e^{\nu_\eps + \lambda_\eps}}{\gb_\eps^2} \frac{\xi^2_\eps}{(r+ \eps)} \Big) \, drdu
 + \frac{1}{2} \int_0^{u_1} \hskip-.2cm
 e^{(\nu_\eps + \lambda_\eps)(u,\delta)} \, du
   = \int_{\delta}^{r_{0,\eps}} \Big(\frac{e^{\nu_\eps + \lambda_\eps}}{\gb_\eps} \Big)(0,r) \, dr.
\ee 
 
\eei
\end{corollary}


\subsection{Regularization at null infinity} 
\label{sec-formulation-fixed-point}

In order to establish Theorem~\ref{lemma-global epsilon existence}, another regularization will be also needed, in which 
we disregard first the initial matter field  {\sl beyond} some large radius (cf.~Section~\ref{sec-formulation-fixed-point}) and we 
work in the interval $[0, b+2]$ (for some positive $b$).
 Thanks to a sign condition enjoyed the characteristic curves (which follows from the attractive nature of gravity)
 we can check that the regularized solutions are then also supported in the same interval. 
Several smaller intervals will arise in our analysis, including $[0,b]$ and, in our presentation, it is convenient to assume $b \geq 1$.  
We then split the spacetime domain into slab of finite proper distance, namely we study the solution $h_{\eps, b} = h_{\eps, b}(u,r)$ successively in 
\be 
(u,r) \in [n,n+c_1]\times[0,b+ 2]
\ee
for integers $n=0, 1, \ldots$. Indeed, we will prove a local existence result covering any slab of the form $[n, n+c_1]$ for arbitrarily large initial data and for a {\sl sufficiently small }$c_1> 0$: more precisely $1/c_1$ 
will be chosen to be of the order of the Klein-Gordon mass $c$. By induction, we can apply this local existence result in order to cover the whole time line $u \geq 0$. 

Importantly, we will guarantee that the characteristics propagate inwards, so that the support of the solution $h_{\eps, b}$ will remain within the spatial interval $[0,b+ 2]$ for all times. For his property it will be crucial to prove that 
\be
\overline g_{\eps, b} > 0. 
\ee
Furthermore, in the following analysis, we will also introduce spacetime domains based on the smaller intervals $[0,b]$ and $[0,b+1]$, since some of our arguments will provide the existence of a solution defined in $[0,b+2]$ but enjoying the desired properties only in slightly smaller intervals. (Recall that the parameter $b$ will go to $+\infty$ eventually.)  

  
\subsection{Challenges overcome in this paper}

Our main focus in this paper is on the definition and existence of generalized solutions to the Einstein equations in spherical symmetry.  
For the large data problem under consideration, a singularity may arise at the center of symmetry and it is thus natural to introduce first 
a regularization at the center of symmetry, based on replacing $r$ by $r+ \eps$ in the equations and then letting the 
parameter $\eps \in (0,1)$ tend to zero eventually. A main part of the analysis will thus be devoted to establish the existence of regularized solutions and to prove that the enjoy the desired properties {\sl uniformly} with respect to $\eps \to 0$. 

Another main difficulty is proving that the generalized solutions satisfy not only the evolution equation,
 but also the mass equation \eqref{mass-evol-eq}, which is (equivalent to) one of the Einstein equations and, therefore, must be validated also at the level of generalized solutions.  A related major difficulty comes from the fact that the mass does not control 
the geometry in the interior of the spheres with radius $2m(r)$, so that 
specific arguments must be developed in order to control the solution in the region near the center of symmetry.

As a matter of fact, several properties and arguments in Christodoulou's proof~\cite{Chr, Chr2} do not carry over from the massless to the massive matter problem.
 
\bei

\item {\bf Positivity and finiteness properties of the mass.} In the massless case $c=0$, the positivity of the function $\gb$ and the finiteness of the Bondi mass are immediate; moreover, they impose no further constraint on the sequence of $\eps$-regularized solutions and, in turn, the evolution equation for the Hawking mass is {\sl not directly required} in order to deal with terms involving the metric coefficient $\gb$.  In contrast, for the massive fields having $c \neq 0$, the positivity property $\gb> 0$ 
(as is also done for the finiteness of the Bondi mass) does not follow from the evolution equations but {\sl 
must be propagated} from assumptions made on the initial hypersurface. 

\item {\bf A consistent construction scheme.} For massless fields ($c=0$), the functions $Dh$ and $D(\del_r h)$ have the same decay rate as the one of the functions $h$ and $\del_r h$, respectively. For these functions, the fall-off at null infinity is preserved during the evolution and  this observation provides Christodoulou a basis for constructing a sequence of functions converging to the desired solution $h$ and, in turn, these features made the construction relatively direct.  
 On the other hand, for massive fields ($c \ne 0$), the scaling invariance is no longer available and 
the functions $Dh$ and $D(\del_r h)$ might decay at slower rates than $h$ and $\del_r h$, respectively. 
Hence, one of the challenges is to guarantee, {\sl simultaneously} in the process of showing the existence of (regularized or exact) solutions, 
(1)  a consistent fall-off for the unknown function $h$, 
 (2) the positiveness of the coefficient $\gb$, and (3) the finiteness of the Bondi mass. 

\item {\bf Cut-off method to cope with null infinity.}
 For massless fields, Christodoulou proved that $h$ and $Dh$ are $C^1$-regular outside the central axis $r=0$. On the other hand, 
for massive fields, the function $Dh$ is not uniformly bounded in the $C^1$ norm in the vicinity of null infinity.
This leads to an additional difficulty in the construction of a sequence of regularized solutions and the convergence analysis toward an actual solution. By taking advantage of the attractive nature of gravity and there a ``favorable'' sign on the characteristics of the reduced Einsteni equation, we introduce here a cut-off technique: we solve the equations first on a compact spacetime region and then rely on the finiteness of the Bondi mass in order to extend this solution. While Christodoulou deals with regularized solutions defined for all $r$ and a compact interval in $u$, we construct regularized solutions defined on a {\sl bounded} space interval $[\delta, b]$ first, before letting $\delta \to 0$ and $b \to + \infty$. 

\item {\bf Fixed-point argument based on the Schauder theorem.}
Our method of proof then requires the construction of a solution operator which is proven to be continuous and compact and, next, the application pf Schauder's fixed-point theorem. A major advantage of the proposed strategy is that, throughout our analysis for, both, the approximate and the exact solutions, we may employ {\sl properties derived from the evolution equations,} including the evolution equations for the Hawking mass. Throughout our analysis, the finiteness of the Bondi mass plays a crucial role. 

\item {\bf Convergence toward generalized solutions.}
Finally, for massless fields, Christodoulou established sup-norm estimates on the solution $h_\eps$ to the regularized problem 
before proving that it converges uniformly to a solution $h$ of the original problem. On the other hand, for massive fields 
such sup-norm estimates are not available and, in order to prove that $h_\eps \to h$, we develop here an alternative argument based on  the weaker control provided by energy estimates.

\eei  
 

\section{Existence for the $(\eps,b)$-regularized problem} 
\label{sec-section3}

\subsection{Main statement on the regularized problem} 
\label{sec-formulation-fixed-point}

\paragraph{Functional setup.}

Throughout this section, a small parameter $\eps \in (0,1)$ is fixed and, for simplicity in the presentation, we suppress the subscript $\eps$ from our notation. Our objective is to apply Schauder's fixed-point theorem the regularized problem, namely, if $X$ is a Banach space and $T:  X \to  X$ is a continuous and compact operator whose image $T(X)$ is bounded, then $T$ admits a fixed-point.
We focus on a spacetime slab limited by two light cones who proper distance is determined by the constant
\bel{notation c1}
c_1:= \min \Big(1,\frac{1}{20c^2} \Big), 
\ee
In order to regularize our reduced equation at null infinity, we also fix a parameter 
\be
b \geq 1, 
\ee  
which will eventually be taken to tend to zero. We consider the Banach space $(X_{b},\,\| \cdot \|_{C^1})$ of all continuously differentiable functions $h=h(u,r)$ endowed with the standard norm $\| \cdot \|_{C^1}$ and vanishing out side the interval $[0, b+2]$, namely
\be
X_{b}:= \bigg\{h \in C^1\big([0,c_1]\times[0,+ \infty)\big)
\, \big/ \,  h(u,r) =0\quad\text{for all $r\geq b+ 2$} \bigg\}. 
\ee


\paragraph{Regularization of the metric.}

In order to guarantee that the expected expression of the metric coefficient $g = \big(1-4c^2\pi r^2 \hb^2\big) e^{\nu+ \lambda}$ formally given in \eqref{fund-eqs} remains {\sl positive} throughout our fixed-point analysis, we proceed as follows. 
Given a sufficiently large parameter $a> 0$, to any pair $(t,h) \in [0,1]\times X_b$ we associate the function
\bel{gath}
g_{a,t,h}(u,r) := \Big( 1-4\pi  \big(c^2 \, (r+ \eps)^2\hb^2-a(1-t)\big)\Big) \, e^{\nu_h+ \lambda_h}.
\ee
We also introduce a smooth function $\rho$ satisfying 
\be
\rho(x)
:= 
\begin{cases}
x, & x \geq 0,
 \\ 
\in [-1, 0]& x < 0, 
\end{cases} 
\ee
and we regularize the family of characteristics as follows. For each point $(u_1, r_1) \in [0,c_1]\times [0,+ \infty)$, we consider the curve $u \mapsto \chi_{a,t, h}(u; u_1,r_1)$ satisfying  
\be
\aligned
& \del_u\chi_{a,t, h} = -  \frac{1}{2} \rho\big(\gb_{a,t,h} \big), 
\qquad 
 \chi_{a,t,h}(u_1; u_1,r_1) =r_1, 
 \qquad 
 u \in [0, u_1]. 
\endaligned
\ee
 

\paragraph{Regularization of the matter field.}

For each $(t,h) \in [0,1]\times X_b $, there exists a unique $k\in C^1$ satisfying the (linear!) transport equation
\bel{main equation}
\aligned
D_{a,t,h} k - \frac{t}{2(r+ \eps)} \big(g_{a,t,h} - \gb_{a,t,h} \big)k
 = \Omega_{a, t, h}, 
\qquad
k\big|_{u=0} = h_0, 
\endaligned
\ee
in which 
\be
\aligned
D_{a, t, h} 
& := \del_u- \frac{1}{2} \rho(\gb_{a,t,h})\del_r,
\\
\Omega_{a, t, h} 
& := - \frac{1}{2(r+ \eps)} \big(g_{a,t,h} - \gb_{a,t,h} \big) \, \hb
   - \frac{c^2}{2}(r+ \eps) \, \hb \, e^{\nu_{h} + \lambda_{h}},
   \endaligned
\ee
This solution however may not have compact support within the interval $[0, b+2]$. So, we define the {\bf solution mapping operator} $\Psi_b : [0,1]\times X_b \to  X_b$ by 
\bel{equa-defpsi}
\Psi_b(t,h) := \eta(r) k, 
\ee
where $\eta$ is a (smooth) cut-off function satisfying
\be 
\eta(r) 
\begin{cases}
= 1, \quad & r < b+1,  
\\
> 0, & r \in (b+1, b+ 2), 
\\
= 0,  & r\geq b+ 2. 
\end{cases} 
\ee


\paragraph{Main statement for this section.} 

We are going to establish the following result. 

\begin{proposition}[Existence theory for the $(\eps,b)$-regularized problem]
\label{lemma Psi cc}
The solution mapping operator  $\Psi_b$ is continuous and compact. 
Consequently, given any data $h_0$ satisfying $\gb_{h_0}> 0$, the solution mapping operator 
$\Psi_b(1,.)$ (with $t=1$) admits a fixed-point $h$. In addition, on the slightly smaller spacetime domain $[0,c_1]\times (0,b+1]$ the positivity property $\gb_h> 0$ holds. 
\end{proposition}

The proof of the first part of Proposition~\ref{lemma Psi cc} is postponed to Appendix~\ref{appendix-one} and relies on a direct application of the properties of the solution mapping operator.
 In the massless case~\cite{Chr}, the fact that $\gb_h> 0$ {\sl was an immediate consequence} of the structure of the Einstein  equations for {\sl massless fields.} 
 For massive field, this property is not immediate and requires the use of the evolution equation for the mass function. 
We single out here two lemmas, and we provide here the main arguments of proof while referring to Appendix~\ref{appendix-one} for further details. 

\begin{lemma}[$C^1$ stability of the solution mapping] 
\label{lemma the bound of Psi}
For any data $h_0$, there exists a positive constant $Q(\| h_0 \|_{C^1}, a, b, \eps)$  
such that 
$$
\sup_{(t,h) \in [0,1]\times X_b} \|\Psi_b(t,h)\|_{C^1} \leq Q(\| h_0 \|_{C^1}, a, b, \eps).
$$
\end{lemma}
  
\begin{lemma}[Sign property of the velocity]
\label{lemma gb> 0} 
Consider any fixed-point solution for $t=1$ or, more generally, any solution $(t,h) \in [0,1]\times X_b$ to the equation  $h =t\Psi_b(t,h)$.  If the initial data satisfies $\gb_{a,t,h_0} > 0$ for all $r \in (0,b+1]$, then the solution satisfies   
$\gb_{a,t,h}> 0$ in $[0,c_1]\times(0,b+1]$. 
\end{lemma}

Indeed, thanks to 
 Lemma \ref{lemma the bound of Psi} (whose proof is outlined in Section~\ref{Lipschitz}, below), we can apply the Schauder fixed-point theorem, implying that $\Psi_b$ has a fixed-point $h$. It then follows from Lemma \ref{lemma gb> 0} (established in Section~\ref{sec-sign-velo}, below) 
 that $\gb_h> 0$ for all $(u,r) \in [0,c_1]\times (0,b+1]$. It is essential to check that the cut-off function $\rho$ reduces to the identity over the range of the solution, and this property
 $\rho\big(\gb_{a,t,h}(u,\chi_{a,t,h}) \big)
= \gb_{a,t,h} \big(u,\chi_{a,t,h} \big)$
 is provided by Lemma~\ref{lemma gb> 0}. 


\subsection{Basic continuity properties} 
\label{Lipschitz}

\paragraph{Lipschitz continuity of the metric coefficient $g_{a,t,h}$.}  

For any $h_1, h_2\in X_b$ we have
 $$
\aligned
|\hb_1- \hb_2| & =  \left|\frac{1}{r+ \eps} \int_0^r(h_1- h_2) \, ds\right| \leq \|h_1- h_2 \, \|_{C^0}, 
\endaligned
$$
and 
$$
\aligned
& |g_{a,t_1,h_1} -g_{a,t_2,h_2}|
 \leq  \bigg|4\pi \Big( c^2 \, (r+ \eps)^2(\hb_2^2- \hb_1^2) -a(t_2-t_1) \Big)e^{\nu_{h_1} + \lambda_{h_1}} \bigg|
\\
& \quad + \bigg|\Big( 1-4\pi \, \big( c^2(r+ \eps)^2\hb_2^2-a(1-t_2)\big) \Big)\big(e^{\nu_{h_1} + \lambda_{h_1}} -e^{\nu_{h_2} + \lambda_{h_2}} \big)\bigg|
\\
& \leq   4\pi a|t_1-t_2|+ 4\pi c^2 \, (r+ \eps)^2 |(\hb_1- \hb_2)(\hb_1+ \hb_2)|
\\
& \quad 
+ \bigg|4\pi \, \Big(1-4\pi c^2 \, (r+ \eps)^2\hb_2^2+4\pi a(1-t_2)\Big)
\int_r^{b+1} \frac{(h_1- \hb_1)^2-(h_2- \hb_2)^2}{r+ \eps} \, ds \bigg|
\endaligned
$$
therefore
$$
\aligned
& |g_{a,t_1,h_1} -g_{a,t_2,h_2}|
  \leq   4\pi a|t_1-t_2|+4\pi c^2 \, (r+ \eps)^2 |(\hb_1+ \hb_2)|\|(\hb_1- \hb_2)\|_{C^0}
\\
&  + 8\pi \, \big|1-4\pi^2 \, c^2 \, (r+ \eps)^2\hb_2^2+4\pi a(1-t_2)\big| \, 
\int_r^{b+1} \frac{|h_1- \hb_1+h_2- \hb_2|}{r+ \eps} \, ds \, \|h_1- h_2 \, \|_{C^0}.
\endaligned
$$
Similarly, we also obtain 
\be
|\gb_{a,t_1,h_1} - \gb_{a,t_2,h_2}|
+ 
|\del_rg_{a,t_1,h_1} - \del_rg_{a,t_2,h_2}| \leq  a_0\|h_1- h_2 \, \|_{C^0} +a_1|t_1-t_2|,
\ee
where the constants $a_0,\,a_1> 0$ depend on $\|h_1\|_{C^0},\,\|h_2 \, \|_{C^0}, a, b, \eps$.


\paragraph{$C^1$ estimate for the matter function $k$.} 

Given $(t, h) \in [0,1]\times X_b$, let $k$ be the solution of  \eqref{main equation} associated with $(t, h)$. First of all, if $\|h\|_{C^0},\,\|\del_rk\|_{C^0}$ are bounded, then $\|\del_uk\|_{C^0}$ is bounded by a positive constant only depending on $\| h \|_{C^0}, \| \del_r k \|_{C^0}, a, b, \eps$.
In fact, by solving \eqref{main equation} along the characteristics $\chi_{a,t,h}(\cdot; u_1,r_1)$ with 
$$
A(u_1) := {\int_0^{u_1} 
 \frac{t(g_{a,t,h} - \gb_{a,t,h})}{2(r_2+ \eps)}  \big|_{\chi_{a, t, h}} \, du_2}
$$
we find 
\begin{equation*}
\label{k along chi}
\aligned
k(u_1,r(u_1)) 
& = h_0(0,r(0)) e^{A(u_1)}  
- \int_0^{u_1}  \Big( \frac{1}{2(r_2+ \eps)} \big(g_{a,t,h} - \gb_{a,t,h} \big) \, \hb- \frac{c^2}{2}(r_2+ \eps) \, \hb \, e^{\nu_{h} + \lambda_{h}}  \Big)_{\chi_{a,t,h}}
e^{A(u_1) - A(u_2)}  \, du_2.
\endaligned
\end{equation*}
 Therefore, we have  
\bel{estiamte C0 of k}
\|k\|_{C^0} \leq a_2,
\ee
where the constant $a_2> 0$ depends on $h_0,\|h\|_{C^0}, a, b, \eps$. Differentiating the equation \eqref{main equation} with respect to $r$, we obtain
\bel{differentiaing main equation}
D_{a,t,h}(\del_rk) - \frac{1}{2(r+ \eps)} \Big( (g_{a,t,h} - \gb_{a,t,h})\rho^\prime(\gb_{a,t,h})+t(g_{a,t,h} - \gb_{a,t,h}) \Big)\del_rk=f_1(h,k,t),
\ee
where
\bel{def of f_1}
\aligned
f_1(h,k,t):& =  \frac{\del_rg_{a,t,h}}{2(r+ \eps)}(tk- \hb) - \frac{1}{2(r+ \eps)^2}(g_{a,t,h} - \gb_{a,t,h})(2tk+h-3\hb)\\
& \quad
- \frac{c^2}{2}he^{\nu_h+ \lambda_{h}} -2\pi c^2\hb(h - \hb)^2e^{\nu_h+ \lambda_{h}}.
\endaligned
\ee
Solving this equation along the characteristics $\chi_{a,t,h}(\cdot; u_1,r_1)$ with 
$$
B(u_1) := \int_0^{u_1} \Big( \frac{1}{2(r_2+ \eps)} \big((g_{a,t,h} - \gb_{a,t,h})\rho^\prime(\gb_{a,t,h})+t(g_{a,t,h} - \gb_{a,t,h})\big) \Big)_{\chi_{a,t,h}} \, du_2, 
$$
we have
\bel{del_rk along chi}
\aligned
& \del_rk(u_1,r(u_1)) = h_0(0,r(0)) e^{B(u_1)}  
 - \int_0^{u_1}[f_1(h,k,t)]_{\chi_{a,t,h}} e^{B(u_1) -B(u_2)}  
\, du_2, 
\endaligned
\ee
which gives us 
$\|\del_rk\|_{C^0} \leq a_3$
for a positive constant $a_3$ depending on $h_0, \| h \|_{C^0}, \| k \|_{C^0}, a, b, \eps$. Combining this result with \eqref{estiamte C0 of k}, we conclude that 
$\|k\|_{C^1} \leq a_4$ similarly for some constant $a_4> 0$  depending on $h_0, \| h \|_{C^0}, \| k \|_{C^0}, a, b, \eps$. 


\paragraph{$C^2$ estimate for the matter function $k$.} 

When sufficient differentiability is available, by differentiating twice the equation \eqref{main equation} with respect to $r$ we obtain
\bel{del_rr of k}
D_{a,t,h}(\del_{rr}k) - \frac{1}{2(r+ \eps)} \Big(
2(g_{a,t,h} - \gb_{a,t,h})\rho^\prime(\gb_{a,t,h})+t(g_{a,t,h} - \gb_{a,t,h})\Big) \, \del_{rr}k =f_2(h,k),
\ee
where the right-hand side 
$$
\aligned
&f_2(h,k,t) := \frac{\del_rg_{a,t,h} \del_rk}{2(r+ \eps)} \big(2t+ \rho^\prime(\gb_{a,t,h})\big)\\
& - \frac{\del_rk}{2(r+ \eps)^2} \Big(
2(g_{a,t,h} - \gb_{a,t,h})(\rho^\prime(\gb_{a,t,h})+t)+ 2t(g_{a,t,h} - \gb_{a,t,h}) -(g_{a,t,h} - \gb_{a,t,h})^2\rho^{\prime\prime}(\gb_{a,t,h})\Big)
\\
& - \frac{\del_rg_{a,t,h}}{2(r+ \eps)^2}(3tk+ 2h-5\hb)+ \frac{3(g_{a,t,h} - \gb_{a,t,h})}{(r+ \eps)^3}(tk+h-2\hb)\\
& - \frac{(g_{a,t,h} - \gb_{a,t,h})}{2(r+ \eps)^2} \del_r h+ \frac{\del_{rr} \gb_{a,t,h}}{2(r+ \eps)}(tk- \hb)+ \frac{c^2}{2} \del_{rr} \big((r+ \eps) \, \hb \, e^{\nu_h+ \lambda_h} \big)
\endaligned
$$
satisfies $\big\|f_2(h,k,t)\big\|_{C^0} \leq a_5$ for some constant $a_5> 0$ depending on $h_0, \| h \|_{C^1}, a, b, \eps$.  By solving \eqref{del_rr of k} we arrive at the bound $\|\del_{rr}k\|_{C^0} \leq a_6$ for a constant $a_6> 0$ depending on $h_0, \| h \|_{C^1}, a, b, \eps$.
 The above estimates allow us to complete the proof of Lemma~\ref{lemma Psi cc}.


\subsection{Sign property for the velocity}
\label{sec-sign-velo}

We now give the proof of Lemma \ref{lemma gb> 0}.  We assume that the equation 
$h =t\Psi_b(t,h)$ holds and that the initial data satisfies $\gb_{a,t,h_0} > 0$. 
Since the case $t=0$ is trivial, 
we consider values $t\ne0$ and argue by contradiction. Let us assume that there exists $(u_0,r_0) \in  [0,c_1]\times[0,b+1]$ such that $\gb_{a,t,h}(u_0,r_0)\leq 0$ and let us set 
\begin{equation*}
u_*:= \sup\Big\{
u_1\in [0,c_1] \, \big/ \,
 \text{$\gb_{a,t,h}(u,r) > 0$ for all $(u,r) \in [0,u_1]\times[0,b+1]$} \Big\}.
\end{equation*}
Since, by our assumption, $\gb_{a,t,h}(0,r) = \gb_{a,t,h_0}> 0$ for all $r\in [0,b+1]$, we have 
$u_*> 0$. It follows from the definition of $u_*$ that  there exists $r_*\in (0,b+1]$ such that $\gb_{a,t,h}(u,r) > 0$ for all $(u,r) \in [0,u_*)\times [0,r_*]$ while $\gb_{a,t,h}(u_*,r_*) =0$ vanishes at the ``first'' point $(u_*,r_*)$. 
Since $h =t\Psi_b(t,h)$, according to the definition of $\Psi_b$ and, in particular, within the domain $[0,u_*)\times [0,r_*]$, we have 
\bel{h =tk}
\aligned
& D_{a,t,h}h = \frac{t}{2(r+ \eps)} \big(g_{a,t,h}(t_0) - \gb_{a,t,h}(t_0)\big)(h - \hb) - \frac{tc^2}{2}(r+ \eps) \, \hb \, e^{\nu_{h} + \lambda_{h}}, 
\\
& h(0,r) =t \, h_0(r),
\endaligned
\ee
As before, given $e^{\nu_h - \lambda_h} := \gb_{a,t,h}$ defined on $[0,u_*)\times [0,r_*]$, 
we define the Hawking mass associated with $h$ by
\bel{(t,epsilon) bondi-mass}
m_h(u,r) := \frac{r+ \eps}{2} \big(1-e^{-2\lambda_{h}(u,r)} \big).
\ee
A straightforward computation based on the identities \eqref{g-equation} gives us 
\bel{formula Dm}
\aligned
\del_rm_h & = 2\pi \Big( e^{-2\lambda_{h}}(h - \hb)^2+c^2(r+ \eps)^2\hb^2-a(1-t)\Big),
\\
\del_um_h & = (r+ \eps) e^{-2\lambda_h} \del_u\lambda_h.
\endaligned
\ee
We  compute the derivative $\del_u\lambda_h$, as follows. Since $h$ has compact support, from \eqref{g-equation} and \eqref{gath} we obtain
$$
\aligned
\del_u(\nu_h+ \lambda_h) & =  -8\pi \int_r^{+ \infty} \frac{(h - \hb)}{s+ \eps}(\del_u h - \del_u\hb) \, ds,
\\
\del_u(\nu_h - \lambda_h) & =  \frac{1}{\gb_{a,t,h}(r+ \eps)} \int_0^rg_{a,t,h} \del_u(\nu_h+ \lambda_h) \, ds
  - \frac{8\pi c^2 }{\gb_{a,t,h}(r+ \eps)} \int_0^r\bigg(\int_0^shds_1\int_0^s\del_u hds_1\bigg) e^{\nu_h+ \lambda_h} \, ds.
\endaligned
$$
After integration by parts we find 
\begin{equation*}
\label{formula del_u lambda}
\aligned
\del_u\lambda_h 
& = \frac{4\pi }{\gb_{a,t,h}(r+ \eps)} \int_0^{r}g_{a,t,h} \int_s^{+ \infty} \frac{(h - \hb)}{s_1+ \eps}(\del_u h - \del_u\hb) \, ds_1ds-4\pi \int_r^{+ \infty} \frac{(h - \hb)}{s+ \eps}(\del_u h - \del_u\hb) \, ds
\\
& \quad + \frac{4\pi c^2}{\gb_{a,t,h}(r+ \eps)} \int_0^r\bigg(\int_0^shds_1\int_0^s\del_u hds_1\bigg) e^{\nu_h+ \lambda_h} \, ds
\\
& = \frac{4\pi }{\gb_{a,t,h}(r+ \eps)} \bigg(\int_0^r\gb_{a,t,h}(h - \hb)(\del_u h - \del_u\hb) \, ds
+c^2\int_0^r\bigg(\int_0^shds_1\int_0^s\del_u h ds_1\bigg) e^{\nu_h+ \lambda_h} \, ds\bigg).
\endaligned
\end{equation*}
Observe next that 
\bel{epsilon equation hb}
\aligned
\del_u h & =  Dh + {1 \over 2} \gb \del_r h, 
\qquad \quad 
\del_u\hb =  {1 \over r+ \eps} \int_0^r (Dh + {1 \over 2} \gb \del_r h) \, ds,
\endaligned
\ee
so that, by swapping from $\del_u$-operator to the $D$-operator, we rewrite \eqref{formula del_u lambda} on $[0,u_*)\times [0,r_*]$ as
$$ 
\aligned
 \del_u\lambda_h
& =  - \frac{4\pi }{\gb_{a,t,h}(r+ \eps)} 
\int_0^r \Big( \big(g- \gb_{a,t,h} \big) \frac{h - \hb}{s+ \eps} - c^2(s+ \eps) \, \hb \, e^{\nu_h+ \lambda_h} 
 + \gb_{a,t,h} \del_r h\Big) \int_0^s \big( D_\eps h+ \frac{1}{2} \gb_{a,t,h} \del_r h \big) \, ds_1  \, ds 
\\
& \quad
+ \frac{4\pi }{\gb_{a,t,h}(r+ \eps)}
 \int_0^r \Big( \big( g_{a,t,h} - \gb_{a,t,h} \big) \frac{h - \hb}{2(s+ \eps)} 
+ \gb_{a,t,h} \del_r h - \gb_{a,t,h} \frac{h - \hb}{s+ \eps} \Big) \int_0^s \big( D_\eps h+ \frac{1}{2} \gb_{a,t,h} \del_r h \big) \, ds_1\, ds 
\\
& \quad + \frac{4\pi }{\gb_{a,t,h}(r+ \eps)} \int_0^r\gb_{a,t,h}(h - \hb)\big(D_\eps h+ \frac{1}{2} \gb_{a,t,h} \del_r h\big) \, ds. 
\endaligned
$$
Thus, we have 
\be
\aligned
\del_u\lambda_h
 = - \frac{4\pi }{\gb_{a,t,h}(r+ \eps)}  \Bigg( \frac{1}{t} \bigg(\int_0^r\big(D_\eps h+ \frac{1}{2} \gb_{a,t,h} \del_r h\big) \, ds\bigg)^2
  - \gb_{a,t,h}(h - \hb) \int_0^r \big( D_\eps h+ \frac{1}{2} \gb_{a,t,h} \del_r h \big) \, ds \Bigg).
\label{partialu}
\endaligned
\ee
Combining this equation with \eqref{formula Dm}, we arrive at 
\be\label{evolution of m with (t,epsilon)}
\aligned
& \del_u m_h - \frac{1}{2} \gb_{a,t,h} \del_rm_h - \pi \, \gb_{a,t,h}(1-t)a
  = - \frac{\pi}{e^{\nu_h+ \lambda_h}} \bigg(2(r+ \eps)D_\eps \hb+ \big(1-t\big)\gb_{a,t,h}(h - \hb)\bigg)^2
\\
&  - \pi \gb_{a,t,h}c^2(r+ \eps)^2\hb^2- \frac{4\pi}{e^{\nu_h+ \lambda_h}} \bigg(\frac{1}{t} -1\bigg)\bigg(\int_0^r\big(D_\eps h+ \frac{1}{2} \gb_{a,t,h} \del_r h\big) \, ds\bigg)^2.
\endaligned
\ee
The right-hand side of this identity is clearly non-positive, and we thus deduce that 
\be
\del_u m_h\leq \frac{\gb_{a,t,h}}{2} \del_rm_h+ \pi \, \gb_{a,t,h} (1-t)a,
\ee
or, equivalently by using the mass evolution equation in \eqref{formula Dm}, for all $(u,r) \in [0,u_*)\times [0,r_*]$, 
\be
\del_u\lambda_h\leq \pi \, \gb_{a,t,h} \frac{(h - \hb)^2}{r+ \eps} + \pi c^2(r+ \eps) \, \hb^2e^{\nu_h+ \lambda_h}. 
\ee
By integration in $u$, within the domain $[0,u_*)\times [0,r_*]$ we obtain  
\bel{bound above of lambda}
\lambda_h(u,r)\leq \lambda_h(0,r)+ \int_0^u\bigg(\pi \, \gb_{a,t,h} \frac{(h - \hb)^2}{r+ \eps} + \pi c^2(r+ \eps) \, \hb^2e^{\nu_h+ \lambda_h} \bigg) \, du, 
\ee
which is a finite upper bound for $\lambda_h$ and prevents $\lambda_h(u,r)$ to blow-up. 
By recalling that we have assumed that $\gb_{a,t,h}(u_*,r_*) =0$, in the limit $u \to u_*$ this inequality tells us that
$$
0= \lim_{u\to u_*} e^{(\nu_h - \lambda_h) (u,r_*)} \lim_{u\to u_*} e^{-(\nu_h+ \lambda_h) (u,r_*)}= \lim_{u\to u_*} e^{-2\lambda_h (u,r_*)}> 0,
$$
which is a contradiction. This completes the proof of Lemma \ref{lemma gb> 0}.


\section{Geometric estimates on the mass and energy}
\label{sec-section4}

\subsection{Hawking mass} 

We assume that an initial data $h_0$ is given and satisfy $\gb_{h_0}> 0$ for all $r\geq0$. Given a sufficiently large $b \geq 1$, we can now assume that the mapping $\Psi_b(1,.)$ admits a fixed point denoted by $h_{\eps,b}$, satisfying the following properties:
\begin{enumerate}

\item The velocity satisfies $e^{\nu_{\eps,b} - \lambda_{\eps,b}}:= \gb_{\eps,b}> 0$ in the domain 
$\Ocal_{c_1, b}$. 

\item The mass function 
\bel{ep bondi mass}
m_{\eps,b}= \frac{r+ \eps}{2} \bigg(1- \frac{\gb_{\eps,b}}{e^{\nu_{\eps,b} + \lambda_{\eps,b}}} \bigg),
\ee
satisfies the identities and inequalities \eqref
{m is increasing w.r.t. r} --\eqref{mass-bound} below
(as follows from the computation in \eqref{formula Dm} with $t=1$). 

\end{enumerate}

\begin{proposition}[Monotonicity properties of the Hawking mass]
 The $(\eps, b)$-regularized solutions to the Einstein-Klein-Gordon equation satisfy
\bel{m is increasing w.r.t. r}
\del_rm_{\eps,b}=2\pi \, \big( e^{-2\lambda_{\eps,b}}(h_{\eps,b} - \hb_{\eps,b})^2+c^2(r+ \eps)^2\hb_{\eps,b}^2\big)\geq 0,
\ee 
\bel{ep equation mass}
D_\eps m_{\eps,b}= -  \frac{4\pi}{e^{\nu_{\eps,b} + \lambda_{\eps,b}}}(r+ \eps)^2\big(D_\eps \hb_{\eps,b} \big)^2- \pi \, \gb_{\eps,b}c^2(r+ \eps)^2\hb^2_{\eps,b} 
\ee
 in the slab $[0,c_1]\times [0,b+1]$
and, consequently, $m_{\eps,b}$ is non-decreasing when $r$ increases, and is non-increasing when $u$ increases, with 
\bel{mass-bound}
m_{\eps,b}(u,r)\leq m_{\eps,b}(u,b+1)\leq M_\eps(0), 
\qquad 
(u,r) \in [0,c_1]\times [0,b+1]. 
\ee
\end{proposition}

More precisely, \eqref{ep equation mass} follows from the proof of Lemma \ref{lemma gb> 0} (with $t=1$), while 
\eqref{mass-bound} is obtained by combining \eqref{m is increasing w.r.t. r} and \eqref{ep equation mass}.
Before analyzing the energy functional associated with the matter field $h_{\eps,b}$, let us derive some consequences of the positivity and monotonicity of the Hawking mass.
The first lemma below gives an expression equivalent to the Einstein field equation $G_{\theta\theta}=8\pi T_{\theta\theta}$ (for the angular direction on the orbits of symmetry)
and the second lemma provides us with pointwise  
Observe that it follows directly from \eqref{m is increasing w.r.t. r} that
\bel{integral of del_rm}
\aligned
\int_0^b{2\pi \, \big( e^{-2\lambda_{\eps,b}}(h_{\eps,b} - \hb_{\eps,b})^2+c^2(r+ \eps)^2\hb_{\eps,b}^2\big) \, dr}
= m_{\eps,b}(u,b+1) -m_{\eps,b}(u,0)\leq  m_{\eps,b}(u,b+1). 
\endaligned
\ee



\begin{lemma}[A remarkable identity]
\label{Fzeta epsilon=0}
On the spacetime slab $[0,c_1]\times [0,b+1]$ one has
$$
\frac{\del_r\gb_{\eps,b}}{r+ \eps} + \del_r\big(\gb_{\eps,b} \del_r\nu_{\eps,b} \big) - \del_{ur}(\nu_{\eps,b} + \lambda_{\eps,b}) -8\pi \, \Big(\del_r \, \hb_{\eps,b}D_\eps \hb_{\eps,b} - \frac{c^2}{2}e^{\nu_{\eps,b} + \lambda_{\eps,b}} \hb_{\eps,b}^2\Big) =0.
$$	
\end{lemma}

\begin{lemma}[Pointwise bounds on $h_{\eps,b}$ and $\hb_{\eps,b}$]  
\label{lemma-rhb} 
On the spacetime slab $[0,c_1]\times [0,b]$ one has
\bel{control rhb in C0}
(r+ \eps)|\hb_{\eps,b}|(u,r)\leq \Acal_\eps
(u,r),
\ee
where $\Acal_\eps
(u,r)$ is defined in \eqref{def A}.
Moreover,  provided that $b \geq 1$ is sufficiently large,  one also has 
\be
\label{bound of rhb}
\aligned
\sup_{(u,r) \in [0,c_1]\times [b,b+ 2]}|h_{\eps,b}| \leq  A_1(c_1,b),
\qquad\quad
\sup_{(u,r) \in [0,c_1]\times [0,+ \infty)}|(r+ \eps) \, \hb_{\eps,b}| \leq A_2(c_1,b),
\endaligned
\end{equation}
with
$$
\aligned
A_1(c_1,b):& = 2 \Big( \|h_0\|_{C^0([b- \frac{1}{8c^2},+ \infty))} 
+C(c_1,b)\Acal_\eps(c_1,0) \Big)e^{c_1 / (b- \frac{c_1}{2} + \eps)},
\\
A_2(c_1,b):& = 2A_1(c_1,b)+ \Acal_\eps
(c_1,0),
\qquad
\quad 
C(c_1,b) := c_1\bigg(\frac{4\big(1+ 2e^{-2}c^2(2b+4+ \eps)^2\big)}{(2b-c_1+ 2\eps)^2} +c^2\bigg).
\endaligned
$$	
\end{lemma}

\begin{proof}[Proof of Lemma~\ref{Fzeta epsilon=0}]
From the identity
$$
e^{2\nu_{\eps,b}} \big(1+ 2(r+ \eps)\del_r\nu_{\eps,b} \big) =e^{\nu_{\eps,b} + \lambda_{\eps,b}} \del_r\big((r+ \eps) e^{\nu_{\eps,b} - \lambda_{\eps,b}} \big)+(r+ \eps) e^{2\nu_{\eps,b}} \del_r(\nu_{\eps,b} + \lambda_{\eps,b}), 
$$
it follows that
\bel{formaula derivative of nu}
\aligned
\del_r\nu_{\eps,b}& = \frac{1}{2(r+ \eps)} \bigg(\frac{g_{\eps,b}}{\gb_{\eps,b}} -1\bigg)+ \frac{1}{2} \del_r(\nu_{\eps,b} + \lambda_{\eps,b}),\\
\gb_{\eps,b} \del_{rr} \nu_{\eps,b}& = \frac{2\pi g_{\eps,b}(h_{\eps,b} - \hb_{\eps,b})^2}{(r+ \eps)^2} - \frac{g_{\eps,b}^2}{2\gb_{\eps,b}(r+ \eps)^2}
+ \frac{\gb_{\eps,b}}{2(r+ \eps)^2} -4\pi c^2 h_{\eps,b} \hb_{\eps,b}e^{\nu_{\eps,b} + \lambda_{\eps,b}} + \frac{\gb_{\eps,b}}{2} \del_{rr}(\nu_{\eps,b} + \lambda_{\eps,b}).
\endaligned
\ee
In turn from \eqref{g-equation} we get
$\del_r(\nu_{\eps,b} + \lambda_{\eps,b}) = 4\pi \, \frac{(h_{\eps,b} - \hb_{\eps,b})^2}{r+ \epsilon}$
on $[0,c_1]\times [0,b+1]$, and then
$$
\aligned
& \del_{ur}(\nu_{\eps,b} + \lambda_{\eps,b}) - \frac{\gb_{\eps,b}}{2} \del_{rr}(\nu_{\eps,b} + \lambda_{\eps,b}) 
 =  4\pi D_\eps \bigg(\frac{(h_{\eps,b} - \hb_{\eps,b})^2}{r+ \eps} \bigg)\\
 & = 8\pi \frac{(h_{\eps,b} - \hb_{\eps,b})}{r+ \eps}D_\eps(h_{\eps,b} - \hb_{\eps,b})+ 2\pi \, \frac{(h_{\eps,b} - \hb_{\eps,b})^2}{(r+ \eps)^2} \gb_{\eps,b} \\
 & = 4\pi \, \frac{(h_{\eps,b} - \hb_{\eps,b})^2}{(r+ \eps)^2}(g_{\eps,b} - \gb_{\eps,b}) -4\pi c^2(h_{\eps,b} - \hb_{\eps,b}) \, \hb_{\eps,b}e^{\nu_{\eps,b} + \lambda_{\eps,b}} 
 + 2\pi \, \frac{\gb_{\eps,b}(h_{\eps,b} - \hb_{\eps,b})^2}{(r+ \eps)^2} -8\pi \, \frac{h_{\eps,b} - \hb_{\eps,b}}{r+ \eps}D_\eps \hb_{\eps,b},
\endaligned
$$
where we have integrated by parts and used the expression for $\del_u\lambda_{\eps,b}$ adapted from \eqref{partialu}. Finally, 
combining the above two identities together, 
we arrive at the desired result. 
\end{proof}


\begin{proof}[Proof of Lemma~\ref{lemma-rhb}]

Given $(u_1,r_1) \in (0,c_1]\times [0,b]$, since $\frac{\eps}{2} \leq m_{\eps,b}(u_1,r_1)\leq M_\eps(0)$,
by solving \eqref{ep equation mass} along the characteristics $\chi_\eps(\cdot; u_1, r_1)$ we obtain
\begin{equation*}
\aligned
\frac{\eps}{2} \leq m_{\eps,b}(u_1,r_1) 
& =  \  m_{\eps,b}(0,r(0))
- \int_0^{u_1}  \Big( \frac{4\pi}{e^{\nu_{\eps,b} + \lambda_{\eps,b}}}(r_2+ \eps)^2(D_\eps \hb_{\eps,b})^2
+ \pi c^2(r_2+ \eps)^2\gb_{\eps,b} \hb_{\eps,b}^2 \Big)_{\chi_{\eps,u_1}} \, du_2 \\
& = \ m_{\eps,b}(0,r(0))
- \int_0^{u_1}  \Big( \frac{4\pi (D\big((r_2+ \eps) \, \hb_{\eps,b} \big)+ \frac{1}{2} \gb_{\eps,b} \hb_{\eps,b})^2}{e^{\nu_{\eps,b} + \lambda_{\eps,b}}}
+ \pi c^2(r_2+ \eps)^2\gb_{\eps,b} \hb_{\eps,b}^2 \Big)_{\chi_{\eps,u_1}} \, du_2
\\
& \leq  \ m_{\eps,b}(0,r(0))
- \int_0^{u_1}  \Big( 2\pi \, \bigg(D\big((r_2+ \eps) \, \hb_{\eps,b} \big)+ \frac{1}{2} \bigg(\frac{\gb_{\eps,b}}{r_2+ \eps} +c\sqrt{\gb_{\eps,b}} \bigg)(r_2+ \eps) \, \hb_{\eps,b} \bigg)^2 \Big)_{\chi_{\eps,u_1}} \, du_2.
\endaligned
\end{equation*}
With 
$J(u_1) := \int_0^{u_1} \Big(
\frac{\gb_{\eps,b}}{r_2+ \eps} +c\sqrt{\gb_{\eps,b}} \Big)_{\chi_{\eps,u_1}} \, du_2
$
it follows that
$$
\aligned
|(r_1+ \eps) \, \hb_{\eps,b}(u_1,r_1)| 
& \leq   |(r(0)+ \eps) \, \hb_0(r(0))|e^{- \frac{1}{2} J(u_1)}
\\
& \quad + \bigg|\int_0^{u_1}  
 \bigg(D\big((r_2+ \eps) \, \hb_{\eps,b} \big)
+ \frac{1}{2} \bigg(\frac{\gb_{\eps,b}}{r_2+ \eps} +c\sqrt{\gb_{\eps,b}} \bigg)(r_2+ \eps) \, \hb_{\eps,b} \bigg)
\,
 e^{- \frac{1}{2} (J(u_1) - J(u_2))} 
 \big|_{\chi_{\eps,u_1}} \, du_2\bigg|,
\endaligned
$$
hence, by H\"{o}lder's inequality we can show \eqref{control rhb in C0}
\begin{equation*}
\aligned
\label{rhb is uniformly bounded}
&  
|(r_1+ \eps) \, \hb_{\eps,b}(u_1,r_1)| 
\\
& \leq  |(r(0)+ \eps) \, \hb_0(r(0))|e^{- \frac{1}{2} \int_0^{u_1} \Big( \frac{\gb_{\eps,b}}{r_2+ \eps} +c\sqrt{\gb_{\eps,b}}  \Big)_{\chi_{\eps,u_1}} \, du_2}
\\
& \quad +  \Big( \int_0^{u_1}e^{- (J(u_1) - J(u_2))} \, du_2 \Big)^{1/2}
\, 
 \Big( \int_0^{u_1}  \Big( \bigg(D\big((r_2+ \eps) \, \hb_{\eps,b} \big)+ \frac{1}{2}
\bigg(\frac{\gb_{\eps,b}}{r_2+ \eps} +c\sqrt{\gb_{\eps,b}} \bigg)(r_2+ \eps) \, \hb_{\eps,b} \bigg)^2  \Big)_{\chi_{\eps,u_1}} \hspace{-0.5cm} \, du_2 \Big)^{1/2} 
\\
& \leq  |(r(0)+ \eps) \, \hb_0(r(0))|+ \sqrt{\frac{u_1 M_\eps(0)}{2\pi}}.
\endaligned
\end{equation*}
Next, let $(u_*,r_*)\nu_{\eps,b} + \lambda_{\eps,b}$ be a maximum point of $(r+ \eps)^2\hb_{\eps,b}^2e^{\nu_{\eps,b} + \lambda_{\eps,b}}$ in $[0,c_1]\times [b- \frac{c_1}{2},+ \infty)$. 
Since $(r+ \eps)^2e^{\nu_{\eps,b} + \lambda_{\eps,b}}$  is increasing with respect to $r$, then $r_*$ is also a maximum point of $|\hb_{\eps,b}|(u_*,.)$ in $[r_*,+ \infty)$. Therefore, 
we obtain 
$$
\aligned
 (r_*+ \eps)^2\big(|\hb_{\eps,b}|^2e^{\nu_{\eps,b} + \lambda_{\eps,b}} \big)(u_*,r_*)& \leq  (2b+4+ \eps)^{2} \big(|\hb_{\eps,b}|^2e^{-2\sqrt{2\pi}|\hb_{\eps,b}|} \big)(u_*,r_*)
  \leq  \frac{1}{2\pi} \Big(\frac{2b+4+ \eps}{e} \Big)^{2},
\endaligned
$$
hence, for all $(u,r) \in [0,c_1]\times [b- \frac{c_1}{2},+ \infty)$
\bel{rhb exp 2}
|g_{\eps,b}| = \bigg|
\big(
1-4\pi \, c^2(r+ \eps)^2\hb_{\eps,b}^2\big) e^{\nu_{\eps,b} + \lambda_{\eps,b}} \bigg| \leq 2 \, c^2\bigg(\frac{2b+4+ \eps}{e} \bigg)^{2}.
\ee
Now, let $k_{\eps,b}$ be the solution of \eqref{main equation} associated with $(1,h_{\eps,b})$. Then, we obtain from \eqref{k along chi} with $t=1$ that, for all $(u_1,r_1)\in [0,c_1]\times [b,b+ 2]$,
$$
\aligned 
&
|k_{\eps,b}(u_1,r_1)|  \leq  |h_0(0,r(0))|e^{\int_0^{u_1} 
 \frac{g_{\eps,b} - \rho(\gb_{\eps,b})}{2(r_2+ \eps)}  \big|_{\chi_{\eps,u_1}} \, du_2} 
\\
&  + \bigg|\int_0^{u_1}  \Big( \frac{g_{\eps,b} - \rho(\gb_{\eps,b})}{2(r_2+ \eps)} \hb_{\eps,b} - \frac{c^2}{2}(r_2+ \eps) \, \hb_{\eps,b} e^{\nu_{\eps,b} + \lambda_{\eps,b}}  \Big)_{\chi_{\eps,u_1}}e^{\int_{u_2}^{u_1} \Big( \frac{g_{\eps,b} - \rho(\gb_{\eps,b})}{(r_3+ \eps)}  \Big)_{\chi_{\eps,u_1}} \, du_3} \, du_2\bigg|.
\endaligned
$$
%
Since $\rho(\gb_\eps) \in [-1,1]$,
 along the characteristics $\chi_{h_\eps,u_1}$ we have $r(u) \in [b- \frac{c_1}{2},\,b+ 2+ \frac{c_1}{2}]$ for all $u\in [0,u_1]$. 
 Therefore, by \eqref{rhb exp 2} and the above bound on $k_{\eps,b}(u_1,r_1)$ we have 
$$
\aligned
\sup_{(u,r) \in [0,c_1]\times [b,b+ 2]}|k_{\eps,b}|
& \leq \kappa_1 \, \|h_0\|_{C^0([b- \frac{1}{8c^2},+ \infty))}
 + \kappa_1 \,
\int_0^{u_1}  \Big( \bigg(\frac{1+ 2 \, c^2e^{-2}(2b+4+ \eps)}{2(r_2+ \eps)^2} + \frac{c^2}{2} \bigg)|(r_2+ \eps) \, \hb_{\eps,b}| \Big)_{\chi_{\eps,u_1}} \hspace{-0.5cm} \, du_2
\\
 & \leq  \kappa_1 \, \|h_0\|_{C^0([b- \frac{1}{8c^2},+ \infty))}
 + C(c_1,b) \kappa_1 \,
 \Big( 
 \sup_{(u,r) \in [0,c_1]\times [b,b+ 2]} \hspace{-0.3cm}|h_{\eps,b}| +  \frac{1}{2} \sup_{(u,r) \in [0,c_1]\times [0,b]} \hspace{-0.4cm}|(r+ \eps) \, \hb_{\eps,b}| \Big), 
\endaligned
$$
where $\kappa_1 := e^{c_1/ (b- \frac{c_1}{2} + \eps)}$.
Since $|h_{\eps,b}| \leq |k_{\eps,b}|$, it follows from the definition of $c_1$ that, for any $b \geq 1$ sufficiently large,
$$
\aligned
 \sup_{(u,r) \in [0,c_1]\times [b,b+ 2]}|h_{\eps,b}|
& \leq  2 \kappa_1 \, \|h_0\|_{C^0([b- \frac{1}{8c^2},+ \infty))} 
+  C(c_1,b) \kappa_1 \, 
\sup_{(u,r) \in [0,c_1]\times [0,b]} \hspace{-0.5cm}|(r+ \eps) \, \hb_{\eps,b}| 
\\
 & \leq  A_1(c_1,b),
 \endaligned
 $$
 where the last inequality follows from \eqref{control rhb in C0}. 
 Consequently, we find
 $$ 
 \aligned
 \sup_{(u,r) \in [0,c_1]\times[0,+ \infty)}|(r+ \eps) \, \hb_{\eps,b}|
 & \leq  \sup_{(u,r) \in [0,c_1]\times[0,b]}|(r+ \eps) \, \hb_{\eps,b}|+ 2\sup_{(u,r) \in [0,c_1]\times [b,b+ 2]}|h_{\eps,b}| \leq  A_2(c_1,b), 
 \endaligned
$$
which completes the proof of Lemma~\ref{lemma-rhb}.
\end{proof}
 

\subsection{Statement of the energy estimates}

We now derive $L^2$ bounds for the functions $h_{\eps,b}$, $\hb_{\eps,b}$, and $h_{\eps,b} - \hb_{\eps,b}$. and, to this end, it is convenient to introduce the auxiliary functions 
\be
\aligned
\label{definition-B1-B2}
B_1(u,b,c_1)&:=\int_0^{u} \gb_{\eps,b}(u_1,b) h_{\eps,b}(u_1,b) \, du_1+ 2 
\sup_{(u,r) \in [0,c_1]\times [b,b+ 2]}h^2_\eps,
\\
B_2(u,b,c_1)&:=\int_b^{b+ 2}(h_{\eps,b} - \hb_{\eps,b})^2 \, dr +	(b+ 2+ \eps)^2|\hb_{\eps,b}|^2(u,b+ 2+ \eps) \int_{b+ 2}^{+ \infty} \frac{1}{(r+ \eps)^2} \, dr \\
& \quad + \frac{1}{4\pi} \int_0^u\bigg(\frac{g_{\eps,b} - \gb_{\eps,b}}{2} + 2\pi \, \gb_{\eps,b}(h_{\eps,b} - \hb_{\eps,b})^2\bigg)(u_1,b) \, du_1\\
& \quad + \int_0^u\bigg(\frac{(\Acal_\eps
(u_1,0))^2}{b} \int_0^b\frac{g_{\eps,b} - \gb_{\eps,b}}{(r+ \eps)^2} \, dr + \frac{\Acal_\eps
(u_1,0)}{\sqrt{b}} \sqrt{\frac{m_{\eps,b}(u_1,b)}{\pi}} \bigg) \, du_1,
\endaligned
\ee
where $\Acal_\eps
(u,r)$ is defined in \eqref{def A}.
Observe that these functions are controlled in terms of the initial data only. Namely, 
by Lemma \ref{lemma-rhb}, for all $u\in [0,c_1]$ and any sufficiently large $b \geq 1$, we have 
\bel{B-bounds}
\aligned
\limsup_{b\to + \infty}B_1(u,b,c_1) \leq  c^2 B_I(c_1),\qquad
\limsup_{b\to + \infty}B_2(u,b,c_1) \leq  \frac{c_1}{4\pi} +c^2 B_{II}(c_1),
\endaligned
\ee
where
\be
\aligned
B_I (c_1)&:=(8e^{-2} +1)c_1^2\Acal_\eps 
(c_1,0) \big( 1+ 2 \, c^2\Acal_\eps(c_1,0) \big)
\\
B_{II} (c_1)&:= \pi c_1 \, \Acal_\eps(c_1,0)+ \big(8e^{-2} +1\big)^2\bigg(\frac{c_1}{2} + 2\bigg)c_1^2 \, c^2 (\Acal_\eps(c_1,0))^2.
\endaligned
\ee 
\begin{proposition}[Energy estimates for  the functions $h_{\eps,b}$  and $\hb_{\eps,b}$]
\label{lemma-energy} 
The following inequalities hold
\bel{equa----415} 
\aligned
&(1) 
&&
&
c^2\int_0^b
(r+ \eps)^2\hb_{\eps,b}^2 \, dr 
&&& \leq \frac{m_{\eps,b}(u,b)}{2\pi},
\\
&(2) 
&&
&
\int_0^{+ \infty}h^2_{\eps,b} \, dr 
&&& \leq \int_{0}^{b} \hspace{-0.2cm}h_0^2 \, dr + \frac{u}{8\pi} + c^2 \Hcal_\eps(u)
+B_1(u,b,c_1),
\endaligned
\ee
where $\Hcal_\eps(u)$ 
depends on the initial data and is defined in \eqref{def A}.
\end{proposition}

The proof of this proposition is given below, while the proof of the next result is analogous and is
postponed to Appendix~\ref{Appendix-third}. 

\begin{proposition}[$L^2$ estimate for  the function $h_{\eps,b} - \hb_{\eps,b}$]
\label{lemma-energy-difference} 
The following inequality holds
\begin{equation*}
\int_0^{+ \infty}(h_{\eps,b} - \hb_{\eps,b})^2 \, dr 
 \leq\int_0^b (h_0- \hb_0)^2 \, dr + \frac{u}{4\pi}
+c^2\Hcal_\eps(u)+ \Kcal_\eps(u)
+B_2(u,b,c_1),
\end{equation*}
where $\Hcal_\eps(u)$ and $\Kcal_\eps(u)$ depend on the initial data and are defined in \eqref{def A}.
\end{proposition} 


\subsection{Derivation of the energy estimates}

We now give a proof of Proposition~\ref{lemma-energy}. The inequality (1) follows from \eqref{mass-bound} and \eqref{integral of del_rm}, so we now focus on the inequality (2). 
By integrating by parts, we obtain 
\be
\aligned
\label{formula energy h2}
& \del_u  \bigg(\int_0^{b}h_{\eps,b}^2 \, dr\bigg) = \int_0^{b} \Big( D_\eps(h_{\eps,b}^2)+ \frac{1}{2} \gb_{\eps,b} \del_r(h_{\eps,b}^2) \Big) \, dr\\
&= \int_0^{b}  \Big( h_{\eps,b}(h_{\eps,b} - \hb_{\eps,b})\del_r\gb_{\eps,b} 
 - \frac{1}{2}h^2_{\eps,b}\del_r\gb_{\eps,b}
- c^2(r+ \eps) e^{\nu_{\eps,b} + \lambda_{\eps,b}}h_{\eps,b} \hb_{\eps,b}  \Big) \, dr + \gb_{\eps,b}(u,b)h_{\eps,b}^2(u,b)\\
& = \int_0^{b}  \Big( - \frac{1}{2} \hb_{\eps,b}^2\del_r\gb_{\eps,b} - \frac{\gb_{\eps,b}}{2(r+ \eps)}(h_{\eps,b} - \hb_{\eps,b})^2 \Big) \, dr
+ \frac{1}{8\pi} \int_0^{b} \del_rg_{\eps,b} \, dr + \gb_{\eps,b}(u,b)h_{\eps,b}^2(u,b)
\\
& = 2\pi c^2\int_0^{b}(r+ \eps) \, \hb^4_\eps e^{\nu_{\eps,b} + \lambda_{\eps,b}} \, dr 
- \int_0^{b}  \Big( \frac{e^{\nu_{\eps,b} + \lambda_{\eps,b}} - \gb_{\eps,b}}{2(r+ \eps)} \hb_{\eps,b}^2+ \frac{\gb_{\eps,b}}{2(r+ \eps)}(h_{\eps,b} - \hb_{\eps,b})^2 \Big) \, dr\\
& \quad
+ \frac{1}{8\pi} \int_0^{b} \del_rg_{\eps,b} \, dr + \gb_{\eps,b}(u,b)h_{\eps,b}^2(u,b),
\endaligned
\end{equation}
We claim that the first term of the right-hand side is bounded by a continuous function of $u$. On the one hand, by Lemma \ref{lemma-rhb} we have (for all $r\in [0,b]$) 
\bel{r>1}
|\hb_{\eps,b}| \leq \frac{\Acal_\eps(u,0)}{r+ \eps}.
\ee
On the other hand, assuming that there exists $r_*<1- \eps$ such that
$
|\hb_{\eps,b}(u,r_*) - \hb_{\eps,b}(u,1- \eps)|> \frac{1}{2}|\hb_{\eps,b}(u,r_*)|,
$
then, by the Cauchy-Schwarz inequality, 
\bel{estimate nu+lambda}
\aligned
4\pi \int_{r_*}^{b}(s+ \eps)(\del_r \, \hb_{\eps,b})^2ds 
& \geq  4\sqrt{\pi} \bigg|\int_{r_*}^{1- \eps} \del_r \, \hb_{\eps,b} \, ds\bigg|- \int_{r_*}^{1- \eps} \frac{1}{s+ \eps} \, ds
\\
 & \geq 2\sqrt{\pi} \, |\hb_{\eps,b}(u,r_*)|- \ln\bigg(\frac{1}{r_*+ \eps} \bigg),
\endaligned
\ee
and consequently  
\bel{r<1}
(r_*+ \eps) \big(|\hb_{\eps,b}|^4e^{\nu_{\eps,b} + \lambda_{\eps,b}} \big)(u,r_*)
\leq 
|\hb_{\eps,b}|^4e^{-2\sqrt{\pi} \, |\hb_{\eps,b}|} \big|_{(u,r_*)}
\leq \frac{4}{\pi}e^{-2\sqrt{2} \pi^{1/4}}.
\ee
From \eqref{r>1} and \eqref{r<1}, we deduce that
\bel{estimate (g-gb)hb2}
\aligned
- \frac{1}{4\pi c^2} \int_0^{b} \hb_{\eps,b}^2\del_r\gb_{\eps,b} \, dr& \leq \int_0^{b}(r+ \eps) \, \hb^4_\eps e^{\nu_{\eps,b} + \lambda_{\eps,b}} \, dr\\
& \leq  \max\Big(
16\big( \Acal_\eps(u,0) \big)^4,\,\frac{4}{\pi}e^{-2\sqrt{2} \pi^{1/4}} \Big)
 + \big( \Acal_\eps(u,0) \big)^4 \int_1^{+ \infty} \frac{dr}{r^3} \\
& \leq \frac{33}{2} \max\Big(\big( \Acal_\eps(u,0) \big)^4,\,\frac{1}{4\pi}e^{-2\sqrt{2} \pi^{1/4}} \Big),
\endaligned
\ee
hence thanks to \eqref{formula energy h2}
$$
\aligned
\int_0^{b}h_{\eps,b}^2 \, dr& \leq  \int_{0}^{b}h_0^2 \, dr + \frac{u}{8\pi} + c^2 \Hcal_\eps(u)
+ \int_0^{u} \gb_{\eps,b}(u_1,b) h^2_{\eps,b}(u_1,b) \, du_1.
\endaligned
$$
Combining this result with Lemma \ref{lemma-rhb}, we obtain
\bel{energy h2}
\aligned
\int_0^{+ \infty}h_{\eps,b}^2 \, dr& \leq  \int_0^{b}h_{\eps,b}^2 \, dr + 2 \sup_{(u,r) \in [0,c_1]\times [b,b+ 2]}h^2_\eps 
  \leq  \int_{0}^{b}h_0^2 \, dr + \frac{u}{8\pi} + c^2 \Hcal_\eps(u)
+B_1(u,b,c_1). 
\endaligned
\qedhere
\ee


\section{Existence for the $\eps$-regularized problem}
\label{sec-section5}

\subsection{Local existence} 

The proof of Theorem~\ref{lemma-global epsilon existence} concerning the global existence for the regularized problem is divided in several steps, as follows. 
We already derived $L^2$-bounds for $h_{\eps,b}$, $\hb_{\eps,b}$, and $h_{\eps,b} - \hb_{\eps,b}$; cf. Propositions~\ref{lemma-energy} and \ref{lemma-energy-difference}. 
We will show that 
$g_{\eps,b} \to g_\eps$ and $\gb_{\eps,b} \to \gb_\eps$ uniformly on compact sets.
 We will next analyze the evolution equation for the $\eps$-Hawking mass (Section~\ref{section---52}, below). 
Finally, we will derive energy estimates which will enable us to finalize the argument (Section~\ref{section53-}). 

Let $k_{\eps,b}$ be the unique solution of \eqref{main equation} associated with $(1,h_{\eps,b})$. By solving that equation along the characteristics 
$\chi_\eps(u; u_1,r_1)$, with $L(u_1) := \int_0^{u_1} 
\frac{g_{\eps,b} - \gb_{\eps,b}}{2(r_2+ \eps)}  \big|_{\chi_{\eps,u_1}} \, du_2$ we have
\be
\aligned
k_{\eps,b}&(u_1,r(u_1)) 
= h_0(0,r(0)) e^{L(u_1)} 
  - \int_0^{u_1}  \Big( \frac{g_{\eps,b} - \gb_{\eps,b}}{2(r_2+ \eps)} \hb_{\eps,b} - \frac{c^2}{2}(r_2+ \eps) \, \hb_{\eps,b} e^{\nu_{\eps,b} + \lambda_{\eps,b}}  \Big)_{\chi_{\eps,u_1}}
  e^{-L(u_1) + L(u_2)} \, du_2.
\endaligned
\end{equation}
It follows, from \eqref{bound of rhb}, that there exists a constant $d_1> 0$ only depending on $(h_0,c_1,\eps)$ such that
\bel{bound uniform k}
\sup_{(u,r) \in [0,u_1]\times [0,+ \infty)}|k_{\eps,b}|(u,r)\leq d_1.
\ee
Differentiating now the equation \eqref{main equation} with respect to $r$, we get
\bel{dpartialrk}
D_{h_{\eps,b}}(\del_rk_{\eps,b}) - \frac{1}{2(r+ \eps)} \Big( (g_{\eps,b} - \gb_{\eps,b})\rho^\prime(\gb_{\eps,b})+(g_{\eps,b} - \gb_{\eps,b}) \Big)\del_rk_{\eps,b}=f_1(h_{\eps,b},k_{\eps,b}),
\ee
where
$$
\aligned
f_1(h_{\eps,b},k_{\eps,b})
& := \frac{\del_rg_{\eps,b}}{2(r+ \eps)}(k_{\eps,b} - \hb_{\eps,b})
       - \frac{1}{2(r+ \eps)^2} \Big( (g_{\eps,b} - \gb_{\eps,b})(k_{\eps,b} - \hb_{\eps,b})
 +(g_{\eps,b} - \gb_{\eps,b})(k_{\eps,b} +h_{\eps,b} -2\hb_{\eps,b}) \Big)
\\
& \quad
- \frac{c^2}{2}h_{\eps,b}e^{\nu_{\eps,b} + \lambda_{\eps,b}} -2\pi c^2\hb_{\eps,b}(h_{\eps,b} - \hb_{\eps,b})^2e^{\nu_{\eps,b} + \lambda_{\eps,b}}
\endaligned
$$
while solving \eqref{dpartialrk} along the characteristics $\chi_{h_\eps,u_1}(u,r_1)$, gives
$$
\aligned
\del_rk_{\eps,b}(u_1,r_1) & = \del_r h_0(r(0)) e^{\int_{0}^{u_1} \Big( \frac{1}{2(r_2+ \eps)} \big((g_{\eps,b} - \gb_{\eps,b})\rho^\prime(\gb_{\eps,b})+(g_{\eps,b} - \gb_{\eps,b})\big) \Big)_{\chi_{\eps,u_1}} \, du_2} \\
& \quad+ \int_0^{u_1}[f_1(h_{\eps,b},k_{\eps,b})]_{\chi_{\eps,u_1}}e^{\int_{u_2}^{u_1} \Big( \frac{1}{2(r_3+ \eps)} \big((g_{\eps,b} - \gb_{\eps,b})\rho^\prime(\gb_{\eps,b})+(g_{\eps,b} - \gb_{\eps,b})\big) \Big)_{\chi_{\eps,u_1}} \, du_3} \, du_2.
\endaligned
$$
Similarly, it follows from \eqref{bound of rhb} and \eqref{bound uniform k} that there exists a constant $d_2> 0$, only depending on $(h_0,c_1,\eps, d_1)$, such that
\bel{bound uniform del_rk}
\sup_{(u,r) \in [0,u_1]\times [0,+ \infty)}|\del_rk_{\eps,b}|(u,r)\leq d_2.
\ee
Since $h_{\eps,b}:= \eta k_{\eps,b}$, from \eqref{bound uniform k} and \eqref{bound uniform del_rk} we conclude that $\|h_{\eps,b} \|_{C^1([0,u_1]\times [0,+ \infty))}$ is uniformly bounded with respect to $b$.

On the other hand, for all $b_1,\,b_2> 0$ sufficiently large we have 
 $$
\aligned
|\hb_{\eps,b_1} - \hb_{\eps,b_2}| & = \big |\frac{1}{r+ \eps} \int_0^r(h_{\eps,b_1} - h_{\eps,b_2}) \, ds\big| \leq \|h_{\eps,b_1} - h_{\eps,b_2} \|_{C^0}, 
\endaligned
$$
and
$$
\aligned
|g_{\eps,b_1} -g_{\eps,b_2}|
& \leq   4\pi c^2 \, (r+ \eps)^2 |(\hb_{\eps,b_1} + \hb_{\eps,b_2})|\|(\hb_{\eps,b_1} - \hb_{\eps,b_2})\|_{C^0} + 8\pi \, \big|1-4\pi^2 \, c^2 \, (r+ \eps)^2\hb_{\eps,b_1}^2\big|
\\
& \quad\times \bigg(\int_0^{+ \infty} \frac{|h_{\eps,b_1} - \hb_{\eps,b_1} +h_{\eps,b_1} - \hb_{\eps,b_2}|}{r+ \eps} \, ds\bigg)\|h_{\eps,b_1} - h_{\eps,b_2} \|_{C^0}
\\
& \leq   4\pi c^2 \, (r+ \eps)^2 |(\hb_{\eps,b_1} + \hb_{\eps,b_2})|\|(\hb_{\eps,b_1} - \hb_{\eps,b_2})\|_{C^0}
 + 8\sqrt{2} \pi \, \big|1-4\pi^2 \, c^2 \, (r+ \eps)^2\hb_1^2\big|
\\
& \quad\times \Big( \bigg(\int_{0}^{+ \infty} \hspace{-0.3cm} \frac{dr}{(r+ \eps)^2} \bigg)\bigg(\int_0^{+ \infty} \hspace{-0.3cm}(h_{\eps,b_1} - \hb_{\eps,b_1})^2+(h_{\eps,b_2} - \hb_{\eps,b_2})^2 \, dr\bigg) \Big)^{1/2} \|h_{\eps,b_1} - h_{\eps,b_2} \|_{C^0}
\\
& \leq  d_3 \|h_{\eps,b_1} - h_{\eps,b_2} \|_{C^0},
\endaligned
$$
where, by Proposition \ref{lemma-energy}, $d_3> 0$ is a constant depending on 
$h_0,c_1,\eps,d_1,\|h_{\eps,b_1} \|_{C^0},\|h_{\eps,b_2} \|_{C^0}$. In a similar way, we also obtain
$$
|\gb_{\eps,b_1} - \gb_{\eps,b_2}|,\,|\del_rg_{\eps,b_1} - \del_rg_{\eps,b_2}| \leq d_4 \, \|h_{\eps,b_1} - h_{\eps,b_2} \|_{C^0}, 
$$
for a constant $d_4> 0$ depending on $h_0,c_1,\eps,d_1,d_3,\|h_{\eps,b_1} \|_{C^0},\|h_{\eps,b_2} \|_{C^0}$. 

Therefore, taking into account that 
$\|h_{\eps,b} \|_{C^1([0,u_1]\times [0,+ \infty))}$ is uniformly bounded in $b$, an analysis similar to the one in 
the proof of Lemma \ref{lemma Psi cc} shows that there exist a subsequence $h_{\eps,b}$ and a limit function $h_\eps\in C^1{([0,c_1]\times [0,+ \infty))}$ such that, for any compact subset $[0,c_1]\times [0,r_0]$ with $r_0> 0$, we have
\bel{h_eb converges to h_b }
\lim_{b\to + \infty} \|h_{\eps,b} - h_\eps \|_{C^1([0,c_1]\times [0,r_0))}=0.
\ee
In particular, for all $r_0> 0$ we have 
$
\int_0^{r_0}h_{\eps,b}^2 \, dr \to  \int_0^{r_0}h_\eps^2 \, dr.
$
Then, it follows from the second inequality in Proposition \ref{lemma-energy} and from the bounds \eqref{B-bounds} that
$$
\aligned
\int_0^{r_0}h_\eps^2(u,r) \, dr& \leq  \int_{0}^{+ \infty}h_0^2 \, dr + \frac{u}{8\pi} + c^2 \, \big( \Hcal_\eps(u)+ B_I(c_1) \big).
\endaligned
$$
Now, consider the sequence of functions $\{f_{\eps,m} \}$  such that 
$
f_{\eps,m}:= 
\begin{cases} 
0, & \quad r> m,
\\
h_\eps^2 & \quad r \leq m.
\end{cases} 
$
Since $f_{\eps,m}$ is an increasing sequence of measurable functions satisfying
$$
\aligned
\int_0^{+ \infty}f_{\eps,m} \, dr& = \int_0^{m}h_\eps^2 \, dr 
\leq
 \int_{0}^{+ \infty}h_0^2 \, dr + \frac{u}{8\pi} + c^2 \, \big(\Hcal_\eps(u) + B_I(c_1) \big)
\endaligned 
$$
and, since $f_{\eps,m}$ converges pointwise to $h^2_\eps$, thanks to 
 the monotone convergence theorem we obtain 
  $h_\eps\in L^2([0,+ \infty))$ (at each $u$) and
\bel{energy h^2}
\aligned
\int_0^{+ \infty}h_\eps^2 \, dr& \leq  \int_{0}^{+ \infty}h_0^2 \, dr + \frac{u}{8\pi} + c^2 \, (\Hcal_\eps(u)
+B_I(c_1)).
\endaligned
\ee
Similarly, 
from \eqref{(h- hb)2} (in Appendix~\ref{Appendix-third} ), \eqref{energy h2} and the Cauchy-Schwarz inequality we also find 
\be
\aligned
\label{energy hb}
\int_0^{+ \infty} \hb_\eps^2 \, dr& \leq 2\bigg(\int_{0}^{+ \infty}h_0^2 \, dr + \int_0^{+ \infty}(h_0- \hb_0)^2 \, dr\bigg)+ \frac{34}{\pi} \int_0^u\bigg(\max\Big( \frac{c^2}{e^{2}}, M_\eps (0) \Big) \bigg) \, du_1 
\\
& \quad+c^2(4\Hcal_\eps(u)+ 2B_I(c_1)+B_{II}(c_1))+  \frac{1}{4\pi}(c_1+ 3u)
\endaligned
\end{equation}
and 
from Proposition~\ref{lemma-energy} 
\be
c^2\int_0^{+ \infty}(r+ \eps)^2\hb_\eps^2 \, dr\leq c^2\int_0^{+ \infty} \hb_\eps^2 \, dr + \frac{M_\eps (0)}{4\pi}e^{4\pi \, \int_0^{+ \infty}(h_\eps- \hb_\eps)^2 \, dr} .\label{energy rhb}
\end{equation}
Next, we will prove that 
\be
\nu_{\eps,b} + \lambda_{\eps,b} \to  \nu_\eps + \lambda_\eps:= -4\pi \int_r^{+ \infty} \frac{(h_\eps- \hb_\eps)^2}{s+ \eps} \, ds
\ee
uniformly on $[0,c_1]\times [0,+ \infty)$.
In fact, for each $\alpha> 0$, by Proposition \ref{lemma-energy-difference} together with \eqref{B-bounds}, \eqref{energy h^2} and \eqref{energy hb}, there exists $r_\alpha> 0$ such that for all $u\in [0,c_1]$
$$
\bigg|\int_{r_\alpha}^{+ \infty} \frac{(h_{\eps,b} - \hb_{\eps,b})^2}{s+ \eps} \, ds\bigg|,\,\bigg|\int_{r_\alpha}^{+ \infty} \frac{(h_\eps- \hb_\eps)^2}{s+ \eps} \, ds\bigg| \leq \frac{\alpha}{3}.
$$
On the other hand, since $h_{\eps,b}$ converges uniformly to $h_\eps$ on any compact subset $[0,c_1]\times [0,r_0]$ with $r_0> 0$, we obtain  for all $(u,r) \in [0,c_1]\times [0,r_\alpha]$ and for all $b \geq 1$ sufficiently large
$$
\bigg|\int_{r}^{r_\alpha} \frac{(h_{\eps,b} - \hb_{\eps,b})^2}{s+ \eps} \, ds- \int_{r}^{r_\alpha} \frac{(h_\eps- \hb_\eps)^2}{s+ \eps} \, ds\bigg| \leq \frac{\alpha}{3}.
$$
Combining these two inequalities together, for all $(u,r) \in [0,c_1]\times [0,+ \infty)$ we find 
$$
\aligned
&
\bigg|\int_{r}^{+ \infty} \frac{(h_{\eps,b} - \hb_{\eps,b})^2}{s+ \eps} \, ds- \int_{r}^{+ \infty} \frac{(h_\eps- \hb_\eps)^2}{s+ \eps} \, ds\bigg| 
\\
& \leq  \bigg|\int_{\min\{r,r_\alpha\}}^{r_\alpha} \hspace{-0.3cm} \frac{(h_{\eps,b} - \hb_{\eps,b})^2}{s+ \eps} \, ds- \int_{\min\{r,r_\alpha\}}^{r_\alpha} \hspace{-0.3cm} \frac{(h_\eps- \hb_\eps)^2}{s+ \eps} \, ds\bigg|
  + \bigg|\int_{r_\alpha}^{+ \infty} \frac{(h_{\eps,b} - \hb_{\eps,b})^2}{s+ \eps} \, ds\bigg|+ \bigg|\int_{r_\alpha}^{+ \infty} \frac{(h_\eps- \hb_\eps)^2}{s+ \eps} \, ds\bigg|
\\
& \leq  \alpha.
\endaligned
$$
Since $\alpha> 0$ is arbitrarily chosen, this means that 
$$\int_r^{+ \infty} \frac{(h_{\eps,b} - \hb_{\eps,b})^2}{s+ \eps} \, ds\to\int_r^{+ \infty} \frac{(h_\eps- \hb_\eps)^2}{s+ \eps} \, ds$$
uniformly on $[0,c_1]\times [0,+ \infty)$ and, hence, our claim follows.
We now recall that 
$
g_\eps:= \big(1-4\pi c^2(r+ \eps)^2\hb_\eps^2\big) e^{\nu_\eps + \lambda_\eps}.
$ 
Since 
$(r+ \eps) \, \hb_{\eps,b} \to  (r+ \eps) \, \hb_\eps$
and
$e^{\nu_{\eps,b} + \lambda_{\eps,b}} \to e^{\nu_\eps + \lambda_\eps}$
uniformly on any compact subset $[0,c_1]\times [0,r_0]$ with $r_0> 0$, we obtain 
\bel{g_eb converges g_e}
\aligned
g_{\eps,b}& \to  g_\eps, 
\qquad 
\gb_{\eps,b} \to  \gb_\eps 
\qquad \text{ uniformly on $[0,c_1]\times [0,r_0]$. } 
\endaligned
\ee
On the other hand, given $(u,r) \in [0,c_1]\times [0,+ \infty)$, for all $b>r$, we have
$$
\aligned
& D_{h_{\eps,b}}h_{\eps,b}= \frac{(g_{\eps,b} - \gb_{\eps,b})}{2(r+ \eps)}(h_{\eps,b} - \hb_{\eps,b}) - \frac{c^2}{2}(r+ \eps) \, \hb_{\eps,b}e^{\nu_{\eps,b} + \lambda_{\eps,b}},
\\
& h_{\eps,b}(0,r) =h_0(r).
\endaligned
$$

Therefore, taking the limit $b\to + \infty$ in this equation, we obtain by \eqref{h_eb converges to h_b } and \eqref{g_eb converges g_e} that
\begin{equation*}
\aligned
\label{Dheps-equation}
D_\eps h_\eps& = \frac{(g_\eps- \gb_\eps)}{2(r+ \eps)}(h_\eps- \hb_\eps) - \frac{c^2}{2}(r+ \eps) \, \hb_\eps \, e^{\nu_\eps+ \lambda_\eps},
\qquad
h_\eps(0,r)  =h_0(r).
\endaligned
\end{equation*}
In other words, $h_\eps$ is a classical solution of the regularized Einstein problem \eqref{epsilon self-gravitating eq} on $[0,c_1]\times [0,+ \infty)$. 


\subsection{Implications of the Hawking mass}
\label{section---52}

\paragraph{Evolution of the  energy.}

We continue the proof of Theorem~\ref{lemma-global epsilon existence} by deriving
 the $\eps$-Hawking mass evolution equation \eqref{evol-eps-mass} as well as the integral identity 
 \eqref{integral identity}.  
 In order to do that we first need the following result which is non-trivial due to the different fall-off properties of $h_\eps$ and $\hb_\eps$.
 
\begin{lemma}[Evolution of the regularized energy]
\label{lem-swap}
For each $(u_1,r) \in [0,c_1]\times [0,+ \infty)$ one has  
\bel{derivative under integral}
\del_u\bigg(\int_{r}^{+ \infty} \frac{(h_\eps- \hb_\eps)^2}{s+ \eps} \, ds\bigg) (u_1, \cdot)
=
\int_{r}^{+ \infty} \del_u\bigg(\frac{(h_\eps- \hb_\eps)^2}{s+ \eps} \bigg)(u_1, \cdot)
 \, ds.
\ee
\end{lemma}

\begin{proof}[Proof of Lemma \ref{lem-swap}]
By  H\"{o}lder inequality we obtain
\be
\aligned
\label{derivative under integral 1}
\bigg|\int_{r}^{+ \infty} \del_u\bigg(\frac{(h_\eps- \hb_\eps)^2}{s+ \eps} \bigg) \, ds\bigg| 
&=  2\bigg|\int_r^{+ \infty} \frac{(h_\eps- \hb_\eps)}{s+ \eps} \del_u(h_\eps- \hb_\eps)  ds\bigg|
\\
& \leq 2 \, \|\del_u(h_\eps- \hb_\eps)\|_{C^0} \bigg(\int_r^{+ \infty}(h_\eps- \hb_\eps)^2ds\bigg)^{1/2} \bigg(\int_r^{+ \infty} \frac{ds}{(s+ \eps)^2} \bigg)^{1/2}.
 \endaligned
\end{equation}
Similarly, we also have
$$
 \aligned
&
\bigg|  \del_u \bigg(\int_{r}^{+ \infty} \frac{(h_\eps- \hb_\eps)^2(u_1,s)}{s+ \eps} \, ds\bigg)\bigg|
 \leq  \limsup_{u\to u_1} \int_r^{+ \infty} \bigg|\frac{(h_\eps- \hb_\eps)^2(u,s) -(h_\eps- \hb_\eps)^2(u_1,s)}{(s+ \eps)(u-u_1)} \bigg| \, ds
\\
& =  \limsup_{u\to u_1} \int_r^{+ \infty} \bigg|\frac{\big((h_\eps- \hb_\eps)(u,s)+(h_\eps- \hb_\eps)(u_1,s)\big)\big((h_\eps- \hb_\eps)(u,s) -(h_\eps- \hb_\eps)(u_1,s)\big)}{(s+ \eps)(u-u_1)} \bigg| \, ds.
\endaligned
$$
Therefore, by the mean value theorem there exists $u_s$ between $u$ and $u_1$ such that
$$
\aligned
&
\bigg|\del_u\bigg(\int_{r}^{+ \infty} \frac{(h_\eps- \hb_\eps)^2(u_1,s)}{s+ \eps} \, ds\bigg)\bigg|
  =  \limsup_{u\to u_1} \int_r^{+ \infty} \bigg|
  \big((h_\eps- \hb_\eps)(u,s)+(h_\eps- \hb_\eps)(u_1,s)\big)\del_u(h_\eps- \hb_\eps)(u_s,s) \bigg| \, {ds \over s+ \eps} 
\endaligned
$$
hence, by H\"{older} and Cauchy-Schwarz inequalities, 
\be
\label{derivative under integral 2}
\aligned
& 
\bigg|\del_u\bigg(\int_{r}^{+ \infty} \frac{(h_\eps- \hb_\eps)^2(u_1,s)}{s+ \eps} \, ds\bigg)\bigg|
\\
&
 \leq 
\|\del_u(h_\eps- \hb_\eps)\|_{C^0} \limsup_{u\to u_1}  \Big( \bigg(\int_r^{+ \infty} \big((h_\eps- \hb_\eps)(u,s)
 +(h_\eps- \hb_\eps)(u_1,s)\big)^2ds\bigg)^{1/2} \bigg(\int_r^{+ \infty} \frac{ds}{(s+ \eps)^2} \bigg)^{1/2}  \Big) 
\\
& \leq  2 \, \|\del_u(h_\eps- \hb_\eps)\|_{C^0} \limsup_{u\to u_1}  \Big( \bigg(\int_r^{+ \infty} \big((h_\eps- \hb_\eps)^2(u,s)
  +(h_\eps- \hb_\eps)^2(u_1,s)\big) \, ds\bigg)^{1/2} \bigg(\int_r^{+ \infty} \frac{ds}{(s+ \eps)^2} \bigg)^{1/2}  \Big).
\endaligned
\end{equation}
Therefore, for any $\alpha> 0$, in view of \eqref{energy h^2}, \eqref{energy rhb}, \eqref{derivative under integral 1}, and \eqref{derivative under integral 2} we see 
that there exists a sufficiently large $r_\alpha> 0$ such that
$$
\bigg|\int_{r_\alpha}^{+ \infty} \del_u\bigg(\frac{(h_\eps- \hb_\eps)^2(u_1,s)}{s+ \eps} \bigg) \, ds\bigg|,\,
\quad
\bigg|\del_u\bigg(\int_{r_\alpha}^{+ \infty} \frac{(h_\eps- \hb_\eps)^2(u_1,s)}{s+ \eps} \, ds\bigg)\bigg| \leq \frac{\alpha}{2},
$$
which, combined with 
$$
\int_{r_1}^{r_\alpha} \del_u\bigg(\frac{(h_\eps- \hb_\eps)^2(u_1,s)}{s+ \eps} \bigg) \, ds = \del_u\bigg(\int_{r_1}^{r_\alpha} \frac{(h_\eps- \hb_\eps)^2(u_1,s)}{s+ \eps} \, ds\bigg),
$$
yields us
 $$
\aligned
&
\bigg|\int_{r_1}^{+ \infty} \del_u \bigg(\frac{(h_\eps- \hb_\eps)^2 (u_1,s)}{s+ \eps} \bigg) \, ds- \del_u\bigg(\int_{r_1}^{+ \infty} \frac{(h_\eps- \hb_\eps)^2(u_1,s)}{s+ \eps} \, ds\bigg)\bigg|
\\ 
& =  \bigg|\int_{r_\alpha}^{+ \infty} \del_u\bigg(\frac{(h_\eps- \hb_\eps)^2(u_1,s)}{s+ \eps} \bigg) \, ds- \del_u\bigg(\int_{r_\alpha}^{+ \infty} \frac{(h_\eps- \hb_\eps)^2(u_1,s)}{s+ \eps} \, ds\bigg)\bigg|
\\
 & \leq \bigg|\int_{r_\alpha}^{+ \infty} \del_u\bigg(\frac{(h_\eps- \hb_\eps)^2(u_1,s)}{s+ \eps} \bigg) \, ds\bigg|+ \bigg|\del_u\bigg(\int_{r_\alpha}^{+ \infty} \frac{(h_\eps- \hb_\eps)^2(u_1,s)}{s+ \eps} \, ds\bigg)\bigg|
\leq  \alpha.
\endaligned
$$
Since $\alpha> 0$ is arbitrary, \eqref{derivative under integral} follows as claimed.
\end{proof}


\paragraph{Derivation of the integral identity.} 

From the previous result and arguing as in the proof of Lemma \ref{lemma gb> 0}, we obtain the $\eps$-Hawking mass evolution equation  \eqref{evol-eps-mass} (cf.~also \eqref{ep equation mass}):
\bel{evolution of m}
D_\eps m_\eps= -  \frac{4\pi}{e^{\nu_\eps + \lambda_\eps}}(r+ \eps)^2\big(D_\eps \hb_\eps\big)^2- c^2\pi(r+ \eps)^2\hb_\eps^2\gb_\eps,
\ee
and
$\gb_\eps> 0$
for each $(u_1,r) \in [0,c_1]\times [0,+ \infty)$. 
Furthermore, recalling that
$
\frac{e^{\nu_\eps + \lambda_\eps}}{\gb_\eps}= \Big(1- \frac{2m_\eps}{r+ \eps} \Big)^{-1},
$
from \eqref{evolution of m}, together with $\del_r \gb_\eps=(g_\eps- \gb_\eps)/(r+ \eps)$ and \eqref{g-equation}, we compute
\be
\aligned
\label{evolution of 2lambda}
D_\eps \Big(\frac{e^{\nu_\eps + \lambda_\eps}}{\gb_\eps} \Big) & =  - \frac{2\pi e^{\nu_\eps + \lambda_\eps}}{\gb_\eps^2} \frac{\xi^2_\eps}{r+ \eps} - \frac{2\pi c^2  e^{2(\nu_\eps + \lambda_\eps)}}{\gb_\eps}(r+ \eps) \, \hb_\eps^2+ \frac{e^{2(\nu_\eps + \lambda_\eps)}}{\gb_\eps} \frac{m_\eps}{(r+ \eps)^2},\\
& =  -2\pi \, \frac{e^{\nu_\eps + \lambda_\eps}}{\gb_\eps^2} \frac{\xi^2_\eps}{(r+ \eps)} + \frac{1}{2} \del_{r} \gb_\eps \Big(\frac{e^{\nu_\eps + \lambda_\eps}}{\gb_\eps} \Big),
\endaligned
\end{equation}
where we used the notation \eqref{definition of xi}.

Let $\delta>0$ be a fixed real number.  By integratation by parts, we then obtain 
$$
\aligned
&
D_\eps \Big(\int_{\delta}^{r} \frac{e^{\nu_\eps + \lambda_\eps}}{\gb_\eps} \, dr\Big) 
 =  \int_{\delta}^{r} \del_u\Big(\frac{e^{\nu_\eps + \lambda_\eps}}{\gb_\eps} \Big) \, ds- \frac{1}{2}e^{\nu_\eps + \lambda_\eps} 
 =  \int_{\delta}^r\Big(D_\eps \Big(\frac{e^{\nu_\eps + \lambda_\eps}}{\gb_\eps} \Big)+ \frac{1}{2} \gb_\eps \del_{r} \Big(\frac{e^{\nu_\eps + \lambda_\eps}}{\gb_\eps} \Big)\Big) \, ds
- \frac{1}{2}e^{\nu_\eps + \lambda_\eps} 
\\&
 =  \int_{\delta}^r\Big(D_\eps \Big(\frac{e^{\nu_\eps + \lambda_\eps}}{\gb_\eps} \Big) - \frac{1}{2} \del_{r} \gb_\eps \Big(\frac{e^{\nu_\eps + \lambda_\eps}}{\gb_\eps} \Big)\Big) \, ds
- \frac{1}{2}e^{(\nu_\eps + \lambda_\eps)(u,\delta)},
\endaligned
$$
which, in combination with \eqref{evolution of 2lambda}, gives us 
\be
D_\eps \Big(\int_{\delta}^{r} \frac{e^{\nu_\eps + \lambda_\eps}}{\gb_\eps} \, dr\Big) = - 2\pi \, \int_{\delta}^{r} \frac{e^{\nu_\eps + \lambda_\eps}}{\gb_\eps^2} \frac{\xi^2_\eps}{r+ \eps} \, ds- \frac{1}{2}e^{(\nu_\eps + \lambda_\eps)(u,\delta)}.
\ee
Integrating this equation along the characteristics $\chi_\eps(\cdot; u_1,r_1)$, we obtain
the integral identity  \eqref{integral identity}.


\subsection{Global existence of $\eps$-regularized solutions}
\label{section53-} 

\paragraph{Final argument for the existence theory.}
 
We now complete the proof of Theorem~\ref{lemma-global epsilon existence}. 
While our domain of existence until now is limited by $c_1$, since $c_1$ was chosen in a way independent of the initial data $(h_0,\eps)$,
it is straightforward to repeat our construction on successive time slabs and eventually cover the whole spacetime. Hence, 
we conclude that $h_\eps$ is a (classical) global solution of 
\eqref{epsilon self-gravitating eq} satisfying
\be
\gb_\eps> 0, \qquad (u,r) \in [0,+ \infty)\times [0,+ \infty).
\ee
Moreover, since \eqref{derivative under integral} holds for all $(u,r) \in [0,+ \infty)\times [0,+ \infty)$, then arguing as in the proof of Lemma 
\ref{lemma gb> 0} and as in the proof of equation \eqref{control rhb in C0}, we see
that $h_\eps$ satisfies \eqref{evolution of m} and \eqref{equa-1234} on $[0,+ \infty)\times [0,+ \infty)$. 
 
We have established Corollary \ref{coroll-thr-21} and the first two assertions of Theorem \ref{lemma-global epsilon existence}. 
We will now prove assertions (3), (4), and (5), i.e. that $h_\eps$ satisfies the energy-type inequalities \eqref{energy-int}.
The missing ingredient in order to finalize this proof are the following lemmas. 

\begin{lemma}[Pointwise estimates for the matter field] 
\label{lemma-uniform bound of h and delr h with e} 
At each time $u \geq 0$ one has 
$$
\sup_{r\geq \delta}|h_\eps(u,r)| \leq B_5(u,\delta), 
\qquad
 \quad \sup_{r\geq \delta}|\del_r h_\eps(u,r)| \leq B_8(u,\delta), 
$$
where $B_5$ and $B_8$ are upper bound defined in the course of the proof, namely in~\eqref{B5eq} and \eqref{B8eq}, below.
\end{lemma}

\begin{lemma}[Decay of the matter field at null infinity]
\label{prop-lims}
For any time $u_0> 0$ one has 
\bel{decay uniformly in u}
\lim_{r\to + \infty} \bigg(\sup_{u\in [0,u_0]}|h_\eps|\bigg) = \lim_{r\to + \infty} \bigg(\sup_{u\in [0,u_0]}|\del_r h_\eps|\bigg) =0.
\ee
\end{lemma}


\paragraph{Pointwise estimates for the matter field.}

We now give a proof of Lemma~\ref{lemma-uniform bound of h and delr h with e}.
By induction, it is sufficient to prove \eqref{decay uniformly in u} with
$u_0= \frac{3}{8}c^2$.
Solving \eqref{epsilon self-gravitating eq} along the characteristics $\chi_{h_\eps,u_1}(u,r_1)$, we find
$$
\aligned
h_\eps(u_1,r(u_1)) 
& = h_0(r(0)) e^{\frac{1}{2} \int_{0}^{u_1} \Big( \frac{g_\eps- \gb_\eps}{r_2+ \eps}  \Big)_{\chi_{\eps,u_1}} du_2} \\
& \quad - \int_0^{u_1} \Big(\frac{1}{2(r_2+ \eps)}(g_\eps- \gb_\eps) \, \hb_\eps + \frac{c^2}{2}(r_1+ \eps) e^{\nu_\eps + \lambda_\eps} \hb_\eps \Big)_{\chi_{\eps,u_1}}e^{\frac{1}{2} \int_{u_2}^{u_1} \Big( \frac{g_\eps- \gb_\eps}{r_3+ \eps}  \Big)_{\chi_{\eps,u_1}} du_3} \, du_2.
\endaligned
$$
Observe that from \eqref{ep bondi mass} and \eqref{mass-bound}
\bel{g-gb}
g_\eps- \gb_\eps \leq e^{\nu_\eps + \lambda_\eps} - \gb_\eps \leq \frac{2M_\eps(u)}{r+ \eps}e^{(\nu_\eps + \lambda_\eps)(u,r)} \leq \frac{2M_\eps(0)}{r+ \eps}e^{(\nu_\eps + \lambda_\eps)(u,r)}
\ee
and then, if $r(u_1)\geq \delta$, we have from the definition of $\eps$-Hawking mass, from (2) in \eqref{equa-1234}
and \eqref{g-gb} that 
 $$
\aligned
 \int_{0}^{u_1} \Big(\frac{g_\eps- \gb_\eps}{r_2+ \eps} \Big)_{\chi_{\eps,u_1}} du_2
&\leq B_3(u_1,\delta),
\\
 \int_0^{u_1} \Big(\frac{g_\eps- \gb_\eps}{2(r_2+ \eps)} \hb_\eps + \frac{c^2}{2}(r_2+ \eps) \, \hb_\eps \, e^{\nu_\eps + \lambda_\eps} \Big)_{\chi_{\eps,u_1}}e^{\frac{1}{2} \int_{u_2}^{u_1} \Big( \frac{g_\eps- \gb_\eps}{r_3+ \eps}  \Big)_{\chi_{\eps,u_1}} du_3} \, du_2
&\leq B_4(u_1,\delta),
\endaligned
$$
where
$$
\aligned
B_3(u_1,\delta):& =  \frac{2u_1 M_\eps(0)}{\delta^2},\\
B_4(u_1,\delta):& = u_1e^{\frac{B_3(u_1,\delta)}{2}} \Big(\frac{M_\eps(0)}{\delta^3} + \frac{c^2}{2} \Big) \Acal_\eps(u_1,\delta).
\endaligned
$$
Taking this into account, for all $(u_1,r_1) \in [0,u_0]\times [\delta,+ \infty)$, we get
\bel{control h in C0}
|h_\eps|(u_1,r_1)\leq B_5(u_1,\delta),
\ee
with 
\bel{B5eq}
B_5(u_1,\delta) := e^{\frac{B_3(u_1,\delta)}{2}} \|h_0\|_{C^0[r,+ \infty)} +B_4(u_1,\delta).
\ee
Therefore, an analysis similar to the proof of Proposition~\ref{lemma-energy} (with $b=+ \infty$) shows that, for all $u_1\in [0,u_0]$,
\bel{lemma-energy in general}
\int_0^{+ \infty}h_\eps^2(u_1,r) \, dr
+
 c^2\int_0^{+ \infty}(r+ \eps)^2\hb_\eps^2(u_1,r) \, dr \leq B_6(u_1,\delta),
\ee
where $B_6(u_1,\delta) > 0$ is independent of $\eps$.

Now, differentiating \eqref{epsilon self-gravitating eq} with respect to $r$, we obtain 
\begin{equation*}
\label{evolution of del_rh}
D_\eps(\del_r h_\eps) - \frac{g_\eps- \gb_\eps}{r+ \eps} \del_r h_\eps= \frac{h_\eps- \hb_\eps}{2(r+ \eps)} \del_rg_\eps- \frac{3(g_\eps- \gb_\eps)}{2(r+ \eps)^2}(h_\eps- \hb_\eps) - \frac{c^2}{2}e^{\nu_\eps + \lambda_\eps} \big(h_\eps+4\pi \hb_\eps(h_\eps- \hb_\eps)^2\big)
\end{equation*}
and solving this equation along the characteristics $\chi_{\eps,u_1}:=\chi_\eps(u_1; u_1, r_1)$, we find
\be
\aligned
\label{derivative eq}
\del_r h_\eps(u_1,r(u_1))  
& = \del_r h_0(r(0)) e^{\int_{0}^{u_1} \Big( \frac{g_\eps- \gb_\eps}{r_2+ \eps}  \Big)_{\chi_{\eps,u_1}} du_2} 
\\
& \quad - \frac{c^2}{2} \int_0^{u_1} \Big(e^{\nu_\eps + \lambda_\eps} \big(h_\eps+4\pi \hb_\eps(h_\eps- \hb_\eps)^2\big)\Big)_{\chi_{\eps,u_1}}e^{\int_{u_2}^{u_3} \Big( \frac{g_\eps- \gb_\eps}{r_2+ \eps}  \Big)_{\chi_{\eps,u_1}} du_3} \, du_2 
\\
& \quad + \frac{1}{2}\int_0^{u_1} \Big(\frac{h_\eps- \hb_\eps}{r_2+ \eps} \del_rg_\eps- \frac{3(g_\eps- \gb_\eps)}{(r_2+ \eps)^2}(h_\eps- \hb_\eps)\Big)_{\chi_{\eps,u_1}}e^{\int_{u_2}^{u_1} \Big( \frac{g_\eps- \gb_\eps}{r_3+ \eps}  \Big)_{\chi_{\eps,u_1}} du_3} \, du_2.
\endaligned
\end{equation}
For all $r\geq \delta$, from (2) in \eqref{equa-1234}
and \eqref{control h in C0} we deduce that 
\bel{control del_rg}
\aligned
|\del_rg_\eps|(u,r) & = e^{\nu_\eps + \lambda_\eps} \Big(4\pi(1-4\pi c^2(r+ \eps)^2\hb_\eps^2)\frac{(h_\eps- \hb_\eps)^2}{r+ \eps} -8\pi c^2 \, (r+ \eps) \, \hb_\eps h_\eps\Big)
 \leq  B_7(u,\delta),
\endaligned
\ee
where
$$
B_7(u,\delta) := 4\pi \, \bigg( \frac{1+4\pi c^2 \Acal^2_\eps}{\delta}\left (B_5+ \frac{\Acal_\eps}{\delta} \right )^2+2 c^2B_5\Acal_\eps\bigg)(u,\delta).
$$
Taking into account (2) in \eqref{equa-1234},
as well as
\eqref{g-gb}, \eqref{control h in C0} and \eqref{control del_rg} into \eqref{derivative eq}, for all $(u_1,r_1) \in [0,u_0]\times [\delta,+ \infty)$ we obtain  
$
|\del_r h_\eps|(u_1,r_1)\leq B_8(u_1,\delta),
$
where
\bel{B8eq}
\aligned
B_8(u_1,\delta): =  e^{B_3(u_1,\delta)} \|h_0\|_{C^1}
+ &u_1e^{B_3(u_1,\delta)} \bigg(\frac{c^2}{2}B_5+ \frac{2\pi c^2 \Acal_\eps}{\delta}\left(B_5+ \frac{\Acal_\eps}{\delta}\right)^2 \bigg.\\& \quad\bigg.+ \left.\left(B_5+ \frac{\Acal_\eps}{\delta} \right)\left(\frac{B_7}{2\delta} + \frac{3 M_\eps(0) }{\delta^3} \right)\right)(u_1,\delta).
\endaligned
\ee


\paragraph{Decay of the matter field at null infinity.}

We now give a proof of Lemma~\ref{prop-lims}. From the previous lemma, it follows from \eqref{epsilon self-gravitating eq}, 
the inequality (2) in \eqref{equa-1234}
and \eqref{control h in C0} that, for all $(u,r) \in [0,u_0]\times [\delta,+ \infty)$, 
\bel{control del_uh}
|\del_u h_\eps|(u,r)\leq B_9(u,\delta),
\ee
with
$$
B_9(u,\delta) := \bigg( \frac{B_8}{2} + \frac{M_\eps(0)}{\delta^2} \Big(B_5+ \frac{\Acal_\eps}{\delta} \Big)+ \frac{c^2}{2} \Acal_\eps\bigg)(u,\delta).
$$
Now, set 
$
\alpha_\eps:= \limsup_{r\to + \infty} \big(\sup_{u\in[0,u_0]} |\del_r h_\eps| \big)\geq 0$
and suppose that $\alpha_\eps \in (0,1)$. We are going to exclude this case by proceeding by contradiction.
Due to \eqref{lemma-energy in general}, for each $u$, there exists $r_u$ sufficiently large such that for all $r\geq r_u$
\bel{r_u}
\int_{r}^{r+1}h_\eps^2ds\leq \alpha_\eps^2. 
\ee
Define
$
r_u': =  \inf\Big\{\mbox{$r_u\in[1,+ \infty)$ satisfies \eqref{r_u}} \Big\}$
and $r_\infty^\prime: =  \sup_{u\in[0,u_0]} r_u^\prime$. 
We now claim that
$
r^\prime_\infty<+ \infty.
$ 
Indeed, if this is not true, then there exists a sequence $(u_i,r_i)$ converging to $(u_\infty,+ \infty)$ such that
$
\Big|\int_{r_i}^{r_i+1}h_\eps^2(u_i,r) \, dr\Big|\geq \alpha_\eps^2.
$
For any given $\zeta> 0$, taking $u_*\in [0,u_0]$ satisfying $|u_*-u_\infty| \leq \zeta\alpha_\eps^2$, we have by the mean value theorem that 
\bel{ri-eq}
\int_{r_i}^{r_i+1} \big(h_\eps^2(u_i,r) - h_\eps^2(u_*,r)\big) \, dr=h_\eps^2(u_i,r_i^\prime) - h_\eps^2(u_*,r_i^\prime) =2(u_i-u_*)h_\eps(u_i^\prime,r_i^\prime)\del_u h_\eps(u_i^\prime,r_i^\prime),
\ee      
where $u_i^\prime$ (resp. $r_i^\prime$) is between $u_*$ and $u_i$ (resp. $r_i$ and $r_i+1$). Then, it follows from \eqref{control h in C0}, \eqref{lemma-energy in general}, and \eqref{control del_uh} that
$$
\aligned
\alpha_\eps^2\leq \lim\Big(\int_{r_i}^{r_i+1}h_\eps^2(u_i,r) \, dr\Big)& \leq  \limsup\Big(2|u_*-u_i|h_\eps(u_i^\prime,r_i^\prime)\del_u h_\eps(u_i^\prime,r_i^\prime)+ \int_{r_i}^{r_i+1}h_\eps^2(u_*,r) \, dr\Big)
\\
& \leq  2 \, \zeta \, \|h_\eps\|_{C^1}^2\alpha_\eps^2 \leq  \frac{\alpha^2_\eps}{2},
\endaligned
$$
provided that $\zeta\leq \frac{1}{4} \|h_\eps\|_{C^1}^{-2}$, which is a contradiction and, hence, $r_\infty^\prime<+ \infty$ as claimed.

We note, furthermore, that for all $r\geq r^\prime_\infty$  
\bel{control decay of h}
\aligned
|h_\eps|(u,r) & = \bigg|\int_r^{r+1}h_\eps(u,s) \, ds - \frac{1}{2} \del_r h_\eps(u,r^\prime)\bigg|
 \leq \bigg|\int_r^{r+1}h_\eps(u,s) \, ds \bigg|+ \frac{1}{2}|\del_r h_\eps|(u,r^\prime)\\
& \leq  \bigg(\int_r^{r+1}h_\eps^2(u,s) \, ds\bigg)^{1/2}
+ \frac{1}{2}|\del_r h_\eps|(u,r^\prime) \leq  \alpha_\eps+ \frac{1}{2}|\del_r h_\eps|(u,r^\prime),
\endaligned
\ee
where we used (by Taylor expansion) $r^\prime\in[r,r+1]$ and 
the H\"{o}lder inequality as well as \eqref{r_u}.

Next, let $\{u_i,r_i\} \in [0,u_0]\times [1,+ \infty)$ be a sequence such that
$r_i\to + \infty$ and
$
\lim|\del_r h_\eps|(u_i,r_i) = \alpha_\eps.
$
Along the characteristics $\chi_\eps(\cdot; u_i, r_i)$, using (2) in  \eqref{equa-1234},
\eqref{g-gb}, \eqref{control h in C0}, \eqref{control del_rg} and \eqref{control decay of h}, we obtain thanks to \eqref{derivative eq} that
$$
\aligned
 |\del_r h_\eps|(u_i,r_i) 
 & \leq  e^{B_3(u_0,r_i)}|\del_r h_0(r(0))|
\\
& \quad + u_0e^{B_3(u_0,r_i)} \left(B_5+ \frac{\Acal_\eps}{r} \right)\Big(\frac{2\pi c^2 \Acal_\eps}{r} \Big(B_5+ \frac{\Acal_\eps}{r} \Big)+ \frac{B_7}{2r} + \frac{3M_\eps(0)}{r^3} \Big)(u_0,r_i)
\\
& \quad + \frac{c^2}{2}u_0e^{B_3(u_0,r_i)} \Big(\alpha_\eps+ \frac{1}{2} 
\sup_{u\leq u_i,\,r_i\leq r} |\del_r h_\eps(u,r) | \Big). 
\endaligned
$$
Taking $i\to + \infty$, it now follows from the definition of $\alpha_\eps$ and our choice $u_0= \frac{3}{8}c^2$ that
$$
\alpha_\eps= \lim|\del_r h_\eps|(u_i,r_i)\leq \frac{3}{4}c^2u_0 \alpha_\eps
\leq \frac{1}{2} \alpha_\eps, 
$$
which is a contradiction. Therefore, we have $\alpha_\eps=0$. 

Now, let $\{r_i\}$ be a sequence converging to $+ \infty$. If there exists a sequence $u_i\in [0,u_0]$ such that
$
\lim\int_{r_i}^{r_i+1}h_\eps^2(u_i,r) \, dr= \zeta> 0,
$
then, for any $u_*\in [0,u_0]$, we can use \eqref{ri-eq},
 \eqref{control h in C0}, \eqref{lemma-energy in general}, and \eqref{control del_uh} and obtain 
$$
\aligned
\zeta= \lim\Big(\int_{r_i}^{r_i+1}h_\eps^2(u_i,r) \, dr\Big)& \leq  \limsup\Big(2|u_*-u_i|h_\eps(u_i^\prime,r_i^\prime)\del_u h_\eps(u_i^\prime,r_i^\prime)+ \int_{r_i}^{r_i+1}h_\eps^2(u_*,r) \, dr\Big)\\
& \leq  |u_*-u_i| \|h_\eps\|_{C^1}^2.
\endaligned
$$
Since (after passing to a subsequence) $u_i$ converges to $u_\infty\in [0,u_0]$, this is a contradiction as long as 
$|u_*-u_\infty| \leq (\zeta/2) \|h_\eps\|_{C^1}^{-2}$. 
Therefore, we must have 
\bel{int h2 at infitity}
\lim{\Big(\sup_{u\in[0,u_0]} \int_{r_i}^{r_i+1}h_\eps^2(u,r) \, dr\Big)}=0.
\ee   
In particular, for each $u$ there exists a sequence $\{r^\prime_i(u)\}$ such that
$r_i^\prime(u) \in [r_i,r_i+1]$ and 
$\lim|h_\eps|(u,r_i^\prime(u)) =0$.
Combining this result with the fact that $\alpha_\eps=0$, for any sequence $\{u_i\} \subset [0,u_0]$, we obtain by Taylor expansion that
$$
\lim|h_\eps|(u_i,r_i)\leq \lim|h_\eps|(u_i,r_i^\prime(u))+(r_i^\prime(u) -r_i)\lim|\del_r h_\eps|(u_i,r_i^{\prime\prime}(u)) =0,
$$
for some $r_i^{\prime\prime}(u) \in [r_i,r_i'(u)]$. This means that
$
\lim_{r\to + \infty} \big(\sup_{u\in[0,u_0]}|h_\eps|\big) =0$, 
which completes the proof of \eqref{decay uniformly in u}.


\paragraph{Closing the argument for the $\eps$-regularized problem.}

Then, in view of \eqref{decay uniformly in u} and the inequality (2) in \eqref{equa-1234}, 
an analysis similar to the one in the proof of Proposition~\ref{lemma-energy} (with $b=+ \infty$) shows that $h_\eps$ satisfies the inequalities \eqref{energy-int}. This completes the proof of Theorem \ref{lemma-global epsilon existence}.


\section{Global estimates for generalized solutions} 
\label{sec-section6}
 
\subsection{Pointwise estimates for the matter field}
\label{sec-section--61}

\paragraph{Objectives: first properties of the limit.}

We have established the global existence of a classical solution to the regularized problem \eqref{epsilon self-gravitating eq} and, by letting $\eps$ tend to zero, we will now arrive at a proof of Theorem \ref{thrgensol}, that is, 
the global existence of generalized solutions to the reduced Einstein equation \eqref{self-gravitating eq}.
 A crucial aspect of our strategy is the analysis of the evolution of the Hawking mass ---which was not necessary at this stage for massless scalar fields. 
 Indeed, we must adapt the arguments in \cite[Section 5]{Chr2} and suitably take the effect of 
the massive field into account. Given $u_0> 0$, for any $r_0>\delta > 0$ we use the notation
\be
\aligned
 \Ocal (u_0) :& =  \big\{ (u,r) \, / \, 0\leq u\leq u_0, \, r \in (0, +\infty) \big\},
\\
 \Ocal_\delta (u_0) :& =  \big\{ (u,r) \, / \, 0\leq u\leq u_0, \, r \in (\delta, +\infty) \big\},
\\ 
\Ocal_{\delta, r_0} (u_0) :& =  \big\{ (u,r) \, / \, 0\leq u\leq u_0, \, r \in (\delta, r_0) \big\}.
\endaligned
\ee
By taking all of the available estimates into account, we see that there exists $h\in C^1(\Ocal)$ such that (after passing to a subsequence)
 the following properties hold:   
\be
\aligned
& \text{$\hb_\eps,\,(r+ \eps) \, \hb_\eps, h_\eps,\,\del_r h_\eps$ converge uniformly on compact subsets $\Ocal_{\delta,r_0} \subset \Ocal$}
\\
& \text{ toward some limits denoted by
 $\hb,\,r \, \hb, h,\,\del_r h$ respectively.}
\endaligned
\ee 
We collect the properties of the limit $h$ as follows. 

\vskip.2cm
\fbox{\begin{minipage}{.9\textwidth}%

\begin{itemize}

\item[(i)]  The functions $\hb,\,r \, \hb, h,\,\del_r h$ are continuous in $\Ocal$.

\item[(ii)] Inequality (2) in  \eqref{equa-1234},
Lemma \ref{lemma-uniform bound of h and delr h with e}, and Lemma \ref{lemma-rhb small at 0} are true for $h$ (by tacitly  removing $\eps$ throughout our notation therein). 

\item[(iii)] The function $h$ satisfies the inequality (5) in \eqref{equa-1234}. 

\item[(iv)] At each $u$, $\lim_{r \to 0} \big(r \, \hb(u,r)\big) =0$, thanks to the bound on the derivative. 
\item[(v)] In particular, by (ii) and (iv) we see that $r \, \hb$ can be continuously extended on $[0,u_0]\times[0,+ \infty)$ with 
$$
\lim_{r \to 0} \sup_{u\in[0,u_0]} r \, h(u,r) =0.
$$
\end{itemize}

\end{minipage}
}


\paragraph{Statement of the pointwise bounds.}

We begin with a series of preliminary estimates for the matter field. Given $u_0> 0$, for any $\delta> 0$, we use the notation
$\Ocal_\delta = \big\{ (u,r)| 0\leq u \leq u_0,\,\delta<r\big\}$. The proofs of the following results
are postponed to Appendix~\ref{section-technical-lemmas}. 

\begin{lemma}
\label{lemma-rhb small at 0}
For each $\delta> 0$ and $\eps\leq  \delta/2$, the following estimate holds with with $\Hcal_\eps$  given in \eqref{def A}: 
\be
\sup_{r\ge \delta} |\hb_\eps (u,r)| \leq  \frac{1}{\sqrt\delta} \left( \int_{0}^{+ \infty}h_0^2 \, dr + \frac{u}{8\pi} + c^2 \Hcal_\eps(u)\right)^{1/2}. 
\ee 
\end{lemma}

\begin{lemma}
\label{uniform bound w.r.t. u}
The solution of the $\eps$-regularized evolution equation \eqref{epsilon self-gravitating eq} enjoys 
the following fall-off property (at each $u_0$) 
$$
\lim_{r\to + \infty} \big(\sup_{u\in[0,u_0]}|(r+ \eps) \, \hb_\eps|\big) =	\lim_{r\to + \infty} \big(\sup_{u\in[0,u_0]}|h_\eps|\big) = \lim_{r\to + \infty} \big(\sup_{u\in[0,u_0]}|\del_r h_\eps|\big) =0.
$$
\end{lemma}
 
\begin{lemma}
\label{lemma equi-continuity of delr h}
The family of functions $\{\del_r h_\eps|0<\eps\leq \delta/2\}$ is equicontinuous in $\Ocal_\delta$.
\end{lemma}

\begin{lemma}
\label{equi-continuity of hb}
The family of functions $\{\hb_\eps|0<\eps\leq \frac{\delta}{2} \}$ is equicontinuous in $\Ocal_\delta$.
\end{lemma}


\paragraph{Consequence of the pointwise estimates.}
 
By Lemmas \ref{lemma-uniform bound of h and delr h with e} and~\ref{lemma equi-continuity of delr h}, 
we deduce that there exists $h\in C^1(\Ocal)$ such that (after passing to a subsequence) $h_\eps$ and $\del_r h_\eps$ converge uniformly on compact subsets $\Ocal_{\delta,r_0}$ to $h$ and $\del_r h$,
 respectively. Moreover, by Lemmas~\ref{lemma-uniform bound of h and delr h with e} and~\ref{uniform bound w.r.t. u} we find
\bel{uniform bound of h and delr h}
	\sup_{r>\delta}|h(u,r)| \leq c_3(u,\delta), 
	\qquad
	\quad 	\sup_{r>\delta}|\del_r \, h(u,r)| \leq c_5(u,\delta)
\ee
together with 
\bel{uniform bound w.r.t. u to h}
\lim_{r\to + \infty} \big(\sup_{u\in[0,u_0]}|h|\big) = \lim_{r\to + \infty} \big(\sup_{u\in[0,u_0]}|\del_r h|\big) =0.
\ee 
Therefore,  $h_\eps$ and $\del_r h_\eps$ converge uniformly in $\Ocal_\delta$ to $h$ and $\del_r h$ respectively.\\

In particular, for any $r_0>\delta> 0$ we have 
$
\int_{\delta}^{r_0}h_\eps^2 \, dr \to  \int_{\delta}^{r_0}h^2 \, dr,
$
uniformly in $u$, hence by Proposition \ref{lemma-energy}
\bel{f is bounded}
\int_{\delta}^{r_0}h^2 \, dr\leq \int_{0}^{+ \infty}h_0^2 \, dr + \Big(\frac{1}{8\pi} +4\pi \, \int_0^{+ \infty} \min\big\{A_0,A_1r^{-2} \big\} \, dr\Big)u+ \int_{0}^{u}M_\eps(u_1) \, du_1.
\ee
 Next consider the sequence
$$
f_i:= 
\begin{cases} 
0, & r \notin [\delta/2^i,i],
\\
h^2, & \text{otherwise.} 
\end{cases}
$$
Observe that $\{f_i\}$ is an increasing sequence of measurable functions and pointwise converges to $h$, then by the monotone convergence theorem we obtain from \eqref{f is bounded} that $h\in L^2(0,+ \infty)$ at each $u$ and 
\bel{L2-bound of h}
\int_0^{+ \infty}h^2 \, dr\leq \int_{0}^{+ \infty}h_0^2 \, dr + \Big(\frac{1}{8\pi} +4\pi \, \int_0^{+ \infty} \min\big\{A_0,A_1r^{-2} \big\} \, dr\Big)u+ \int_{0}^{u}M_\eps(u_1) \, du_1.
\ee
Now by  Lemmas \ref{lemma-uniform bound of h and delr h with e} and \ref{equi-continuity of hb}, $\{\hb_\eps\}$ is equibounded and equicontinuous in $\Ocal_\delta$. Therefore, the same is true for the sequence $\{ (r+ \eps) \, \hb_\eps\}$. Again, by the Ascoli-Arzela theorem there exists a function $\zeta\in C^0(\Ocal)$ such that (after passing to a subsequence) $(r+ \eps) \, \hb_\eps$ converges uniformly on compact subsets $\Ocal_{\delta,r_0}$ of $\Ocal$. Moreover, in view of Lemma \ref{lemma-rhb small at 0} we see that $
(\delta+ \eps) \, \hb_\eps(u,\delta)\leq \delta B_1$
(for each $\delta> 0$). 
It follows that at each $u$, $(\delta+ \eps) \, \hb_\eps(u,\delta)\to 0$ as $\delta\to 0$. 


\subsection{Convergence properties} 
\label{subsec-comp-conv}

\paragraph{Properties of the metric coefficient $g$.}

Now, for any $r\in [\delta,+ \infty)$, we can write  
$
\frac{(h_\eps- \hb_\eps)^2}{r+ \eps} \leq 2\Big(\frac{h_\eps^2}{r+ \eps} + \frac{\hb_\eps^2}{r+ \eps} \Big),
$
which combined with the inequality (5)  in Theorem~\ref{lemma-global epsilon existence} and (iii), stated above, gives that, for each $\eta> 0$, there exists $r_0> 0$ sufficiently large such that
\bel{r>r_0}
\int_{r_0}^{+ \infty} \frac{(h_\eps- \hb_\eps)^2}{r+ \eps} \, dr
\leq \frac{\eta}{3}, 
\qquad\quad
\int_{r_0}^{+ \infty} \frac{(h - \hb)^2}{r} \, dr\leq \frac{\eta}{3}
\ee
uniformly in $u$ and $\eps$. On the other hand, by (i) we may take $\eps \in (0,1)$ small enough such that, for all $\delta\leq r_1\leq r_0$,
\bel{r<r_0}
\Big|\int_{r_1}^{r_0} \frac{(h_\eps- \hb_\eps)^2}{r+ \eps} \, dr - \int_{r_1}^{r_0} \frac{(h - \hb)^2}{r} \, dr\Big| \leq \frac{\eta}{3}.
\ee 
Therefore, from \eqref{r>r_0} and \eqref{r<r_0} we obtain ($r_1\geq \delta$) 
$$
\aligned
\Big|\int_{r_1}^{+ \infty} \frac{(h_\eps- \hb_\eps)^2}{r+ \eps} \, dr - \int_{r_1}^{+ \infty} \frac{(h - \hb)^2}{r} \, dr\Big|
& \leq  \Big|\int_{r_1}^{r_0} \frac{(h_\eps- \hb_\eps)^2}{r+ \eps} \, dr - \int_{r_1}^{r_0} \frac{(h - \hb)^2}{r} \, dr\Big|
\\
& \quad + \int_{r_0}^{+ \infty} \frac{(h_\eps- \hb_\eps)^2}{r+ \eps} \, dr + \int_{r_0}^{+ \infty} \frac{(h - \hb)^2}{r} \, dr \leq  \eta.
\endaligned
$$
This means that 
\bel{e epsilon to e}
\int_{r}^{+ \infty} \frac{(h_\eps- \hb_\eps)^2}{r+ \eps} \, dr\to \int_{r}^{+ \infty} \frac{(h - \hb)^2}{r} \, dr,
\ee 
uniformly in $\Ocal_{\delta}$, for any $\delta> 0$, hence
$$
e^{\nu_\eps + \lambda_\eps}:=e^{-4\pi \, \int_{r}^{+ \infty} \frac{(h_\eps- \hb_\eps)^2}{r+ \eps} \, dr} \to  e^{\nu+ \lambda}:=e^{-4\pi \int_{r}^{+ \infty} \frac{(h - \hb)^2}{r} \, dr}
$$
uniformly in $\Ocal_{\delta}$ for each $\delta> 0$. Now, note that at each $u$, $e^{(\nu+ \lambda)(u,r)}$ is positive, continuous and non-decreasing in $r$. Therefore, we may set
$ e^{(\nu+ \lambda)(u,0)}:= \lim_{r \to 0}e^{(\nu+ \lambda)(u,r)}$, hence by (i) and (v) we obtain 
\be
\text{
$g_\eps$ converges uniformly in $\Ocal_{\delta,r_0}$ to $g$}
\ee
and 
\vskip.2cm

\fbox{\begin{minipage}{.9\textwidth}%

\begin{itemize}
\item[(vi)] $g:=(1-8\pi r^2\hb^2) e^{\nu+ \lambda}$ is continuous in $[0,u_0]\times[0,+ \infty)$, with $g(u,0) =e^{(\nu+ \lambda)(u,0)}$. 
\end{itemize}

\end{minipage}}

\vskip.2cm

Furthermore, 
since $\{e^{(\nu+ \lambda)(u,1/i)}|i\in\mathbb{N} \}$ is a non-increasing sequence of continuous functions of $u$ which take values in $[0,1]$, the function
$$
g(u,0) =e^{(\nu+ \lambda)(u,0)}= \lim_{i\to + \infty}e^{(\nu+ \lambda)(u,1/i)}
$$
is measurable in $u$ and takes values in $[0,1]$.


\paragraph{Properties of the function $\gb$.}

Recall that $\gb = \frac{1}{r} \int_0^rg \, ds$.  
 Since $g$ is continuous with respect to $r$ for all $r\geq0$, so does $\gb$ and $\gb(u,0) =g(u,0)$. Now given $\eta> 0$ and $r_0>\eta$, since $g$ is continuous with respect to $u$ in $\Ocal$, there exists $\zeta> 0$ small enough such that for all $|u-u^\prime| \leq \zeta$ and for all $r\in \leq [\eta/2,r_0]$ we have 
 $
 |g(u,r) -g(u^\prime,r)| \leq \frac{\eta}{2r_0}.
 $
 It follows that
 $$
 \aligned
& \Big|\int_0^rg(u,s) \, ds- \int_0^rg(u^\prime,s) \, ds\Big| 
\\
& \leq \int_0^{\eta/2}|g(u,s) -g(u^\prime,s)| \, ds+ \int_{\eta/2}^r|g(u,s) -g(u^\prime,s)| \, ds\leq \frac{\eta}{2} + \big(\frac{\eta}{2r_0} \big)r_0\leq \eta.
\endaligned
 $$
 That means that $\int_0^rg \, ds$, hence $\gb$ is continuous with respect to $u$ in $\Ocal$.

 We will next prove that $\gb_\eps$ converges uniformly in $\Ocal_{\delta}$ to $\gb$. First of all, observe that for all $r\geq r_0$
 $$
 \aligned
 \gb_\eps(u,r) & =  \frac{(r_0+ \eps)}{(r+ \eps)} \gb_\eps(u,r_0)+ \frac{1}{r+ \eps} \int_{r_0}^rg_\eps(u,s) \, ds,
 \\
 \gb(u,r) & =  \frac{r_0}{r} \gb(u,r_0)+ \frac{1}{r} \int_{r_0}^rg(u,s) \, ds,
 \endaligned
 $$
 and $g_\eps \to  g$ uniformly in $\Ocal$. Then to show $\gb_\eps$ converges uniformly in $\Ocal_{\delta}$ to $\gb$, it suffices to show that $\gb_\eps$ converges uniformly in $\Ocal_{\delta,r_0}$ to $\gb$. 
 By a standard argument we are now in a position to deduce that 
%
\be
\text{
$\gb_\eps$ converges uniformly in $\Ocal_{\delta,r_0}$ to $\gb$.}
\ee
so 
\vskip.2cm

\fbox{\begin{minipage}{.9\textwidth}%

\begin{itemize}
\item[(vii)] $\gb$ is continuous on $[0,u_0]\times[0,+ \infty)$ with $\gb(u,0) =g(u,0)$, 
\end{itemize}

\end{minipage}
}


\paragraph{Properties of the characteristics.}

Next, for each $(u_1,r_1) \in \Ocal$, we consider the characteristic $\chi (\cdot; u_1,r_1)$, namely
  the solution of the ordinary differential equation
$\del_u\chi= -  \frac{1}{2} \gb(u,\chi)$ with 
$\chi(u_1; u_1,r_1) =r_1$.
Since $\gb$ and $\del_r\gb$ are continuous in $\Ocal$, then $\chi (\cdot ; u_1, r_1)$ is well-defined and unique. 
On the other hand, let $\chi_\eps(u; u_1, r_1)$ be the unique solution of 
$\del_u\chi_\eps= -  \frac{1}{2} \gb_\eps(u,\chi_\eps)$ with 
$\chi_\eps(u_1; u_1, r_1) =r_1$. 
Since $(u_1,r_1) \in \Ocal$, we have $(u_1,r_1) \in \Ocal_\delta$ for some $\delta> 0$, and therefore 
$(u,\chi_\eps(u; u_1,r_1)),\, (u,\chi (u;u_1,r_1)) \in\Ocal_{\delta,r_1+ u_0/2}
$ 
for all $u\leq u_1$.  
We then see that 
\be
\text{$\chi_\eps(\cdot; u_1, r_1)$ converges uniformly to $\chi (\cdot; u_1,r_1)$.
}
\ee


\paragraph{Differentiability of $h$.}

Now, along the characteristic $\chi_\eps(\cdot; u_1,r_1)$ we have
$$
h_\eps(u_1,r_1) =h_0(\chi_\eps(0; u_1,r_1))+ \int_0^{u_1} \Big(\frac{(g_\eps- \gb_\eps)(h_\eps- \hb_\eps)}{2(r_2+ \eps)} - \frac{c^2}{2}(r_2+ \eps) \, \hb_\eps \, e^{\nu_\eps + \lambda_\eps} \Big)_{\chi_\eps(u; u_1,r_1)} \, du_2.
$$
From this equation, using (iv), (vii), (ix), (x) and taking $\eps\to 0$, for all $(u_1,r_1) \in \Ocal$ we obtain 
\bel{main equation 1}
h(u_1,r_1) =h_0(\chi (0; u_1,r_1))+ \int_0^{u_1} \Big(\frac{(g- \gb)(h - \hb)}{2r_2} - \frac{c^2}{2}r_2\hb \, e^{\nu+ \lambda} \Big)_{\chi (u; u_1,r_1)} \, du_2.
\ee
This tells us that $h$ satisfies (in the domain $\Ocal$) the nonlinear evolution equation \eqref{self-gravitating eq} in the integral sense.

Furthermore, we also claim that $h$ is continuously differentiable with respect to $u$ in $\Ocal$. Let us set 
$$
f:= \frac{1}{2r}(g- \gb)(h - \hb) - \frac{c^2}{2}re^{\nu+ \lambda} \hb.
$$
Given $(u_1,r_1) \in \Ocal$ and $\Delta u_1$, let $r_1+ \Delta r_1$ be the value of $r$ at which the characteristic $\chi (\cdot; u_1,r_1)$ intersects the line $u=u_1- \Delta u_1$. From \eqref{main equation 1}, we obtain
\begin{equation*}
\label{main equation 1 with u}
\aligned
& \frac{h(u_1,r_1) - h(u_1- \Delta u_1,r_1)}{\Delta u_1} 
\\
& = \frac{h(u_1,r_1) - h(u_1- \Delta u_1,r_1+ \Delta r_1)}{\Delta u_1}
+ \frac{h(u_1- \Delta u_1,r_1+ \Delta r_1) - h(u_1- \Delta u_1,r_1)}{\Delta u_1}
\\
& = \frac{1}{\Delta u_1} \int_{u_1- \Delta u_1}^{u_1}f(u,\chi(u,r_1)) \, du + \Big(\frac{h(u_1- \Delta u_1,r_1+ \Delta r_1) - h(u_1- \Delta u_1,r_1)}{\Delta r_1} \Big)\Big(\frac{\Delta r_1}{\Delta u_1} \Big).
\endaligned
\end{equation*}
Since $h$ is continuously differentiable with respect to $r$ in $\Ocal$, and since
$$
\frac{\Delta r_1}{\Delta u_1}= \frac{1}{\Delta u_1} \int_{u_1- \Delta u_1}^{u_1} \frac{1}{2} \gb(u,\chi(u,r_1)) \, du,
$$
then, taking $\Delta u_1\to 0$, the right-hand side of \eqref{main equation 1 with u} tends to the limit $\Big(f+ \frac{\gb}{2} \del_r h\Big)(u_1,r_1)$. Therefore, $h$ is continuously differentiable with respect to $u$ in $\Ocal$ and 
$$
\del_ uh = \frac{1}{2r}(g- \gb)(h - \hb) - \frac{c^2}{2}r \, \hb \, e^{\nu+ \lambda} + \frac{\gb}{2} \del_r h.
$$
We conclude that $h\in C^1(\Ocal)$ and satisfies the nonlinear evolution equation \eqref{self-gravitating eq} 
as an equality between differentiable functions in $\Ocal$. 


\paragraph{Toward the Hawking mass equation.}

\begin{lemma}
For each $u_0 \geq 0$ and  $r_0> 0$ one has
$
\Big(\frac{e^{\nu_\eps + \lambda_\eps}}{\gb_\eps} \Big)(u,.) \in L^1(0,r_0).
$
\end{lemma}

\begin{proof} Observe first that for all $r\geq \delta$
\bel{bound below of e^nu+lambda}
\aligned
e^{(\nu_\eps + \lambda_\eps)(u,r)}  & \geq  e^{- \frac{4\pi}{\delta} \int_0^{+ \infty}(h_\eps- \hb_\eps)^2(u,s) \, ds}
\geq C_1(u_0,\delta),
\endaligned
\ee
where, thanks to (iii), the constant $C_1(u_0,\delta) > 0$ is independent of $u$ and $\eps$. On the other hand, arguing as in the proof of \eqref{bound above of lambda}, for all $(u,r) \in [0,u_0]\times [\delta,+ \infty)$ we find 
$$
\lambda_\eps(u,r)\leq \lambda_\eps(0,r)+ \int_0^{u_0} \bigg(\pi \, \gb_\eps \frac{(h_\eps- \hb_\eps)^2}{r+ \eps} + \pi c^2(r+ \eps) \, \hb_\eps^2 e^{\nu_\eps + \lambda_\eps} \bigg) \, du,
$$
hence, by (2) in \eqref{equa-1234}, together with 
\eqref{g-gb}, \eqref{control h in C0}, and \eqref{control del_rg} there exists a constant $C_2(u_0,\delta) > 0$
(independent of $u$ and $\eps$) such that
\bel{bound above of lambda t=1}
\lambda_\eps(u,r)\leq C_2(u_0,\delta).
\ee
Combining this result with \eqref{bound below of e^nu+lambda}, we have 
$$
\gb_\eps= e^{\nu_\eps + \lambda_\eps}e^{-2\lambda_\eps} \geq C_1(u_0,\delta) e^{-2C_2(u_0,\delta)}, 
\qquad r\geq \delta.
$$
In other words, $\gb_\eps(u,\delta)$ is uniformly bounded from below (in $\eps$ and $u$) by a positive constant. Therefore, we find 
\bel{e/gb converges}
\int_{\delta}^{r_0} \Big(\frac{e^{\nu_\eps + \lambda_\eps}}{\gb_\eps} \Big)(u,r) \, dr \to  \int_{\delta}^{r_0} \Big(\frac{e^{\nu+ \lambda}}{\gb} \Big)(u,r) \, dr
\qquad \text{ uniformly in $u$.}
\ee
Now, by the identity \eqref{integral identity}, for each $u$, we have
\bel{integral identity consequence 1}
\aligned
\int_{\delta}^{r_0} \Big(\frac{e^{\nu_\eps + \lambda_\eps}}{\gb_\eps} \Big)(u,r) \, dr& \leq  \int_{\delta}^{r_{0,\eps}} \Big(\frac{e^{\nu_\eps + \lambda_\eps}}{\gb_\eps} \Big)(0,r) \, dr
 \leq  \int_{\delta}^{r_{0} +u/2} \Big(\frac{e^{\nu_\eps + \lambda_\eps}}{\gb_\eps} \Big)(0,r) \, dr 
\\
& \leq  
2\int_{\delta}^{r_{0} +u/2} \Big(\frac{e^{\nu+ \lambda}}{\gb} \Big)(0,r) \, dr 
\leq C_3(r_0+ \frac{1}{2}u),
\endaligned
\ee
where the second inequality follows from $r_{0,\eps} \leq r_0+ u/2$, the third from \eqref{e/gb converges} and
\bel{c3-def}
C_3:=2\int_{0}^{r_{0} +u/2} \Big(\frac{e^{\nu+ \lambda}}{\gb} \Big)(0,r) \, dr
\ee
 is independent of $\eps$ and $\delta$. Therefore, we have 
\bel{bound of f_i}
\int_{\delta}^{r_0} \Big(\frac{e^{\nu+ \lambda}}{\gb} \Big) \, dr\leq C_3\Big(r_0+ \frac{1}{2}u\Big). 
\ee
Next, letting $\{\delta_i\}$ be a subsequence converging to $0$, we consider the sequence of functions $\{f_i\}$ given by
$$
f_i:= 
\begin{cases} 
\frac{e^{\nu+ \lambda}}{\gb}, 
& r\geq \delta_i,  
\\
0, & \text{otherwise.} 
\end{cases} 
$$
Since, at each $u$, $\{f_i\}$ is an increasing sequence of measurable functions and converges pointwise to $\frac{e^{\nu+ \lambda}}{\gb}$, then by the monotone convergence theorem 
and thanks to \eqref{bound of f_i} 
we conclude that 
(at each $u$) $\frac{e^{\nu+ \lambda}}{\gb} \in L^1(0,r_0)$ for all arbitrary $r_0$ and, as desired, 
$$
\int_0^{r_0} \Big(\frac{e^{\nu+ \lambda}}{\gb} \Big)(u,r) \, dr\leq C_3\Big(r_0+ \frac{1}{2}u\Big). 
\hskip2.cm 
\qedhere 
$$
\end{proof}


\section{Mass equation and completion of the proof}
\label{sec-section7}

\subsection{The function $\xi$} 

\paragraph{Analysis of the function $\xi$.}

The main remaining difficulty now is to prove that $h$ satisfies the evolution equation for the Hawking mass. 
In order to do that, we begin by analyzing $\xi_\eps$. 
Similarly as in \eqref{integral identity consequence 1}, we have from the $\eps$-integral identity  \eqref{integral identity} that for any small $(\eps,\delta)$
$$
\iint_{\Ocal_{\delta,r_0}} \frac{e^{\nu_\eps + \lambda_\eps}}{\gb_\eps^2} \frac{\xi_\eps^2}{(r+ \eps)} \, dudr\leq C_4
:= \frac{C_3(r_0+u_0/2)}{2\pi}.
$$
This means that the sequence 
$\Big(\frac{e^{\frac{\nu_\eps + \lambda_\eps}{2}}}{\gb_\eps} \frac{\xi_\eps}{(r+ \eps)^{1/2}} \Big)$
is contained in the closed ball of radius $C_4$ in $L^2(\Ocal_{\delta,r_0})$, for any $r_0$ and small $(\eps,\delta)$. Therefore, an analysis similar to that in \cite[Section 5]{Chr2} shows that there exists a function defined on $\Ocal$, which we denote by
$\frac{e^{\frac{\nu+ \lambda}{2}}}{\gb} \frac{\xi}{r^{1/2}}$
such that (after passing to a subsequence), for any $r_0\geq \delta> 0$,  the sequence $\frac{e^{\frac{\nu_\eps + \lambda_\eps}{2}}}{\gb_\eps} \frac{\xi_\eps}{(r+ \eps)^{1/2}}$
converges weakly in $L^2(\Ocal_{\delta,r_0})$ to the expected limit above
 and, moreover, we obtain
$$
\frac{e^{\frac{\nu+ \lambda}{2}}}{\gb} \frac{\xi}{r^{1/2}} \in L^2(\Ocal(r_0))
 \quad \text{for all $r_0$,} 
 \quad  \quad  \quad 
\iint_{\Ocal(r_0)} \frac{e^{\frac{\nu+ \lambda}{2}}}{\gb} \frac{\xi}{r^{1/2}} \, dudr\leq C_4,
$$
with $\Ocal(r_0) := \big\{ (u,r) \in \Ocal|r\leq r_0\big\}$. 


In particular, for any $L^2$ function $\omega$ whose support is a compact set in $\Ocal$, we have
\bel{convergece of Theta}
\iint_{\Ocal} \frac{e^{\frac{\nu_\eps + \lambda_\eps}{2}}}{\gb_\eps} \frac{\xi_\eps}{(r+ \eps)^{1/2}} \omega \, dudr 
\to 
\iint_{\Ocal} \frac{e^{\frac{\nu+ \lambda}{2}}}{\gb} \frac{\xi}{r^{1/2}} \omega \, dudr.
\ee
Since, for any given $r_0\geq \delta> 0$, the quantity $\frac{e^{\frac{\nu_\eps + \lambda_\eps}{2}}}{\gb_\eps}$ converges uniformly in $\Ocal_{\delta,r_0}$  to $\frac{e^{\frac{\nu+ \lambda}{2}}}{\gb}$, it follows from \eqref{convergece of Theta} that, for any smooth function $\varphi$ whose support is compact and contained in the interior of $\Ocal$,
\bel{Theta integral}
\iint_{\Ocal} \xi_\eps \varphi \, dudr \to  \iint_{\Ocal} \xi\varphi \, dudr.
\ee 
Now, using the definition of $\xi_\eps$ we get
$
\del_r\xi_\eps= \gb_\eps \frac{(h_\eps- \hb_\eps)}{(r+ \eps)} - c^2(r+ \eps) \, \hb_\eps e^{\nu_\eps + \lambda_\eps}
$
and then
\bel{delr Theta 1}
- \iint_{\Ocal} \Big(\gb_\eps \frac{(h_\eps- \hb_\eps)}{(r+ \eps)} - c^2(r+ \eps) \, \hb_\eps e^{\nu_\eps + \lambda_\eps} \Big)\varphi \, dudr= \iint_{\Ocal} \xi_\eps \del_r\varphi \, dudr.
\ee
Since $\gb_\eps \frac{(h_\eps- \hb_\eps)}{(r+ \eps)} - c^2(r+ \eps) \, \hb_\eps e^{\nu_\eps + \lambda_\eps}$ converges uniformly in $\Ocal_{\delta,r_0}$, taking $\eps \to  0$ we obtain, from \eqref{Theta integral} and \eqref{delr Theta 1}, that
$$
- \iint_{\Ocal} \Big(\gb\frac{(h - \hb)}{r} - c^2r \, \hb \, e^{\nu+ \lambda} \Big)\varphi \, dudr= \iint_{\Ocal} \xi\del_r\varphi \, dudr.
$$
This means that $\xi$ is weakly differentiable with respect to $r$ in $\Ocal$ and 
\bel{delr Theta}
\del_r\xi= \gb\frac{(h - \hb)}{r} - c^2r \, \hb \, e^{\nu+ \lambda}.
\ee 


\paragraph{Continuity of the function $\xi$.}
\begin{lemma}
\label{lemma-lost} 
At almost all $u$, $\xi(u,r)$ is continuous with respect to $r$ and $\xi(u,r) \to  0$ as $r\to0$ and, consequently,
\bel{formula Theta 2}
\xi(u,r) = \lim_{\delta\to0} \int_{\delta}^{r} \Big(\gb\frac{(h - \hb)}{s} - c^2s\hb \, e^{\nu+ \lambda} \Big) \, ds.
\ee
\end{lemma}

\begin{proposition} \label{key of Dgamma}
At almost all $u$, $\frac{\xi}{\sqrt{e^{\nu+ \lambda}}}$ is a continuous function of $r$ and converges to zero as $r \to 0$. 
At each $r$, one has $\frac{\xi}{\sqrt{e^{\nu+ \lambda}}} \in L^2(0,u_0)$ and, furthermore, 
\bel{sup Theta 0}
\sup_{r\in [0,+ \infty)} \Big|\frac{\xi}{\max\{1,r\} \sqrt{e^{\nu+ \lambda}}} \Big|\in L^2(0,u_0),
\qquad\qquad 
\lim_{r \to 0} \Big(\int_0^{u_0} \frac{\xi^2}{e^{\nu+ \lambda}}(u,r) \, du\Big) =0.
\ee
\end{proposition}

By using similar methods to \cite[Section 5]{Chr2} we can show that $\hb$ is weakly differentiable with respect to $u$ in $\Ocal$ and 
\bel{evolution hb}
D\hb= \frac{\xi}{2r}.
\ee
We now give the proof of the previous two statements. 

\begin{proof}[Proof of Lemma~\ref{lemma-lost} ] 
Since $\frac{e^{\nu+ \lambda}}{\gb} \geq 1$, the fact that $\frac{e^{\frac{\nu+ \lambda}{2}}}{\gb} \frac{\xi}{r^{1/2}} \in L^2(\Ocal(r_0))$ implies that
\bel{Theta(0) =0}
\frac{\xi}{\sqrt{r\gb}} \in L^2(\Ocal(r_0)).
\ee
In particular, for almost all $u$, $\frac{\xi(u,r)}{\sqrt{r\gb(u,r)}} \in L^2(0,r_0)$, hence so does $\frac{\xi(u,r)}{\sqrt{re^{\nu+ \lambda}}}$. Let us introduce 
\bel{define Theta*}
r:=r_0e^{-s},\,\,\,\xi^*(u,s) := \xi(u,r)\quad\text{and} \quad e^{(\nu^*+ \lambda^*)(u,s)}:=e^{(\nu+ \lambda)(u,r)}.
\ee
Since
\bel{Theta L2}
\int_0^{+ \infty} \frac{\xi^{*2}}{e^{\nu^*+ \lambda^*}}(u,s) \, ds= \int_0^{r_0} \frac{\xi^2(u,r)}{re^{\nu+ \lambda}}(u,r) \, dr,
\ee
then $\frac{\xi^*}{\sqrt{e^{\nu^*+ \lambda^*}}}(u,.) \in L^2(0,+ \infty)$ for almost all $u$. On the other hand, by \eqref{delr Theta}, we have
$
\del_s\xi^*= - r\del_r\xi= -  \gb(h - \hb)+c^2r^2\hb \, e^{\nu+ \lambda}.
$
By a direct calculation,
\bel{delr Theta L2}
\aligned
\int_0^{+ \infty} \frac{1}{e^{\nu^*+ \lambda^*}} \big(\del_s\xi^*\big)^2ds
&= \int_0^{r_0} \frac{1}{re^{\nu+ \lambda}} \Big(\gb(h - \hb) \, dr - c^2r^2\hb \, e^{\nu+ \lambda} \Big)^2 \, dr
\\
& \leq   2\int_0^{r_0}e^{\nu+ \lambda} \frac{(h - \hb)^2}{r} \, dr +c^4\int_0^{r_0}e^{\nu+ \lambda}r^3\hb^2
\\
& \leq  \frac{1}{2\pi} \big(e^{(\nu+ \lambda)(r_0)} -e^{(\nu+ \lambda)(0)} \big)+c^4r^3_0\int_0^{r_0} \hb^2 \, dr
  \leq  \frac{1}{2\pi} +c^4r_0^3\int_0^{r_0} \hb^2 \, dr.
\endaligned
\ee
Since $e^{\nu^*+ \lambda^*} \leq 1$ and since $\int_0^{+ \infty} \hb^2 \, dr$ is bounded by a continuous function of $u$, then \eqref{Theta L2} and \eqref{delr Theta L2} tell us that, for almost all $u$, $\xi^*$ belongs to the Sobolev space $H^1(0,+ \infty)$.
Therefore, by the Sobolev embedding theorem, for almost all $u$, $\xi^*(u,.) \in C^{1/2}[0,+ \infty)$ and $\xi^*(u,s) \to  0$ as $s\to + \infty$. Consequently, for almost all $u$, $\xi(u,r)$ is continuous with respect to $r$ and $\xi(u,r) \to  0$ as $r\to0$.

Now, for any $\delta> 0$, the function $\xi$ is expressed as
\bel{formula Theta}
\xi(u,r) = \xi(u,\delta)+ \int_{\delta}^{r} \Big(\gb\frac{(h - \hb)}{s} - c^2s\hb \, e^{\nu+ \lambda} \Big) \, ds,
\ee
and taking $\delta\to0$ in this equation, the result of the lemma follows.
\end{proof}


\begin{proof}[Proof of Proposition~\ref{key of Dgamma}] 
From \eqref{Theta(0) =0}, we have that $\frac{\xi}{\sqrt{e^{\nu+ \lambda}}}(.,\delta) \in L^2(0,u_0)$. It follows, from \eqref{formula Theta}, that so does $\frac{\xi}{\sqrt{e^{\nu+ \lambda}}}(.,r)$ for all $r> 0$. Moreover, it is also true by \eqref{Theta(0) =0} that, 
for any $(u_1,r_1) \in \Ocal$, the restriction of $\xi$ to the characteristic $\chi(\cdot; u_1,r_1)$ belongs to $L^2(0,u_1)$.

Now, let $(s,\,\xi^*,\,e^{\nu^*+ \lambda^*})$ be defined as in \eqref{define Theta*}. For every $s_0\in [0,+ \infty)$, by the Sobolev inequality we get
\bel{sobolev inequality}
\sup_{s\geq s_0} \xi^{*2}(s)\leq c_*\int_{s_0}^{+ \infty} \Big(\xi^{*2} + \big(\del_s\xi^{*} \big)^2\Big) \, ds,
\ee
where $c_*$ is a constant, independent of $s_0$. Since $e^{\nu^*+ \lambda^*}$ is a non-increasing function of $s$, it follows that
$$
\Big(\frac{\xi^{*2}}{e^{\nu^*+ \lambda^*}} \Big)(s_0)\leq \frac{c_*}{e^{(\nu^*+ \lambda^*)(s_0)}} \int_{s_0}^{+ \infty} \Big(\xi^{*2} + \big(\del_s\xi^{*} \big)^2\Big) \, ds\leq c_*\int_{s_0}^{+ \infty} \frac{1}{e^{\nu^*+ \lambda^*}} \Big(\xi^{*2} + \big(\del_s\xi^{*} \big)^2\Big) \, ds.
$$ 
Therefore, we have
\bel{consequence sobolev inequality}
\sup_{s\geq 0} \Big(\frac{\xi^{*2}}{e^{\nu^*+ \lambda^*}} \Big)(s)\leq c_*\int_{0}^{+ \infty} \frac{1}{e^{\nu^*+ \lambda^*}} \Big(\xi^{*2} + \big(\del_s\xi^{*} \big)^2\Big) \, ds.
\ee
Taking into account \eqref{Theta L2} and the second inequality of \eqref{delr Theta L2},  from \eqref{consequence sobolev inequality} 
we deduce that
\bel{sup Theta}
\aligned
\sup_{r\in[0,r_0]} \Big(\frac{\xi^{2}}{e^{\nu+ \lambda}} \Big)(r)
& \leq  c_*\Big(\int_0^{r_0} \frac{\xi^2}{re^{\nu+ \lambda}} \, dr +c^4r_0^3 \int_0^{r_0} \hb^2 \, dr + \frac{1}{2\pi} \big(e^{(\nu+ \lambda)(r_0)} -e^{(\nu+ \lambda)(0)} \big)\Big)
\\
& \leq  c_*\Big(\int_0^{r_0} \frac{\xi^2}{re^{\nu+ \lambda}} \, dr +c^4r_0^3F(u)+ \frac{1}{2\pi} \big(e^{(\nu+ \lambda)(r_0)} -e^{(\nu+ \lambda)(0)} \big)\Big),
\endaligned
\ee
where $F(u)$ is a continuous function of $u$ and we applied \eqref{energy-int}. Since $e^{\nu+ \lambda}$ is continuous on $[0,u_0]\times [0,+ \infty)$, taking $r_0\to 0$ in \eqref{sup Theta}, we derive from \eqref{Theta(0) =0} that, at almost all $u$, $\frac{\xi^2}{e^{\nu+ \lambda}} \to  0$ as $r\to0$, as claimed.

Now, by integrating \eqref{sup Theta} with respect to $u$ we obtain 
\bel{sup Theta 2}
\aligned
& \int_0^{u_0} \sup_{r\in[0,r_0]} \Big(\frac{\xi^{2}}{e^{\nu+ \lambda}} \Big)(u,r) \, du
\\
& \leq c_*\Big(\iint_{\Ocal(r_0)} \frac{\xi^2}{re^{\nu+ \lambda}} \, drdu+c^4r_0^3\int_0^{u_0}F(u) \, du+ \frac{1}{2\pi} \int_0^{u_0} \big(e^{(\nu+ \lambda)(r_0)} -e^{(\nu+ \lambda)(0)} \big) \, du\Big).
\endaligned
\ee
It then follows from \eqref{Theta(0) =0} that
\bel{r_0<1}
\sup_{r\in [0,1]} \Big| \xi \, e^{- (\nu+ \lambda)/2} \Big|\in L^2(0,u_0).
\ee
On the other hand, since 
$$
\xi(u,r) = \xi(u,1)+ \int_{1}^{r} \Big(\gb\frac{(h - \hb)}{s} - c^2s\hb \, e^{\nu+ \lambda} \Big) \, ds,
$$
we have, for all $r>1$,
$$
\aligned
\Big(\frac{\xi^2}{r^2e^{\nu+ \lambda}} \Big)(u,r)& \leq  2\Big(\frac{\xi^2}{e^{\nu+ \lambda}} \Big)(u,1)+ \frac{2}{r^2e^{\nu+ \lambda}} \Big(\int_{1}^{r} \Big(\gb\frac{(h - \hb)}{s} - c^2s\hb \, e^{\nu+ \lambda} \Big) \, ds\Big)^{2} 
\\
& \leq 2\Big(\frac{\xi^2}{e^{\nu+ \lambda}} \Big)(u,1)+ \frac{2}{r} \int_1^r\Big((h - \hb)^2+c^4s^2\hb^2\Big) \, ds \leq  2\Big(\frac{\xi^2}{e^{\nu+ \lambda}} \Big)(u,1)+F_1(u),
\endaligned
$$
where $F_1(u)$ is a continuous function of $u$ and we applied \eqref{energy-int}. Combining this result with \eqref{r_0<1},  gives 
the first statement in \eqref{sup Theta 0}, as claimed.

Finally, since by dominated convergence,
$$
\int_0^{u_0} \big(e^{(\nu+ \lambda)(r_0)} -e^{(\nu+ \lambda)(0)} \big) \, du \to  0,\quad\text{as $r\to0$},
$$
taking $r_0\to0$ in \eqref{sup Theta 2}, we complete the proof with 
$
\int_0^{u_0} \Big(\frac{\xi^{2}}{e^{\nu+ \lambda}} \Big)(u,r) \, du \to  0,\quad\text{as $r\to0$}.
$ 
\end{proof}


\subsection{Compactness arguments based on the mass equation}

It is convenient to introduce the function 
$
\gamma(u,r) := -  \frac{\nu+ \lambda}{4\pi}= \int_r^{+ \infty} \frac{(h - \hb)^2}{s} \, ds.
$
Before we start studying in more detail the terms that appear in the evolution equation for the Hawking mass, 
we recall that (cf. \eqref{derivative under integral})
$$
\del_u\bigg(\int_{r}^{+ \infty} \frac{(h_\eps- \hb_\eps)^2}{s+ \eps} \, ds\bigg) = \int_{r}^{+ \infty} \del_u\bigg(\frac{(h_\eps- \hb_\eps)^2}{s+ \eps} \bigg) \, ds,
$$
so that, by straightforward calculations,
$$
\aligned
& \del_u\bigg(\int_{r}^{+ \infty} \frac{(h_\eps- \hb_\eps)^2}{s+ \eps} \, ds\bigg)
 =  \int_r^{+ \infty} \frac{2(h_\eps- \hb_\eps)}{s+ \eps} \Big(Dh_\eps+ \frac{1}{2} \gb_\eps \del_r h_\eps-D\hb_\eps- \frac{1}{2} \gb_\eps \del_r \, \hb_\eps\Big) \, ds
\\
& =  \int_r^{+ \infty} \frac{h_\eps- \hb_\eps}{s+ \eps} \Bigg(\frac{(g_\eps- \gb_\eps)(h_\eps- \hb_\eps)}{s+ \eps} - c^2(s+ \eps) \, \hb_\eps e^{\nu_\eps + \lambda_\eps}
 + \gb_\eps \del_r h_\eps- \frac{\xi_\eps}{s+ \eps} - \frac{\gb_\eps}{s+ \eps}(h_\eps- \hb_\eps)\Bigg) \, ds.
\endaligned
$$
Combining this equation with  
$$
\frac{1}{r+ \eps} \del_r\gb_\eps(h_\eps- \hb_\eps)^2+ \frac{\gb_\eps(h_\eps- \hb_\eps)}{r+ \eps} \del_r h_\eps- \frac{\gb_\eps(h_\eps- \hb_\eps)^2}{(r+ \eps)^2}= \frac{g_\eps(h_\eps- \hb_\eps)^2}{2(r+ \eps)^2} + \del_r\Big(\gb_\eps \frac{(h_\eps- \hb_\eps)^2}{2(r+ \eps)} \Big),
$$
we obtain
\bel{evolution gamma 0}
\aligned
& \del_u\bigg(\int_{r}^{+ \infty} \frac{(h_\eps- \hb_\eps)^2}{s+ \eps} \, ds\bigg) 
\\
& =  - \frac{1}{2(r+ \eps)} \gb_\eps(h_\eps- \hb_\eps)^2
 + \int_r^{+ \infty} \Big(\frac{1}{2}g_\eps(h_\eps- \hb_\eps) - \xi_\eps- c^2(s+ \eps)^2\hb_\eps e^{\nu_\eps + \lambda_\eps} \Big)\frac{h_\eps- \hb_\eps}{(s+ \eps)^2} \, ds.
\endaligned
\ee
So, for any smooth function $\varphi$ whose support is compact and contained in the interior of $\Ocal$, we have
\bel{integral term}
\aligned
& \iint_{\Ocal} \bigg(- \frac{\gb_\eps(h_\eps- \hb_\eps)^2}{2(r+ \eps)} + \int_r^{+ \infty} \Big(\frac{g_\eps(h_\eps- \hb_\eps)}{2} - \xi_\eps- c^2(s+ \eps)^2\hb_\eps e^{\nu_\eps + \lambda_\eps} \Big)\frac{h_\eps- \hb_\eps}{(s+ \eps)^2} \, ds\bigg)\varphi \, dudr
\\
& = - \iint_{\Ocal} \bigg(\int_{r}^{+ \infty} \frac{(h_\eps- \hb_\eps)^2}{s+ \eps} \, ds\bigg)(\del_u\varphi) dudr,
\endaligned
\ee
and, by \eqref{e epsilon to e},
\bel{integral term right hand}
\iint_{\Ocal} \bigg(\int_{r}^{+ \infty} \frac{(h_\eps- \hb_\eps)^2}{s+ \eps} \, ds\bigg)\del_u\varphi \, dudr \to  \iint_{\Ocal} \gamma\del_u\varphi \, dudr.
\ee

We now study the left-hand side of \eqref{integral term}. First of all, 
thanks to (2) in  \eqref{equa-1234}
and  H\"{o}lder's inequality, we have 
\bel{integral term 0}
\aligned
\bigg|\int_r^{+ \infty}g_\eps \frac{(h_\eps- \hb_\eps)^2}{(s+ \eps)^2} \, ds\bigg| & \leq \frac{1+4\pi c^2 \Acal^2_\eps(u,r)}{(r+ \eps)^2} \int_r^{+ \infty}(h_\eps- \hb_\eps)^2ds,
\\
\bigg|\int_r^{+ \infty} \hb_\eps(h_\eps- \hb_\eps) e^{\nu_\eps + \lambda_\eps} \, ds\bigg| & \leq \frac{1}{r+ \eps} \bigg|\int_r^{+ \infty}(r+ \eps) \, \hb_\eps(h_\eps- \hb_\eps) e^{\nu_\eps + \lambda_\eps} \, ds\bigg|
\\
& \leq \frac{1}{r+ \eps} \bigg(\int_r^{+ \infty} \big(r+ \eps)^2\hb_\eps^2ds\bigg)^{1/2} \bigg(\int_r^{+ \infty}(h_\eps- \hb_\eps)^2ds\bigg)^{1/2}.
\endaligned
\ee
By the fourth and fifth assertions of Theorem~\ref{lemma-global epsilon existence}, 
$
\int_r^{+ \infty}(h_\eps- \hb_\eps)^2ds$ and $\int_r^{+ \infty} \big(s+ \eps)^2\hb_\eps^2ds
$
are bounded by a continuous function of $u$ and, consequently, \eqref{integral term 0} implies  
(as $r\to + \infty$)
$$
\int_r^{+ \infty}g_\eps \frac{(h_\eps- \hb_\eps)^2}{(s+ \eps)^2} \, ds\to 0,
\quad
\qquad
 \int_r^{+ \infty} \hb_\eps(h_\eps- \hb_\eps) e^{\nu_\eps + \lambda_\eps} \, ds  \to  0. 
$$
Combining this result with the property that (for any $r_0>\delta> 0$) $h_\eps,\,\hb_\eps,\,g_\eps$ converge uniformly on $\Ocal_{\delta,r_0}$ to $h,\,\hb,\,g$, respectively, and proceeding similarly as in
 \eqref{e epsilon to e} we deduce that 
\bel{integral term 1}
\aligned
\int_r^{+ \infty}g_\eps \frac{(h_\eps- \hb_\eps)^2}{(s+ \eps)^2} \, ds & \to  \int_r^{+ \infty}g_{h} \frac{(h - \hb)^2}{s^2} \, ds
 \qquad
\text{uniformly on $\Ocal_{\delta,r_0}$,}
\\
 \int_r^{+ \infty} \hb_\eps(h_\eps- \hb_\eps) e^{\nu_\eps + \lambda_\eps} \, ds
& \to 
\int_r^{+ \infty} \hb(h - \hb) e^{\nu+ \lambda} \, ds
 \qquad
\text{uniformly on $\Ocal_{\delta,r_0}$. }
\endaligned
\ee
Next, by integration by parts we find 
\be
\label{integral term 2}
\aligned
& 
\int_r^{+ \infty} \frac{h_\eps- \hb_\eps}{(s+ \eps)^2} \xi_\eps \, ds
\\
& = \xi_\eps \int_r^{+ \infty} \frac{h_\eps- \hb_\eps}{(s+ \eps)^2} \, ds
  + \int_r^{+ \infty} \bigg(\gb_\eps \frac{(h_\eps- \hb_\eps)}{s+ \eps} - c^2(s+ \eps) \, \hb_\eps e^{\nu_\eps + \lambda_\eps} \bigg)\bigg(\int_s^{+ \infty} \frac{h_\eps- \hb_\eps}{(s_1+ \eps)^2} \, ds_1\bigg) \, ds.
\endaligned
\ee
Applying H\"{o}lder's inequality we have 
$$
\aligned
\int_r^{+ \infty} \frac{h_\eps- \hb_\eps}{(s+ \eps)^2} \, ds& \leq  \bigg(\frac{1}{3(r+ \eps)^3} \bigg)^{1/2} \bigg(\int_r^{+ \infty}(h_\eps- \hb_\eps)^2ds\bigg)^{1/2} \leq \frac{1}{(r+ \eps)^{3/2}}F_2(u),
\endaligned
$$
where $F_2(u)$ is a continuous function of $u$, by Theorem~\ref{lemma-global epsilon existence}. Similarly to \eqref{integral term 1}, 
for any $r_0\geq \delta> 0$ the two expressions 
$$
\int_r^{+ \infty} \frac{h_\eps- \hb_\eps}{(s+ \eps)^2} \, ds\quad\text{and} \quad \int_r^{+ \infty} \bigg(\gb_\eps \frac{h_\eps- \hb_\eps}{s+ \eps} - c^2(s+ \eps) \, \hb_\eps e^{\nu_\eps + \lambda_\eps} \bigg)\bigg(\int_s^{+ \infty} \frac{h_\eps- \hb_\eps}{(s_1+ \eps)^2} \, ds_1\bigg) \, ds
$$
converge uniformly on $\Ocal_{\delta,r_0}$ to, respectively, 
$$ 
\int_r^{+ \infty} \frac{h - \hb}{s^2} \, ds\quad\text{and} \quad \int_r^{+ \infty} \bigg(\gb\frac{h - \hb}{s} - c^2s\hb \, e^{\nu+ \lambda} \bigg)\bigg(\int_s^{+ \infty} \frac{h - \hb}{s_1^2} \, ds_1\bigg) \, ds.
$$
 Taking these expressions as well as \eqref{Theta integral} into account in \eqref{integral term 2}, we obtain
\bel{integral term 3}
\iint_{\Ocal} \bigg(\int_r^{+ \infty} \frac{h_\eps- \hb_\eps}{(s+ \eps)^2} \xi_\eps \, ds\bigg)\varphi \, dudr \to  \iint_{\Ocal} \bigg(\int_r^{+ \infty} \frac{h - \hb}{s^2} \xi ds\bigg)\varphi \, dudr, 
\ee
where, similarly to \eqref{integral term 2}, we used integration by parts. Therefore, we conclude from \eqref{integral term 1} and \eqref{integral term 3} that the left-hand side of \eqref{integral term} converges to 
$$
\iint_{\Ocal} \bigg(- \frac{1}{2r} \gb(h - \hb)^2+ \int_r^{+ \infty} \Big(\frac{1}{2}g(h - \hb) - \xi- c^2s^2\hb \, e^{\nu+ \lambda} \Big)\frac{(h - \hb)}{s^2} \, ds\bigg)\varphi \, dudr.
$$
In view of this result together with \eqref{integral term 0}, we infer from  \eqref{integral term} that $\gamma$ is weakly differentiable with respect to $u$ in $\Ocal$ and 
\bel{evolution gamma}
\aligned
D\gamma= \int_r^{+ \infty} \Big(\frac{1}{2}g(h - \hb) - \xi- c^2s^2\hb \, e^{\nu+ \lambda} \Big)\frac{(h - \hb)}{s^2} \, ds.
\endaligned
\ee
Consequently, $e^{\nu+ \lambda}$ is weakly differentiable with respect to $u$ in $\Ocal$ and 
$
D\big(e^{\nu+ \lambda} \big) = - 4\pi e^{\nu+ \lambda}D\gamma.
$
We next compute the evolution equation for $\frac{1}{r} \int_{\delta}^{r}g \, ds$, with $\delta> 0$. We compute 
\begin{equation*}
\aligned
\label{evolution gb delta 1}
&
D\Big(\frac{1}{r} \int_{\delta}^{r}g \, ds\Big)
 = \frac{\gb}{2r} \Big(\frac{1}{r} \int_{\delta}^{r}(1-4\pi c^2 s^2\hb^2) e^{\nu+ \lambda} \, ds- (1-4\pi c^2 r^2\hb^2) e^{\nu+ \lambda} \Big) 
\\
& \quad + \frac{1}{r} \int_0^r\Big(D\big(e^{\nu+ \lambda} \big)+ \frac{1}{2} \gb\del_r\big(e^{\nu+ \lambda} \big) -8\pi c^2 s^2 e^{\nu+ \lambda} \hb \del_u\hb-4\pi c^2 s^2\hb^2 \del_u\big(e^{\nu+ \lambda} \big)\Big) \, ds 
\\
& = \frac{\gb}{2r^2}(r\gb- \delta\gb(\delta)) - \frac{1}{2r} \gb g 
   - \frac{4\pi}{r} \int_{\delta}^{r} \Big(gD\gamma- \frac{1}{2}g\gb\frac{(h - \hb)^2}{s} +c^2s\hb(h - \hb)\gb e^{\nu+ \lambda} + 2 \, c^2s^2\hb D\hb \, e^{\nu+ \lambda} \Big) \, ds
\endaligned
\end{equation*}
and, from \eqref{evolution gamma},  
$$
\aligned
& \int_{\delta}^{r}gD\gamma ds= r\gb D\gamma- \delta\gb(\delta) D\gamma(\delta) - \int_{\delta}^{r}s\gb\del_rD\gamma ds,
 \\
& \int_{\delta}^{r} \gb\Big(s\del_rD\gamma+ \frac{g}{2} \frac{(h - \hb)^2}{s} - c^2s\hb(h - \hb) e^{\nu+ \lambda} \Big) \, ds= \int_{\delta}^{r} \frac{(h - \hb)}{s} \gb\xi ds,
\endaligned
$$
so that, in view of \eqref{delr Theta} and \eqref{evolution hb} we arrive at the identity 
\bel{evolution gb}
\aligned
D\Big(\frac{1}{r} \int_{\delta}^{r}g \, ds\Big) & =  \frac{\gb}{2r^2}(r\gb- \delta\gb(\delta)) - \frac{1}{2r} \gb g-4\pi \gb D\gamma + \frac{4\pi \, \delta}{r} \gb(\delta)D\gamma(\delta)+ \frac{2 \, c^2\pi}{r} \int_{\delta}^{r} \xi\del_r\xi ds
\\
& = f_\delta- \frac{1}{2r} \gb(g- \gb) -4\pi \gb D\gamma+ \frac{2\pi}{r} \xi^2,
\endaligned
\ee
where 
$$
f_\delta(u,\delta) := -  \frac{\gb}{2r^2} \delta \gb(u,\delta) - \frac{c^2\pi}{r} \xi^2(u,\delta)+ \frac{4\pi \, \delta}{r} \gb(u,\delta)D\gamma(u,\delta).
$$
The first term of $f_\delta$ converges to $0$, as $\delta\to0$, and by Proposition \ref{key of Dgamma}, so does the second term for almost all $u$. 
Therefore, in order to conclude that $f_\delta  \to  0$, at almost all $u$, as $\delta\to0$, it suffices to check the lemma stated next.


\subsection{Final argument}

\begin{lemma} \label{lemma-regularity of rDgamma}
For almost all $u$, one has $\delta D\gamma(u,\delta) \to 0$ as $\delta\to0$.
\end{lemma}

\begin{proof}
In view of \eqref{evolution gamma} we have 
\bel{estimate Dgamma}
\aligned
|D\gamma| & = \Big|\int_\delta^{+ \infty} \Big(\frac{1}{2}g(h - \hb) - \xi- c^2r^2\hb \, e^{\nu+ \lambda} \Big)\frac{(h - \hb)}{r^2} \, dr\Big|
\\
& \leq  \frac{1}{2} \int_{\delta}^{+ \infty}e^{\nu+ \lambda} \frac{(h - \hb)^2}{r^2} \, dr +c^4\pi \, \int_\delta^{+ \infty} \hb^2(h - \hb)^2e^{\nu+ \lambda} \, dr
\\
&
\quad
+ c^2\int_\delta^{+ \infty}|\hb(h - \hb)|e^{\nu+ \lambda} \, dr + \int_{\delta}^{+ \infty}|\xi|\frac{|h - \hb|}{r^2} \, dr.
\endaligned
\ee
By a direct computation, we have
$$
\aligned
&
\delta\int_{\delta}^{+ \infty}e^{\nu+ \lambda} \frac{(h - \hb)^2}{r^2} \, dr =  \delta\int_{\delta}^{\sqrt{\delta}}e^{\nu+ \lambda} \frac{(h - \hb)^2}{r^2} \, dr + \delta\int_{\sqrt{\delta}}^{+ \infty}e^{\nu+ \lambda} \frac{(h - \hb)^2}{r^2} \, dr
\\
& \leq  \int_{\delta}^{\sqrt{\delta}}e^{\nu+ \lambda} \frac{(h - \hb)^2}{r} \, dr + \sqrt{\delta} \int_{\sqrt{\delta}}^{+ \infty}e^{\nu+ \lambda} \frac{(h - \hb)^2}{r} \, dr
  \leq \frac{1}{4\pi} \Big(e^{(\nu+ \lambda)(\sqrt{\delta})} -e^{(\nu+ \lambda)(\delta)} + \sqrt{\delta}(1-e^{(\nu+ \lambda)(\sqrt{\delta})})\Big)
\endaligned
$$
and, since the right-hand side of the previous inequality converges to $0$, as $\delta\to 0$, we obtain
$$
\delta\int_{\delta}^{+ \infty}e^{\nu+ \lambda} \frac{(h - \hb)^2}{r^2} \, dr \to  0\quad\text{as $\delta\to 0$}.
$$
Combining this result with the fact that $r \, \hb$ is bounded by a continuous function of $u$ (cf. \eqref{equa-1234}), we get
$$
\delta\int_{\delta}^{+ \infty}e^{\nu+ \lambda} \hb^2(h - \hb)^2 \, dr= \delta\int_{\delta}^{+ \infty} \big(r \, \hb\big)^2\Big(e^{\nu+ \lambda} \frac{(h - \hb)^2}{r^2} \Big) \, dr \to  0\quad\text{as $\delta\to 0$}.
$$
Next, by H\"{o}lder's inequality and (iii) in Section \ref{subsec-comp-conv} we deduce that 
$$
\delta\int_\delta^{+ \infty}|\hb(h - \hb)|e^{\nu+ \lambda} \, dr\leq \delta\Big(\int_{\delta}^{+ \infty}|\hb|^2 \, dr\Big)^{1/2} \Big(\int_{\delta}^{+ \infty}(h - \hb)^2 \, dr\Big)^{1/2} \to  0\quad\text{as $\delta\to0$,}
$$
and the last integral in \eqref{estimate Dgamma} can be rewritten as
$$
\int_{\delta}^{+ \infty}|\xi|\frac{|h - \hb|}{r^2} \, dr= \int_{\delta}^{\sqrt{\delta}}|\xi|\frac{|h - \hb|}{r^2} \, dr + \int_{\sqrt{\delta}}^{1}|\xi|\frac{|h - \hb|}{r^2} \, dr + \int_{1}^{+ \infty}|\xi|\frac{|h - \hb|}{r^2} \, dr.
$$
By H\"{o}lder's inequality, we also obtain
$$
\aligned
\delta\int_{\delta}^{\sqrt{\delta}}|\xi|\frac{|h - \hb|}{r^2} \, dr
 \leq  \int_{\delta}^{\sqrt{\delta}}|\xi|\frac{|h - \hb|}{r} \, dr
& \leq  \Big(\int_{\delta}^{\sqrt{\delta}} \frac{\xi^2}{re^{\nu+ \lambda}} \, dr\Big)^{1/2} \Big(\int_{\delta}^{\sqrt{\delta}}e^{\nu+ \lambda} \frac{(h - \hb)^2}{r} \, dr\Big)^{1/2} 
\\
& = \Big(\int_{\delta}^{\sqrt{\delta}} \frac{\xi^2}{re^{\nu+ \lambda}} \, dr\Big)^{1/2} \Big(\frac{1}{4\pi}(e^{(\nu+ \lambda)(\sqrt{\delta})} -e^{(\nu+ \lambda)(\delta)})\Big)^{1/2}.
\endaligned
$$
Since, at almost all $u$, $\frac{\xi}{\sqrt{re^{\nu+ \lambda}}} \in L^2(0,r_0)$ for all $r_0> 0$, then $\int_{\delta}^{\sqrt{\delta}} \frac{\xi^2}{re^{\nu+ \lambda}} \, dr\to 0$ at almost all $u$ as $\delta\to 0$. Hence, by the previous inequality
$$
\delta\int_{\delta}^{\sqrt{\delta}}|\xi|\frac{|h - \hb|}{r^2} \, dr \to 0\quad \text{as $\delta\to0$}.
$$
Similarly, we obtain
$$
\aligned
\delta\int_{\sqrt{\delta}}^{1}|\xi|\frac{|h - \hb|}{r^2} \, dr
 \leq  \sqrt{\delta} \int_0^1|\xi|\frac{|h - \hb|}{r} \, dr
& \leq  \sqrt{\delta} \Big(\int_{0}^{1} \frac{\xi^2}{re^{\nu+ \lambda}} \, dr\Big)^{1/2} \Big(\int_{0}^{1}e^{\nu+ \lambda} \frac{(h - \hb)^2}{r} \, dr\Big)^{1/2} 
\\
& \leq  \Big(\frac{\delta}{4\pi} \Big)^{1/2} \Big(\int_{0}^{1} \frac{\xi^2}{re^{\nu+ \lambda}} \, dr\Big)^{1/2} \to  0\quad\text{as $\delta\to0 $}.
\endaligned
$$
To complete the proof of the lemma, it only remains to check that at almost all $u$
\bel{6th Dgamma}
\delta\int_{1}^{+ \infty}|\xi|\frac{|h - \hb|}{r^2} \, dr \to  0\quad\text{as $\delta\to0$}.
\ee
In fact, by  H\"{o}lder's inequality
$$
\aligned
\int_{1}^{+ \infty}|\xi|\frac{|h - \hb|}{r^2} \, dr& \leq  \Big(\int_{1}^{+ \infty} \frac{\xi^2}{r^4} \, dr\Big)^{1/2} \Big(\int_{1}^{+ \infty}(h - \hb)^2 \, dr\Big)^{1/2} 
\\
& \leq  \sup_{r\geq1} \Big(\frac{|\xi|}{r} \Big)\Big(\int_{1}^{+ \infty} \frac{1}{r^2} \, dr\Big)^{1/2} \Big(\int_{1}^{+ \infty}(h - \hb)^2 \, dr\Big)^{1/2}.
\endaligned
$$
Observe that, thanks to \eqref{energy-int}, 
$
\int_{1}^{+ \infty}(h - \hb)^2 \, dr
$
is bounded by a continuous function of $u$. Furthermore, by Proposition \ref{key of Dgamma}, $\sup_{r\geq 1} \frac{|\xi|}{r}$ is finite for almost all $u$. Therefore, $\int_{1}^{+ \infty}|\xi|\frac{|h - \hb|}{r^2} \, dr$ is  finite for almost all $u$, hence \eqref{6th Dgamma} holds for almost all $u$, as desired.
\end{proof}


We have established that, for almost all $u$ and all $r> 0$, $f_\delta \to  0$ as $\delta\to0$. It then follows from \eqref{evolution gb} that
\bel{evolution gb delta=0 1}
\lim_{\delta\to0}D\Big(\frac{1}{r} \int_{\delta}^{r}g \, ds\Big) = -  \frac{1}{2r} \gb(g- \gb) -4\pi \gb D\gamma+ \frac{2\pi}{r} \xi^2.
\ee
Furthermore, for any smooth function $\varphi$ whose support is a compact set contained in the interior of $\Ocal$, we have
$$
- \iint_{\Ocal} \Big(\frac{1}{r} \int_{\delta}^rg \, ds\Big)D\varphi \, dudr= \iint_{\Ocal}D\Big(\frac{1}{r} \int_{\delta}^{r}g \, ds\Big)\varphi \, dudr.
$$
Since $\frac{1}{r} \int_{\delta}^{r}g \, ds \to  \gb$ for $\delta\to0$ uniformly in the support of $\varphi$, it follows that
\bel{evolution gb delta 2}
- \iint_{\Ocal} \gb D\varphi \, dudr= \lim_{\delta\to0} \iint_{\Ocal}D\Big(\frac{1}{r} \int_{\delta}^{r}g \, ds\Big)\varphi \, dudr.
\ee
On the other hand,
 from Proposition \ref{key of Dgamma} and the proof of Lemma \ref{lemma-regularity of rDgamma} 
 we deduce 
 that $\varphi f_\delta$ is dominated by an integrable function and converges to $0$ almost everywhere as $\delta\to0$. Therefore, by the dominated convergence theorem we have
\bel{evolution gb delta 3}
\lim_{\delta\to0} \iint_{\Ocal}D\Big(\frac{1}{r} \int_{\delta}^{r}g \, ds\Big)\varphi \, dudr= \iint_{\Ocal} \lim_{\delta\to0}D\Big(\frac{1}{r} \int_{\delta}^{r}g \, ds\Big)\varphi \, dudr.
\ee
Hence, we conclude from \eqref{evolution gb}, Lemma \ref{lemma-regularity of rDgamma} and \eqref{evolution gb delta=0 1} --\eqref{evolution gb delta 3} that $\gb$ is weakly differentiable in $\Ocal$ and
$$
D\gb= \lim_{\delta\to0}D\Big(\frac{1}{r} \int_{\delta}^{r}g \, ds\Big) = -  \frac{1}{2r} \gb(g- \gb) -4\pi \gb D\gamma+ \frac{2\pi}{r} \xi^2.
$$
This means that the Hawking mass $m$ is weakly differentiable in $\Ocal$ and 
$$
\aligned
Dm & =  - \frac{r}{2e^{\nu+ \lambda}} \Big(D\gb- \frac{\gb}{e^{\nu+ \lambda}}D\big(e^{\nu+ \lambda} \big)\Big) - \frac{\gb(g- \gb)}{4e^{\nu+ \lambda}}
- c^2\pi r^2\hb^2\gb
   =  - \frac{\pi}{r e^{\nu+ \lambda}} \xi^2- c^2\pi r^2\hb^2\gb.
\endaligned
$$
Finally, we also check similarly 
 that  $\int_0^r\Big(\frac{e^{\nu+ \lambda}}{\gb} \, ds\Big)$ is weakly differentiable in $\Ocal$ and 
$$
D\Big(\int_{0}^{r} \frac{e^{\nu+ \lambda}}{\gb} \, dr\Big) = - 2\pi \, \int_{0}^{r} \Big(\frac{e^{\nu+ \lambda}}{\gb^2} \frac{\xi^2}{r} \Big) \, dr - \frac{1}{2}e^{(\nu+ \lambda)(u,0)}.
$$
Integrating this equation along the characteristics $\chi (u; u_1, r_1)$, we obtain
$$
\aligned
& \int_0^{r_1} \Big(\frac{e^{\nu+ \lambda}}{\gb} \Big)(u_1,r) \, dr + 2\pi \, \iint_{\Ocal(u_1,r_1)} \Big(\frac{e^{\nu+ \lambda}}{\gb^2} \frac{\xi^2}{(r+ \eps)} \Big) \, drdu+ \frac{1}{2} \int_0^{u_1}e^{(\nu+ \lambda)(u,0)} \, du
  = \int_{0}^{r_{0}} \Big(\frac{e^{\nu+ \lambda}}{\gb} \Big)(0,r) \, dr,
\endaligned
$$
with $
\Ocal(u_1,r_1) = \big\{ (u,r) \, / \, 0<u<u_1,\,0<r<\chi (u; u_1,r_1)\big\}.
$
Consequently, the generalized solution exists globally and this concludes the proof of Theorem \ref{thrgensol}.


\small

\paragraph*{Acknowledgments.}

This work was done in part when PLF was a visiting research fellow at the Courant Institute for Mathematical Sciences, New York University, and a visiting professor at the School of Mathematical Sciences, Fudan University, Shanghai. 
The second author was supported by FCT sabbatical grant SFRH/BSAB/130242/2017, CMAT, Univ. Minho, through FCT/Portugal projects UIDB/00013/2020 and FEDER Funds COMPETE as well as CAMGSD, IST-ID, through projects UIDB/04459/2020 and UIDP/04459/2020. 
The third author (TCN) was a postdoc fellow of the Fondation des Sciences Math\'ematiques de Paris (FSMP). 



\appendix  

\section{Continuity and compactness of the solution mapping}
\label{appendix-one} 

\paragraph{Continuity of the mapping $\Psi_b$.} 

Assume that $(t_i, h_{i}) \to  (t,h)$ in $[0,1]\times X_b$ (that is, in the $C^1$ norm). Let $k$ and $k_i$ be solutions of the transport equation \eqref{main equation} associated with $(t, h)$ and $(t_i, h_i)$, respectively. Let us prove that $k_i$ converges to $k$ in the $C^1$ norm.
Since the set of all smooth functions in $X_b$ is dense, we can pick up a sequence of smooth functions $\{z_i\} \subset X_b$ such that $\|h_i-z_i\|_{C^1} \to  0$. We then define $w_i$ to be the unique solution of \eqref{main equation} associated with $(t_i,\,z_i)$. Obviously, we have 
\bel{triangle}
\|k_i-k\|_{C^1} \leq \|k_i-w_i\|_{C^1} + \|w_i-k\|_{C^1}
\ee
and we analyze first the term $\|w_i-k\|_{C^1}$. By definition, we have 
\bel{main equation i}
\aligned
D_{a,\,t_i,z_i}w_i- \frac{t_i}{2(r+ \eps)}(g_{a,\,t_i,z_i} - \gb_{a,\,t_i,z_i})w_i
 & =  - \frac{1}{2(r+ \eps)}(g_{a,\,t_i,z_i} - \gb_{a,\,t_i,z_i})\wb_i
     - \frac{c^2(r+ \eps)}{2}e^{\nu_{z_i} + \lambda_{z_i}} \wb_i,
\\ 
 D_{a,t,h}k- \frac{t}{2(r+ \eps)}(g_{a,t,h} - \gb_{a,t,h})k 
& =  - \frac{1}{2(r+ \eps)}(g_{a,t,h} - \gb_{a,t,h}) \, \hb 
          - \frac{c^2(r+ \eps)}{2}e^{\nu_{h} + \lambda_{h}} \hb
\endaligned
\ee
and, therefore, 
\bel{subtract C0}
D_{a,t,h}(w_i-k) - \frac{t}{2r}(g_{a,t,h} - \gb_{a,t,h})(w_i-k) =P(t, h,\,t_i,\,z_i),
\ee
with
\bel{def P(.,.,.,.)}
\aligned
P(t, h,\,t_i,\,z_i)&:= \frac{1}{2(r+ \eps)}(g_{a,t,h} - \gb_{a,t,h}) \, \hb- \frac{1}{2(r+ \eps)}(g_{a,\,t_i,z_i} - \gb_{a,\,t_i,z_i})\wb_i
\\
& \quad + \frac{c^2(r+ \eps)}{2}e^{\nu_{h} + \lambda_{h}} \hb- \frac{c^2(r+ \eps)}{2}e^{\nu_{z_i} + \lambda_{z_i}} \wb_i
+ \frac{1}{2} \big(\rho(\gb_{a,\,t_i,z_i}) - \rho(\gb_{a,t,h}) \big)\del_rw_i
\\
& \quad + \frac{1}{2(r+ \eps)} \big(t_i(g_{a,\,t_i,z_i} - \gb_{a,\,t_i,z_i}) -t(g_{a,t,h} - \gb_{a,t,h})\big)w_i.
\endaligned
\ee
Since $\|h-z_i\|_{C^1}$ and $|t_i-t|$ tend to $0$, it follows from the continuity of $\Psi_b$ and the $C^2$ bound on $k$ that
\bel{right hand C0}
P(t, h,\,t_i,\,z_i)  \to  0 \qquad \mbox{in $C^0$}.
\ee
On the other hand, since $w_i(0,r) =k(0,r) =h_0(r)$, along the characteristics $\chi_{h,u_1} =\chi_{a,t,h}(\cdot; u_1, r_1)$ we can write the transport equation \eqref{subtract C0} in the form 
\bel{eq C0}
(w_i-k)(u_1,r(u_1)) = \int_0^{u_1}[P(t, h,\,t_i,\,z_i)]_{\chi_{h,u_1}}e^{\int_{u_2}^{u_1} 
 \frac{t}{2(r_3+ \eps)}(g_{a,t,h} - \gb_{a,t,h}) \big|_{\chi_{h,u_1}} \, du_3} \, du_2.
\ee
Hence, from \eqref{right hand C0} and \eqref{eq C0} we deduce that 
\bel{converge C0}
\mbox{$w_i \to  k$ in $C^0$.}
\ee
Differentiating now the equations \eqref{main equation i} with respect to $r$, we obtain
\begin{equation*}
\label{second der}
\aligned
& D_{a,t,h}(\del_rk) - \frac{1}{2(r+ \eps)} \Big( (g_{a,t,h} - \gb_{a,t,h})\rho^\prime(\gb_{a,t,h})+t(g_{a,t,h} - \gb_{a,t,h}) \Big)\del_rk=  f_1(h,k,t)  
\\
& D_{a,\,t_i,z_i}(\del_rw_i) - \frac{1}{2(r+ \eps)} \Big( (g_{a,\,t_i,z_i} - \gb_{a,\,t_i,z_i})\rho^\prime(\gb_{a,\,t_i,z_i})+t_i(g_{a,\,t_i,z_i} - \gb_{a,\,t_i,z_i}) \Big)\del_rw_i= f_1(z_i,w_i,t_i), 
\endaligned
\end{equation*}
where $f_1$ is given as in \eqref{def of f_1}. Comparing the previous two equations we find
\bel{subtract C1}
\aligned
D_{a,t,h}(\del_rw_i - \del_rk) -& \frac{1}{2(r+ \eps)} \Big(
(g_{a,t,h} - \gb_{a,t,h})\rho^\prime(\gb_{a,t,h})
+ t(g_{a,t,h} - \gb_{a,t,h})\Big)
\, (\del_rw_i- \del_rk) = \widetilde{P}_i(t, h,\,t_i,\,z_i),
\endaligned
\ee
with
 \begin{equation*}
 \aligned
 \label{def wP(.,.,.,.)}
\widetilde{P}(t, h,\,t_i,\,z_i)&:= f_1(h,k,t) -f_1(z_i,w_i,t_i)+ \frac{1}{2}(\rho(\gb_{a,\,t_i,z_i}) - \rho(\gb_{a,t,h}))\del_{rr}w_i
\\
& \quad + \frac{1}{2(r+ \eps)}  \Big( (g_{a,\,t_i,z_i} - \gb_{a,\,t_i,z_i})\rho^\prime(\gb_{a,\,t_i,z_i}) -(g_{a,t,h} - \gb_{a,t,h})\rho^\prime(\gb_{a,t,h}) \Big)\del_{r}w_i 
\\
& \quad + \frac{1}{2(r+ \eps)}  \Big( t_i(g_{a,\,t_i,z_i} - \gb_{a,\,t_i,z_i}) -t(g_{a,t,h} - \gb_{a,t,h}) \Big)\del_{r}w_i.
\endaligned
\end{equation*}
Since $\del_rw_i(0,r) = \del_rk(0,r) = \del_r h_0$, then solving \eqref{subtract C1} along the characteristic $\chi_{h,u_1}$ yields us 
\be
\aligned
\label{formula integration C1}
& \del_r (w_i-k)(u_1,r(u_1))
  = \int_0^{u_1} \widetilde{P}(t, h,\,t_i,\,z_i) \big|_{\chi_{h,u_1}} 
  e^{\int_{u_2}^{u_1}
   \frac{1}{2(r_3+ \eps)} \big((g_{a,t,h} - \gb_{a,t,h})\rho^\prime(\gb_{a,t,h})+t(g_{a,t,h} - \gb_{a,t,h})\big) \big|_{\chi_{h,u_1}} \, du_3} \, du_2.
\endaligned
\end{equation}
From the basic continuity properties in Section \ref{Lipschitz}, we see that $\|\del_{rr}w_i\|_{C^0}$ and $\|\del_{r}w_i\|_{C^0}$ are uniformly bounded. We deduce from \eqref{converge C0} and \eqref{formula integration C1} that $\widetilde{P}(t,h,t_i,z_i) \to  0$, hence, by \eqref{formula integration C1}, $\del_rw_i$ converges to $\del_rk$ in $C^{0}$. Combining this result with \eqref{converge C0}, we find that 
$\|w_i-k\|_{C^1} \to  0$ and, similarly, $\|w_i-k_i\|_{C^1} \to  0$. Therefore, returning to \eqref{triangle} we conclude that $k_i$ converges to $k$ in the $C^1$ norm, as claimed.


\paragraph{Compactness of the mapping $\Psi_b$.}

Given any bounded sequence $(t_i, h_i)$ in $[0,1]\times X_b$, let us denote by $k_i$ the unique solution of the transport equation \eqref{main equation} associated with $(t_i, h_i)$. We need to prove that there exists a subsequence (still denoted by $\{k_i\}$) such that $\{k_i\}$ converges in the space $X_b$.
Without loss of generality we assume that $h_i$ is smooth. Since $\{h_i\}$ is bounded in $C^1$, a  standard compact embedding theorem tells us that there exists a subsequence (still denoted by $\{h_i\}$) such that $\{h_i\}$ converges in the $C^0$ norm. In particular, given any $\alpha> 0$  we can guarantee that 
\bel{Cauchy C0}
|h_i- h_{j}| \leq \alpha \quad \text{ for all sufficiently large } i,j. 
\ee   
After passing to a subsequence, we may assume that $|t_i-t_j| \leq \alpha$. 
An analysis similar to the one used to prove the continuity of $\Psi_b$ shows that, along the characteristics $\chi_{h_j,u_1}=\chi_{a,t_j,h_j}(\cdot; u_1,r_1)$,
$$  
\aligned
 (k_i-k_j)(u_1,r(u_1))
& = \int_0^{u_1}
 P(t_i, h_i,\,t_j, h_j)\big|_{\chi_{h_j,u_1}}e^{\int_{u_2}^{u_1}  
  \frac{t_j}{2(r_3+ \eps)}(g_{a,\,t_j,h_j} - \rho(\gb_{a,\,t_j,h_j})) \big|_{\chi_{h_j,u_1}} \, du_3} \, du_2, 
\\ 
 \del_r(k_i-k_j)(u_1,r(u_1))
& = \int_0^{u_1} \hspace{-0.2cm}
\widetilde{P}(t_i, h_i,\,t_j, h_j) \big|_{\chi_{h_j,u_1}}
e^{\int_{u_2}^{u_1} 
\frac{1}{2(r_3+ \eps)} \big((g_{a,\,t_j,h_j} - \gb_{a,\,t_j,h_j})\rho^\prime(\gb_{a,\,t_j,h_j})+t_j(g_{a,\,t_j,h_j} - \gb_{a,\,t_j,h_j})\big) \big|_{\chi_{h_j,u_1}} du_3} 
 \, du_2, 
\endaligned
$$
where $P$ and $\widetilde{P}$ are defined as before. 
Thanks to the continuity of $\Psi_b$ and the $C^2$ bound available on $k$, this implies 
that there exists a constant $a_7> 0$ depending on $\|h_i\|_{C^1}$ and $(h_0,a,b,\eps)$ such that 
$\|k_i-k_j\|_{C^{1}} \leq a_7\alpha$. 
We conclude that $\{k_i\}$ converges in $C^1$, and this completes the proof of Proposition~\ref{lemma Psi cc}.


\section{Derivation of energy estimates} 
\label{Appendix-third} 

We give here a proof of Proposition~\ref{lemma-energy-difference}. 
We start from the identity
\begin{equation*}
\del_r (D \hb_{\eps,b}) + \frac{1}{r} D \hb_{\eps,b} = \frac{1}{2r^2} e^{\nu_{\eps,b} - \lambda_{\eps,b}}(h_{\eps,b} - \hb_{\eps,b}) - \frac{c^2}{2} e^{\nu_{\eps,b} + \lambda_{\eps,b}} \hb_{\eps,b}, 
\end{equation*}
which by integration gives us 
\bel{Dhbar}
D_\eps \hb_{\eps,b} =  \frac{1}{2(r+ \eps)} \int_0^r \left(\gb_{\eps,b} \frac{h_{\eps,b} - \hb_{\eps,b}}{r+ \eps} - c^2(r+ \eps) \, \hb_{\eps,b} e^{\nu_{\eps,b} + \lambda_{\eps,b}}  \right ) dr.
\ee
Thanks to Lemma \ref{Fzeta epsilon=0}, on the spacetime domain $[0,c_1]\times [0,b]$ we have 
$$
\del_r\gb_{\eps,b} +(r+ \eps)\del_r\big(\gb_{\eps,b} \del_r\nu_{\eps,b} \big) -4\pi \, \del_{u} \big((h_{\eps,b} - \hb_{\eps,b})^2\big) -8\pi(r+ \eps)\big(\del_r \, \hb_{\eps,b}D_\eps \hb_{\eps,b} - \frac{c^2}{2}e^{\nu_{\eps,b} + \lambda_{\eps,b}} \hb_{\eps,b}^2\big) =0.
$$
Integrating this identity and using \eqref{Dhbar} we obtain
\be
\aligned
\label{formula energy hb}
 4 \pi \, \del_u\bigg(\int_0^{b}(h_{\eps,b} - \hb_{\eps,b})^2 \, dr\bigg)
& = \int_0^{b} \bigg(\del_r\gb_{\eps,b} +(r+ \eps)\del_r\big(\gb_{\eps,b} \del_r\nu_{\eps,b} \big)\bigg) dr
 -8\pi \, \int_0^b(r+ \eps)\big(\del_r  \hb_{\eps,b}D_\eps \hb_{\eps,b} - \frac{c^2}{2}e^{\nu_{\eps,b} + \lambda_{\eps,b}} \hb_{\eps,b}^2\big) dr 
\\
& \leq  1+ \bigg((r+ \eps)\gb_{\eps,b} \del_r\nu_{\eps,b} - \int_0^{r} \gb_{\eps,b} \del_r\nu_{\eps,b} \, ds\bigg)(b)+4\pi c^2\int_0^{b}(r+ \eps) \, \hb_{\eps,b}^2e^{\nu_{\eps,b} + \lambda_{\eps,b}} \, dr 
\\
& \quad -4\pi \, \int_0^{b} \del_r \, \hb_{\eps,b} \bigg(\int_0^{r} \gb_{\eps,b} \del_r \, \hb_{\eps,b} \, ds- c^2\int_0^{r}(s+ \eps) \, \hb_{\eps,b}e^{\nu_{\eps,b} + \lambda_{\eps,b}} \, ds\bigg) dr.
\endaligned
\ee
We now need to estimate each of the terms above. Observe that by the first equation in \eqref{formaula derivative of nu}
$$
\aligned
(b+ \eps)\big(\gb_{\eps,b} \del_r\nu_{\eps,b} \big)(b) & =  \bigg(\frac{g_{\eps,b} - \gb_{\eps,b}}{2} + 2\pi \, \gb_{\eps,b}(h_{\eps,b} - \hb_{\eps,b})^2\bigg)(b)
\\
\int_0^b\gb_{\eps,b} \del_r\nu_{\eps,b} \, ds& =  \int_0^b\bigg(\frac{1}{2} \del_r\gb_{\eps,b} + 2\pi \, \gb_{\eps,b} \frac{(h_{\eps,b} - \hb_{\eps,b})^2}{s+ \eps} \bigg) \, ds\geq 0,
\endaligned
$$
hence 
\bel{estimate nu}
\aligned
\bigg((r+ \eps)\gb_{\eps,b} \del_r\nu_{\eps,b} - \int_0^{r} \gb_{\eps,b} \del_r\nu_{\eps,b} \, ds\bigg)(b)& \leq  \bigg(\frac{g_{\eps,b} - \gb_{\eps,b}}{2} + 2\pi \, \gb_{\eps,b}(h_{\eps,b} - \hb_{\eps,b})^2\bigg)(b). 
\endaligned
\ee
Next, we will use \eqref{equa----415} for estimating $\int_0^{+ \infty}(r+ \eps) \, \hb_{\eps,b}^2e^{\nu_{\eps,b} + \lambda_{\eps,b}} \, dr$. Indeed, we note first that if there exists $(u_\eps,r_\eps) \in [0,c_1]\times [0,3]$ such that
$$
(r_\eps+ \eps) \big(|\hb_{\eps,b}|e^{\frac{\nu_{\eps,b} + \lambda_{\eps,b}}{2}} \big)(u_\eps,r_\eps)\geq \frac{2}{c} \sqrt{\frac{m_{\eps,b}(u,b)}{2\pi}},
$$
thanks to \eqref{equa----415}, we may take $r_\eps^\prime\in (r_\eps,r_\eps+1]$ such that
$$
(r_\eps^\prime+ \eps) \big(|\hb_{\eps,b}|e^{\frac{\nu_{\eps,b} + \lambda_{\eps,b}}{2}} \big)(u_\eps,r_\eps^\prime)\leq \frac{1}{c} \sqrt{\frac{m_{\eps,b}(u,b)}{2\pi}}.
$$
Consequently, we have 
$
|\hb_{\eps,b}|(u_\eps,r_\eps)\geq 2|\hb_{\eps,b}|(u_\eps,r_\eps^\prime),
$
and thus
\bel{control rhb 1}
|\hb_{\eps,b}(u_\eps,r_\eps) - \hb_{\eps,b}(u_\eps,r_\eps^\prime)|\geq \frac{1}{2}|\hb_{\eps,b}(u_\eps,r_\eps)|.
\ee
Now, assume that there exists $r< 1- \eps$ satisfying
\bel{treat with r<1}
\sqrt{r+ \eps} \big(\hb_{\eps,b} e^{\frac{\nu_{\eps,b} + \lambda_{\eps,b}}{2}} \big)(u,r)\geq \frac{4}{c^2} \sqrt{\frac{m_{\eps,b}(u,b)}{2\pi}}.
\ee
If 
$
|\hb_{\eps,b}(u,r) - \hb_{\eps,b}(u,1- \eps)| \leq \frac{1}{2}|\hb_{\eps,b}(u,r)|,
$
we obtain
$$
\big(|\hb_{\eps,b}|e^{\frac{\nu_{\eps,b} + \lambda_{\eps,b}}{2}} \big)(u,1- \eps)\geq \frac{\sqrt{r+ \eps}}{2} \big(|\hb_{\eps,b}|e^{\frac{\nu_{\eps,b} + \lambda_{\eps,b}}{2}} \big)(u,r)\geq \frac{2}{c^2} \sqrt{\frac{m_{\eps,b}(u,b)}{2\pi}},
$$
hence, by \eqref{control rhb 1}, there exists $r^\prime\in (1- \eps,2- \eps]$ satisfying 
$$
|\hb_{\eps,b}(u,1- \eps) - \hb_{\eps,b}(u,r^\prime)|\geq \frac{1}{2}|\hb_{\eps,b}(u,1- \eps)|.
$$
This tells us that if such $r$ exists (i.e. $r<1- \eps$ and satisfies \eqref{treat with r<1}),  then there is $r_*\in (r,2- \eps]$ such that
\bel{for r<1}
|\hb_{\eps,b}(u,r) - \hb_{\eps,b}(u,r_*)|\geq \frac{1}{4}|\hb_{\eps,b}(u,r)|.
\ee
Combined with the fact that (cf. also \eqref{estimate nu+lambda})
$$
4\pi \int_{r}^{b}(s+ \eps)(\del_r \, \hb_{\eps,b})^2ds\geq 4\sqrt{\pi} \bigg|\int_{r}^{r_*} \del_r \, \hb_{\eps,b} \, ds\bigg|- \int_{r}^{r_*} \frac{1}{s+ \eps} \, ds+4\pi \int_{r_*}^{b}(s+ \eps)(\del_r \, \hb_{\eps,b})^2ds,
$$
it then follows that 
\bel{estimate hb e}
(r+ \eps) \big(|\hb_{\eps,b}|^2e^{\nu_{\eps,b} + \lambda_{\eps,b}} \big)(u,r)
\leq 2 \, |\hb_{\eps,b}(u,r)|^2e^{- \sqrt{\pi} \, |\hb_{\eps,b}(u,r) |}
\leq \frac{8}{e^2\pi}.
\ee
Therefore, from \eqref{equa----415} and \eqref{estimate hb e} we find 
\bel{rhb2 e}
\aligned
c^2\int_0^{b}(r+ \eps) \, \hb_{\eps,b}^2e^{\nu_{\eps,b} + \lambda_{\eps,b}} \, dr
& =  c^2\int_0^{1}(r+ \eps) \, \hb_{\eps,b}^2e^{\nu_{\eps,b} + \lambda_{\eps,b}} \, dr +c^2\int_1^{b}(r+ \eps) \, \hb_{\eps,b}^2e^{\nu_{\eps,b} + \lambda_{\eps,b}} \, dr,
\\
& \leq  8\max\Big( \frac{c^2}{e^2\pi},\,\frac{m_{\eps,b}(u,b)}{\pi} \Big) + \frac{m_{\eps,b}(u,b)}{2\pi}.
\endaligned
\ee
We now consider the last term of \eqref{formula energy hb}. Setting
$$
J_\eps:= - 4\pi \, \int_0^{b} \del_r \, \hb_{\eps,b} \bigg(\int_0^{r} \gb_{\eps,b} \del_r \, \hb_{\eps,b} \, ds+c^2\int_0^{r}(s+ \eps) \, \hb_{\eps,b}e^{\nu_{\eps,b} + \lambda_{\eps,b}} \, ds\bigg) \, dr,
$$
we have
$$
\aligned 
J_\eps&= -4\pi \, \hb_{\eps,b}(b) \int_0^b\gb_{\eps,b} \frac{h_{\eps,b} - \hb_{\eps,b}}{s+ \eps} \, ds+4\pi c^2\hb_{\eps,b}(b) \int_0^b(s+ \eps) \, \hb_{\eps,b}e^{\nu_{\eps,b} + \lambda_{\eps,b}} \, ds
\\
 & \quad + 2\pi \, \int_0^{b} \del_r(\hb_{\eps,b}^2)\gb_{\eps,b}  dr +4\pi c^2\int_0^{b}(r+ \eps) \, \hb_{\eps,b}^2e^{\nu_{\eps,b} + \lambda_{\eps,b}}  dr
 \\
 &\leq   -2\pi \, \bigg(\hb_{\eps,b}^2(b)\gb_{\eps,b}(b) -2\hb_{\eps,b}(b) \int_0^b\hb_{\eps,b} \frac{g_{\eps,b} - \gb_{\eps,b}}{s+ \eps}  ds\bigg)
 +4\pi c^2 |\hb_{\eps,b}(b)|\sqrt{b} \bigg(\int_0^b(s+ \eps)^2\hb_{\eps,b}^2e^{2(\nu_{\eps,b} + \lambda_{\eps,b})} \, ds\bigg)^{1/2} 
\\
& \quad -2\pi \, \int_0^{b} \hb_{\eps,b}^2\del_r\gb_{\eps,b}  dr +4\pi c^2\int_0^{b}(r+ \eps) \, \hb_{\eps,b}^2e^{\nu_{\eps,b} + \lambda_{\eps,b}} \, dr, 
\endaligned
$$
hence 
\be
\aligned
\label{estimate Dhb}
J_\eps
& \leq  2m_{\eps,b}(u,b)+32\max\big\{c^2e^{-2},m_{\eps,b}(u,b)\big\} + 33\pi c^2\bigg(\max\big\{ 4\pi A^4(u,0),\,e^{-2\sqrt{2} \pi^{1/4}} \big\} \bigg)
\\
& \quad +  4\pi \bigg(\frac{A^2(u,0)}{b} \int_0^b\frac{g_{\eps,b} - \gb_{\eps,b}}{(r+ \eps)^2} \, dr + \frac{|c|A(u,0)}{\sqrt{b}} \sqrt{\frac{m_{\eps,b}(u,b)}{2\pi}} \bigg), 
\endaligned
\end{equation}
where we the first inequality follows from H\"{o}lder's inequality and the second inequality follows from Lemma \ref{lemma-rhb}, \eqref{estimate (g-gb)hb2}, \eqref{equa----415} and \eqref{rhb2 e}.

Now, taking \eqref{estimate nu}, \eqref{rhb2 e} and \eqref{estimate Dhb} into \eqref{formula energy hb}, we obtain
$$
\aligned 
&\del_u\bigg(\int_0^{b}(h_{\eps,b} - \hb_{\eps,b})^2 \, dr\bigg)
  \leq  \frac{1}{4\pi} + \frac{17}{\pi} \max\Big(
  c^2e^{-2},\,m_{\eps,b}(u,b)\Big) 
   + \frac{33 c^2}{4} \bigg(\max\big\{ 4\pi A^4(u,0),\,e^{-2\sqrt{2} \pi^{1/4}} \big\} \bigg)
\\
&  + \frac{1}{4\pi} \bigg(\frac{g_{\eps,b} - \gb_{\eps,b}}{2} + 2\pi \, \gb_{\eps,b}(h_{\eps,b} - \hb_{\eps,b})^2\bigg)(u,b)
 + \bigg(\frac{A^2(u,0)}{b} \int_0^b\frac{g_{\eps,b} - \gb_{\eps,b}}{(r+ \eps)^2} \, dr + \frac{|c|A(u,0)}{\sqrt{b}} \sqrt{\frac{m_{\eps,b}(u,b)}{2\pi}} \bigg),
\endaligned
$$
hence
$$
\aligned
\int_0^{b}(h_{\eps,b} - \hb_{\eps,b})^2 \, dr& \leq  \int_0^{b}(h_0- \hb_0)^2 \, dr + \frac{u}{4\pi}
+ \frac{17}{\pi} \int_0^u\bigg(\max\big\{c^2e^{-2},\,m_{\eps,b}(u_1,b)\big\} \bigg) \, du_1
\\
& \quad + \frac{33 c^2}{4}  \int_0^u\bigg(\max\big\{ 4\pi A(u_1,0)^4,\,e^{-2\sqrt{2} \pi^{1/4}} \big\} \bigg) \, du_1
 + \frac{1}{4\pi} \int_0^u\bigg(\frac{g_{\eps,b} - \gb_{\eps,b}}{2} + 2\pi \, \gb_{\eps,b}(h_{\eps,b} - \hb_{\eps,b})^2\bigg)(u_1,b) \, du_1
\\
& \quad + \int_0^u\bigg(\frac{A^2(u_1,0)}{b} \int_0^b\frac{g_{\eps,b} - \gb_{\eps,b}}{(r+ \eps)^2} \, dr + \frac{A(u_1,0)}{\sqrt{b}} \sqrt{\frac{m_{\eps,b}(u_1,b)}{\pi}} \bigg) \, du_1.
\endaligned
$$
Therefore, by Lemma \ref{lemma-rhb} we have  
\be
\label{(h- hb)2}
\aligned
 \int_0^{+ \infty}(h_{\eps,b} - \hb_{\eps,b})^2 \, dr 
& \leq   	\int_0^{b}(h_{\eps,b} - \hb_{\eps,b})^2 \, dr +	\int_b^{b+ 2}(h_{\eps,b} - \hb_{\eps,b})^2 \, dr
 +	(b+ 2+ \eps)^2|\hb_{\eps,b}|^2(u,b+ 2+ \eps) \int_{b+ 2}^{+ \infty} \frac{1}{(r+ \eps)^2} \, dr
\\
& \leq \int_0^{b}(h_0- \hb_0)^2 \, dr + \frac{u}{4\pi}
+ \frac{17}{\pi} \int_0^u\bigg(\max\big\{c^2e^{-2},\,m_{\eps,b}(u_1,b)\big\} \bigg) \, du_1
\\
& \quad  + \frac{33 c^2}{4} \int_0^u\bigg(\max\Big( 
 4\pi A(u_1,0)^4,\,e^{-2\sqrt{2} \pi^{1/4}} \Big)
  \bigg) \, du_1+B_2(u,b,c_1) .
\endaligned
\ee
This completes the proof of Proposition~\ref{lemma-energy-difference}. 


\section{Proof of technical lemmas} 
\label{section-technical-lemmas}

\begin{proof}[Proof of Lemma~\ref{lemma-rhb small at 0}] 
By the H\"{o}lder inequality we have 
$
|\hb_\eps|\leq  \frac{1}{r} \int_0^r| h_\eps (u,s)| \, ds
\leq  \frac{1}{\sqrt r} \left(\int_0^r h^2_\eps(u,s) \, ds\right)^{1/2},
$
hence
$$
\sup_{r\ge \delta} |\hb_\eps (u,r)| \leq \frac{1}{\sqrt \delta} \left(\int_0^{+ \infty} h^2_\eps(u,s) \, ds\right)^{1/2},
$$
and the result then follows from the inequality (3) in \eqref{energy-int}.
\end{proof}

\begin{proof}[Proof of Lemma~\ref{uniform bound w.r.t. u}]
From the proof of Theorem~\ref{lemma-global epsilon existence} it follows that 
$$
\lim_{r\to + \infty} \big(\sup_{u\in[0,u_0]}|h_\eps|\big) = \lim_{r\to + \infty} \big(\sup_{u\in[0,u_0]}|\del_r h_\eps|\big) =0,
$$
so it only remains to check that
\be
\lim_{r\to + \infty} \big(\sup_{u\in[0,u_0]}|(r+ \eps) \, \hb_\eps (u,r) |\big) =0.
\ee
We argue by contradiction. Suppose that there exists a sequence $(u_i,r_i)$ in $[0,u_0]\times [0,+ \infty)$ converging to $(u_\infty,+ \infty)$ such that
$$
\alpha:= \lim\big(|(r_i+ \eps) \, \hb_\eps|(u_i,r_i)\big) > 0.
$$
For any $\zeta> 0$, let $u_\zeta\in [0,u_0]$ be such that $|u_\zeta-u_\infty| \leq \zeta$. Observe that the inequalities (3) and (5) in \eqref{energy-int}
tell us that for, each $u$,  $(r+ \eps)h_\eps(u_\zeta,r) \in W^{1,2}$, hence by the Sobolev embedding theorem
\bel{decay rhb}
\int_{0}^{+ \infty}h_\eps(u_\zeta,r) \, dr=0.
\ee
Next, similarly to \eqref{control rhb in C0}, we have
$$
(r_i+ \eps)|\hb_\eps|(u_i,r_i)\leq  \Big\|\int_0^{r_i}h_\eps(u_\zeta,r) \, dr\Big\|_{C^0([r_i,+ \infty))} 
+ \sqrt{\frac{\zeta M_\eps(0)}{2\pi}},
$$
which, in combination with \eqref{decay rhb}, gives us 
$$
\alpha= \lim\big(|(r_i+ \eps) \, \hb_\eps|(u_i,r_i)\big)\leq \sqrt{\frac{\zeta M_\eps(0)}{2\pi}},
$$
which is a contradiction if $\zeta\leq \frac{\alpha^2\pi}{M_\eps(0)}$.
\end{proof}


\begin{proof}[Proof of Lemma~\ref{lemma equi-continuity of delr h}]
In view of the inequality (2) in \eqref{equa-1234}, 
\eqref{g-gb}, \eqref{control h in C0} and \eqref{control del_rg}, we obtain from \eqref{evolution of del_rh} on $\Ocal_\delta$ that
\bel{D(delr h)}
\big|D_\eps(\del_r h_\eps)\big| \leq K_2(u,\delta),
\ee
where
$$
\aligned
K_2(u,\delta)
& := 
\frac{2M_\eps(0)}{\delta}B_8(u,\delta)
+ \frac{c^2}{2}B_5(u,\delta)
  +
 \left(B_5+ \frac{\Acal_\eps}{\delta} \right)\left(\frac{2\pi c^2 \Acal_\eps}{\delta} \left(B_5+ \frac{\Acal_\eps}{\delta} \right)+ \frac{B_7}{2\delta} + \frac{3M_\eps(0)}{\delta^3} \right)(u,\delta).
\endaligned
$$
Thanks to \eqref{D(delr h)}, we only need to prove that $\{\del_r h_\eps\}$ is equicontinuous in $\Ocal_\delta$  with respect to $r$. In fact, uniform continuity with respect to $u$ would then follow by the following arguments. Given $u_1,r_1$ and $\Delta u_1> 0$, let $r_1+ \Delta r_{1,\eps}$ be the value of $r$ at which the characteristic $\chi_\eps(u; u_1, r_1)$ intersects the line $u=u_1- \Delta u_1$. Observe that $\Delta r_{1,\eps} \leq 1/2\Delta u_1$ and then
$$
\aligned
\big|\del_r h_\eps(u_1,r_1) - \del_r h_\eps(u_1- \Delta u_1,r_1)\big|
& \leq  \int_{u_1- \Delta u_1}^{u_1}|D_\eps(\del_r h_\eps)|du+ \big|\del_r h_\eps(u_1- \Delta u_1,r_1+ \Delta r_{1,\eps}) - \del_r h_\eps(u_1- \Delta u_1,r_1)\big|
\\
& \leq  K_2(u_0,\delta)\Delta u_1 + \eta\Delta r_{1,\eps}
 \leq   K_2(u_0,\delta)\Delta u_1 + \frac{1}{2} \eta \Delta u_1,
\endaligned
$$ 
where $\eta$ denotes the common modulus of continuity of the family $\{\del_r h_\eps\}$ with respect to $r$. 

Let now $\delta\leq r_1\leq r_1^\prime$. Given $u_1\leq u_0$, we see that the characteristics $\chi_\eps(u; u_1,r_1)$ and $\chi_\eps(u; u_1,r_1')$ are contained in $\Ocal_\delta$ for $0\leq u\leq u_1$. Define
$$
\psi_\eps(u) := \del_r h_\eps(u,\chi_\eps(u; u_1, r_1)) - \del_r h_\eps(u,\chi_\eps(u; u_1, r_1')).
$$
It then follows from \eqref{derivative eq} that
\bel{del_u psi}
\aligned
\del_u\psi_\eps
& = \big(\del_r\gb_\eps \big)_{(u,\chi_\eps(u; u_1,r_1))} \psi_\eps(u)
+ \Big(\big(\del_r\gb_\eps \big)_{(u,\chi_\eps(u; u_1,r_1))} - \big(\del_r\gb_\eps \big)_{(u,\chi_\eps(u; u_1, r_1^\prime))} \Big)\big(\del_r h_\eps \big)_{(u,\chi_\eps(u; u_1,r_1^\prime))} 
\\
& \quad 
+ f_\eps(u,\chi_\eps(u; u_1, r_1)) -f_\eps(u,\chi_\eps(u; u_1, r_1')),
\endaligned
\ee
where
$$
f_\eps:= \frac{h_\eps- \hb_\eps}{2(r+ \eps)} \del_rg_\eps- \frac{3}{2(r+ \eps)^2}(g_\eps- \gb_\eps)(h_\eps- \hb_\eps) - \frac{c^2}{2}e^{\nu_\eps + \lambda_\eps} \big(h_\eps+4\pi \hb_\eps(h_\eps- \hb_\eps)^2\big).
$$
By straightforward calculations, we find  
$$
\aligned
\del_{rr} \gb_\eps& = \frac{\del_rg_\eps}{r+ \eps} -2\frac{g_\eps- \gb_\eps}{(r+ \eps)^2},
\\
\del_{rr}g_\eps& =  4 \pi \, \del_rg_\eps \frac{(h_\eps- \hb_\eps)^2}{r+ \eps} -12\pi \frac{(h_\eps- \hb_\eps)^2}{(r+ \eps)^2}g_\eps +8\pi \frac{h_\eps- \hb_\eps}{r+ \eps}g_\eps \del_r h_\eps
\\
& \quad-8\pi c^2 e^{\nu_\eps + \lambda_\eps} \big( h_\eps^2+(r+ \eps) \, \hb_\eps \del_r h_\eps +4\hb_\eps h_\eps(h_\eps- \hb_\eps)^2\big),
\\
\del_rf_\eps& = - \frac{3}{2(r+ \eps)^2}(g_\eps- \gb_\eps)\del_r h_\eps+ \Big(\frac{5}{2(r+ \eps)^2}(h_\eps- \hb_\eps)+ \frac{\del_r h_\eps}{2(r+ \eps)} \Big)\del_rg_\eps 
\\
& \quad + \frac{h_\eps- \hb_\eps}{2(r+ \eps)} \del_{rr}g_\eps + \frac{6}{(r+ \eps)^3}(g_\eps- \gb_\eps)(h_\eps- \hb_\eps)
\\
& \quad- \frac{c^2}{2}e^{\nu_\eps + \lambda_\eps} \Big(\del_r h_\eps+4\pi \, \frac{(h_\eps- \hb_\eps)^3}{r+ \eps} \Big) -2\pi c^2 e^{\nu_\eps + \lambda_\eps} \Big(h_\eps+4\pi \, \hb_\eps(h_\eps- \hb_\eps)^2\Big)\frac{(h_\eps- \hb_\eps)^2}{r+ \eps}.
\endaligned
$$
Therefore, by Lemma \ref{lemma-uniform bound of h and delr h with e} we have 
$
|\del_{rr} \gb_\eps|,|\del_rf_\eps| \leq K_3(u_0,\delta)
$
for some constant $K_3(u_0,\delta) > 0$ independent of $\eps$, hence
\bel{equicontinuous 1}
\aligned
\Big|\big(\del_r\gb_\eps \big)_{(u,\chi_\eps(u; u_1, r_1))} - \big(\del_r\gb_\eps \big)_{(u,\chi_\eps(u; u_1, r_1^\prime))} \Big|& \leq K_3(u_0,\delta)\big(\chi_\eps(u; u_1,r_1) - \chi_\eps(u; u_1,r_1^\prime)\big),
\\
\Big|f_\eps(u,\chi_\eps(u; u_1, r_1)) -f_\eps(u,\chi_\eps(u; u_1, r_1'))\Big|& \leq  K_3(u_0,\delta)\big(\chi_\eps(u; u_1, r_1) - \chi_\eps(u; u_1, r_1^\prime)\big).
\endaligned
\ee
Combining this result with \eqref{g-gb}, in view of \eqref{del_u psi} we obtain 
\bel{estimate del_u psi}
|\del_u\psi_\eps| \leq \frac{K_4(u_0,\delta)}{\delta}|\psi_\eps|+K_5(u_0,\delta)\big(\chi_\eps(u; u_1,r_1) - \chi_\eps(u; u_1,r_1^\prime)\big),
\ee
where 
$$
\aligned
K_4(u_0,\delta):& =  \max\Big(1,2(1+ 2\pi c^2\Acal^2_\eps(u_0,\delta))\Big),
\qquad
K_5(u_0,\delta): =  K_3(u_0,\delta)\big(1+B_5(u_0,\delta)\big).
\endaligned
$$
Therefore, integrating \eqref{estimate del_u psi} yields
\bel{equi.continuity w.r.t. r}
|\psi_\eps(u_1)| \leq e^{\frac{K_4(u_0,\delta)u_1}{\delta}} \Big(\psi_\eps(0)+K_5(u_0,\delta) \int_0^{u_1} \big(\chi_\eps(u; u_1, r_1) - \chi_\eps(u; u_1, r_1^\prime)\big)  du\Big).
\ee
Now, since both characteristics $\chi_\eps(u; u_1, r_1)$ and $\chi_\eps(u; u_1, r_1^\prime)$ are contained in $\Ocal_\delta$, the following inequality holds (thanks to \eqref{g-gb} and the mean value theorem): 
$$
\aligned
\chi_\eps(u; u_1,r_1) - \chi_\eps(u; u_1, r^\prime_1) 
& = (r_1-r_1^\prime) e^{\frac{1}{2} \int_u^{u_1} \big(\frac{g_\eps- \gb_\eps}{r+ \eps} \big)_{(u',\chi_\eps(u^\prime; u_1,s))} \, du^\prime} 
  \leq (r_1-r^\prime_1) e^{\frac{(u_1-u) M_\eps(0)}{\delta^2}}
\endaligned
$$
for some $s\in[r_1^\prime,r_1]$. Hence, from \eqref{equi.continuity w.r.t. r} we have 
$$
|\psi_\eps(u_1)| \leq e^{\frac{K_4(u_0,\delta)u_0}{\delta}} \Big(\eta_0+K_5(u_0,\delta)u_0e^{\frac{u_0  M_\eps(0)}{\delta^2}} \Big)(r_1-r_1^\prime),
$$
where $\eta_0$ is a modulus of continuity of $\del_r h_0$.
\end{proof}


\begin{proof}[Proof of Lemma~\ref{equi-continuity of hb}]
We have 
$
|\del_r \, \hb_\eps| = \bigg|\frac{h_\eps- \hb_\eps}{r+ \eps} \bigg| \leq K_6(u_0,\delta),
$
 for all $(u,r) \in \Ocal_\delta$, in which 
$K_6(u_0,\delta) := \frac{B_5\delta+A(u_0,\delta)}{\delta^2}$. Consequently,  
 $\{\hb_\eps\}$ is equicontinuous with respect to $r$ in $\Ocal_\delta$. 
We now show equi-continuity with respect to $u$. In fact, given $u_1,r_1$ and $\Delta u_1> 0$, let $r_1+ \Delta r_{1,\eps}$ be the value of $r$ at which the characteristic $\chi_\eps(u; u_1,r_1)$ intersects the line $u=u_1- \Delta u_1$. Observe that $\Delta r_{1,\eps} \leq 1/2\Delta u_1$. Then, we have
$$
\aligned
&
\big|\hb_\eps(u_1,r_1) - \hb_\eps(u_1- \Delta u_1,r_1)\big|
\\
& \leq  \Big|\int_{u_1- \Delta u_1}^{u_1} \Big( D_\eps \hb_\eps \Big)_{\chi_{\eps,u_1}} \, du\Big|+ \big|\hb_\eps(u_1- \Delta u_1,r_1+ \Delta r_{1,\eps}) - \hb_\eps(u_1- \Delta u_1,r_1)\big|
\\
& \leq  \frac{1}{\delta} \Big|\int_{u_1- \Delta u_1}^{u_1} \Big( (r+ \eps)D_\eps \hb_\eps \Big)_{\chi_{\eps,u_1}} \, du\Big|+K_6(u_0,\delta)\Delta r_{1,\eps} 
\\
& \leq  \frac{\sqrt{\Delta u_1}}{\delta} \Big(\int_{u_1- \Delta u_1}^{u_1} \Big( (r+ \eps)^2(D_\eps \hb_\eps)^2 \Big)_{\chi_{\eps,u_1}} \Big)^{1/2} + \frac{K_6(u_0,\delta)}{2} \Delta u_1
\leq \frac{1}{2\delta} \sqrt{\frac{M_\eps(0)\Delta u_1}{\pi}} + \frac{K_6(u_0,\delta)}{2} \Delta u_1,
\endaligned
$$
which implies the equi-continuity with respect to $u$ as claimed.
\end{proof}
 
\end{document}